\documentclass[twoside,11pt]{article}

\usepackage{hyperref}
\usepackage{url}
\usepackage{amsfonts,amsmath,amscd,dsfont,mathrsfs}
\usepackage{float,psfrag}
\usepackage{wrapfig}
\usepackage{caption}
\usepackage{subcaption}
\usepackage{color}
\usepackage{etoolbox}
\usepackage{jmlr2e}

\appto{\equation}{%
  \notag
  \preto{\label}{%
    \refstepcounter{equation}
    \tag{\arabic{equation}}%
  }
}

\DeclareMathAlphabet{\mathpzc}{OT1}{pzc}{m}{it}

\newtheorem{assumption}{Assumption}
\newtheorem{claim}{Claim}

\newcommand*{\unity}{\textrm{{\usefont{U}{fplmbb}{m}{n}1}}}

\newcommand{\X}{X}
\newcommand{\x}{x}
\newcommand{\Y}{Y}
\newcommand{\y}{y}
\newcommand{\kk}{k}
\newcommand{\n}{\ell}
\newcommand{\D}{D}
\newcommand{\Q}{Q}
\newcommand{\q}{q}
\newcommand{\s}{S}

\newcommand{\T}{T}

\newcommand{\PP}{P}
\newcommand{\pp}{P}
\newcommand{\MM}{M}
\newcommand{\mm}{M}

\newcommand{\Po}{P_0}
\newcommand{\po}{P_0}
\newcommand{\Mo}{M_0}
\newcommand{\mo}{M_0}

\newcommand{\Pt}{P_1}
\newcommand{\pt}{P_1}
\newcommand{\Mt}{M_1}
\newcommand{\mt}{M_1}
\newcommand{\B}{\Theta}
\newcommand{\bb}{\theta}
\newcommand{\dd}{\alpha}

\newcommand{\reals}{{\mathbb R}}
\newcommand{\cX}{{\cal \X}}
\newcommand{\cY}{{\cal \Y}}
\newcommand{\cS}{{\cal \s}}
\newcommand{\Dkl}{D_{\rm kl}}
\newcommand{\prob}{\mathbb P}

\jmlrheading{1}{2000}{1-48}{4/00}{10/00}{Peter Kairouz and Sewoong Oh and Pramod Viswanath}

\firstpageno{1}

\begin{document}

\title{Extremal Mechanisms for Local Differential Privacy}


\author{\name Peter Kairouz \email kairouz2@illinois.edu \\
       \addr Department of Electrical and Computer Engineering\\
        University of Illinois at Urbana-Champaign\\
       Urbana, IL 61801, USA
       \AND
       \name Sewoong Oh \email swoh@illinois.edu \\
       \addr Department of Industrial and Enterprise Systems Engineering\\
       University of Illinois at Urbana-Champaign\\
       Urbana, IL 61820, USA
       \AND
       \name Pramod Viswanath \email pramodv@illinois.edu \\
       \addr Department of Electrical and Computer Engineering\\
        University of Illinois at Urbana-Champaign\\
       Urbana, IL 61801, USA}

\editor{Mehryar Mohri}

\maketitle

\begin{abstract}%
Local differential privacy has recently surfaced as a strong measure of privacy in contexts where personal information remains private even from data analysts. Working in a setting where both the data providers and data analysts want to maximize the utility of statistical analyses performed on the released data, we study the fundamental trade-off between local differential privacy and utility. This trade-off is formulated as a constrained optimization problem: maximize utility subject to local differential privacy constraints. We introduce a combinatorial family of extremal privatization mechanisms, which we call staircase mechanisms, and show that it contains the optimal privatization mechanisms for a broad class of information theoretic utilities such as mutual information and $f$-divergences. We further prove that for any utility function and any privacy level, solving the privacy-utility maximization problem is equivalent to solving a finite-dimensional linear program, the outcome of which is the optimal staircase mechanism. However, solving this linear program can be computationally expensive since it has a number of variables that is exponential in the size of the alphabet the data lives in. To account for this, we show that two simple privatization mechanisms, the binary and randomized response mechanisms, are universally optimal in the low and high privacy regimes, and well approximate the intermediate regime.

\end{abstract}

%
%

\ShortHeadings{Extremal Mechanisms for Local Differential Privacy}{ Kairouz,  Oh and  Viswanath}

\section{Introduction}
\label{sec:intro}

In statistical analyses involving data from individuals, there is an
increasing tension between the need to share data and the need to protect
sensitive information about the individuals. For example, users of social networking sites are
increasingly cautious about their privacy, but still find it inevitable to
agree to share their personal information in order to benefit from customized services such as
recommendations and personalized search \citep{acquisti1,acquisti2}. There is a certain
utility in sharing data for both data providers and data analysts, but at the
same time, individuals want {\em plausible deniability} when it comes to
sensitive information.

For such applications, there is a natural core optimization problem to be solved.
Assuming both the data providers and analysts want to maximize the
utility of the released data, how can they do so while preserving the
privacy of participating individuals? The formulation and study of a framework that addresses the fundamental tradeoff between utility and privacy is the focus of this paper.

\subsection{Local differential privacy}
The need for data privacy appears in two different contexts: the {\em local privacy}
context, as in when individuals disclose their personal information (e.g., voluntarily on social
network sites), and the {\em global privacy} context, as in when institutions
release databases of information of several people or answer queries on such databases
(e.g., US Government releases census
data, companies like Netflix release proprietary data for others to test state of the art
 machine learning algorithms). In both contexts, privacy is achieved by {\em randomizing} the data before releasing it. We study the setting of local privacy, in which data providers do not trust the data collector (analyst). Local privacy dates back to \cite{War65}, who proposed the \textit{randomized response} method to provide plausible deniability for individuals responding to sensitive surveys.

A natural notion of privacy protection is making inference of information beyond what is released  hard.
{\em Differential privacy} has been proposed in the global privacy context
to formally capture this notion of privacy \citep{Dwo06,DMNS06,DL09}. In a nutshell, differential privacy ensures that
an adversary
should not be able to reliably infer an individual's record in a database,
even with unbounded computational power and access to every other record in the database. Recently, \cite{DJW13} extended the notion of differential privacy to the local privacy context. Formally, consider a setting where there are $n$ data providers each owning a data $\X_i$ defined on an input alphabet $\cX$. The $\X_i$'s are independently sampled from some distribution $\PP_\nu$ parameterized by $\nu$. A statistical privatization mechanism $\Q$ is a conditional distribution that maps $\X_i \in \cX$ stochastically to $\Y_i\in \cY$, where $\cY$ is an output alphabet possibly larger than $\cX$. The $\Y_i$'s are referred to as the privatized (sanitized) views of $\X_i$'s. In a non-interactive setting, the same privatization mechanism $\Q$ is used locally by all individuals. This setting is shown in Figure \ref{fig:csm} for the special case of $n=2$. For some non-negative $\varepsilon$, we follow the definition of \cite{DJW13}
and say that a mechanism $\Q$ is
{\em $\varepsilon$-locally differentially private} if
\begin{eqnarray}
	\sup_{S\subset \cY, x,x' \in\cX} \frac{\Q(S|\x)}{\Q(S|\x')} \;\leq\; e^\varepsilon \;,
	\label{eq:deflocalDP1}
\end{eqnarray}
where $\Q(S|x) = \prob(Y_i\in S | X_i=x)$ represents the privatization mechanism.
This ensures that for small values of $\varepsilon$, given a privatized data $Y_i$, it is (almost) equally likely to have come from
any data, i.e. $x$ or $x'$. A small value of $\varepsilon$ means that we require a high level of privacy and a large value corresponds to a low level of privacy. At one extreme, for
$\varepsilon=0$, the privatized output must be independent of the private data, and on the other extreme, for $\varepsilon=\infty$, the privatized output can be made equal to the private data.

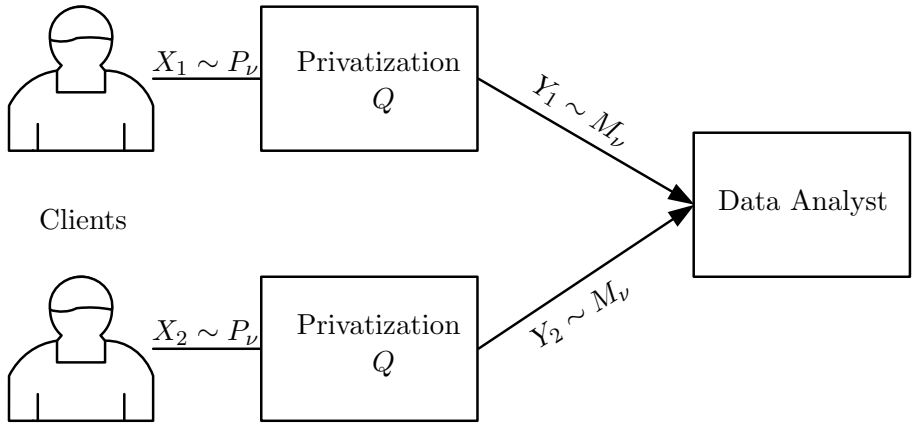
\begin{figure}

\setlength{\unitlength}{5cm}
~~~~~~~~ \begin{picture}(1,1)
     \put(0,0){\includegraphics[width=0.8\textwidth]{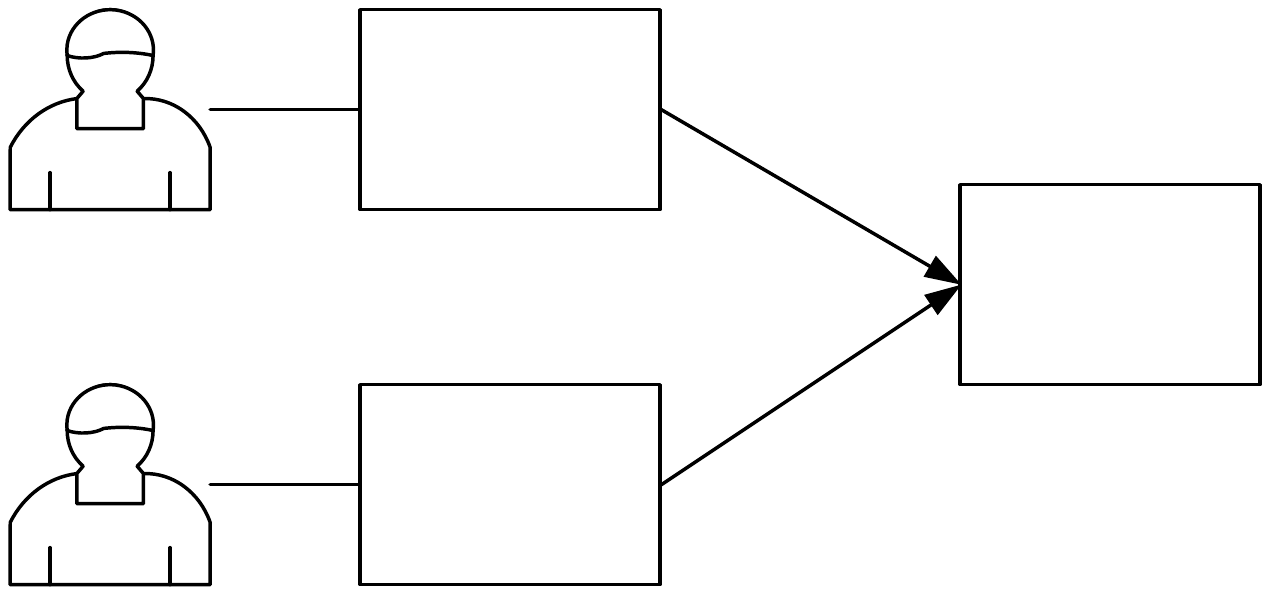}}
    \put(0.4,0.95){{$\X_{1} \sim \PP_{\nu}$}}
     \put(0.4,0.233){{$\X_{2} \sim \PP_{\nu}$}}
     \put(0.1,0.53){Clients}
     \put(0.785,0.95){Privatization}
     \put(0.985,0.85){$\Q$}
     \put(0.785,0.25){Privatization}
     \put(0.985,0.15){$\Q$}
    \put(1.4,0.9){\rotatebox{-32}{{$\Y_{1} \sim \MM_{\nu}$}}}
     \put(1.4,0.21){\rotatebox{32}{{$\Y_{2} \sim \MM_{\nu}$}}}
     \put(1.905,0.58){Data Analyst}
  \end{picture}%

\caption{Client $i$ owns $X_i$ sampled from $\PP_\nu$. Each $\X_i$ is privatized by the same $\varepsilon$-locally differentially private mechanism $\Q$. The data analyst only observes the privatized data ($Y_i$'s) and makes an inference on the statistics of the original distribution of the data.}
\label{fig:csm}
\end{figure}

\subsection{Information theoretic utilities for statistical analysis}
In  analyses of statistical databases, the analyst is interested in the {\em statistics} of the data as opposed to individual records.
Naturally, the utility should also be measured
in terms of the distribution rather than sample quantities.
Concretely, consider a client-server setting where each client with data $\X_i$
releases $\Y_i$, a privatized version of the data, via a non-interactive
$\varepsilon$-locally differentially private privatization mechanism $\Q$.
Assume all the clients use the same privatization mechanism $Q$,
and each client's data is an i.i.d.\ sample from a distribution $P_\nu$ for some parameter  $\nu$.
Given the privatized views $\{\Y_i\}_{i=1}^{n}$, the data analyst would like to make inferences based on the induced marginal distribution
\begin{eqnarray}
	\MM_\nu(S) &\equiv& \sum_{\x\in\cX} \Q(S|\x)  \PP_\nu(\x) \;,
	\label{eq:defM2}
\end{eqnarray}
for $S \subseteq \cY$.
We consider a broad class of convex utility functions,
and identify the class of optimal mechanisms, which we call {\em staircase mechanisms}, in Section \ref{sec:result}.
We apply this framework to two specific applications: (a) hypothesis testing where the utility is measured in $f$-divergences (Section \ref{sec:hyp}) and (b) information preservation where the utility is measured in mutual information (Section \ref{sec:est}).

In the binary hypothesis testing setting, $\nu \in \{0,1\}$; therefore, $\X$ can be generated by one of two possible distributions $\PP_0$ and $\PP_1$. The power to discriminate data generated from $\PP_0$ to data generated from $\PP_1$
depends on the `distance' between the marginals $\MM_0$ and $\MM_1$.
To measure the ability of such statistical discrimination,
our choice of utility of a particular privatization mechanism $\Q$ is an information theoretic quantity called Csisz\'ar's $f$-divergence defined as
\begin{eqnarray}
	D_f(M_0||M_1) = \sum_{\x\in\cX} f\Big(\,\frac{M_0(x)}{M_1(x)}\,\Big) \,M_1(x)\;, \label{eq:defdiv}
\end{eqnarray}
for some convex function $f$ such that $f(1)=0$.
The Kullback-Leibler (KL) divergence $\Dkl(\MM_0||\MM_1)$ is a special case with
$f(x)=x\log x$, and so is the total variation distance $\|\MM_0-\MM_1\|_{\rm TV}$ with $f(x)=(1/2)|x-1|$.
Such $f$-divergences capture the quality of statistical inference, such as
minimax rates of statistical estimation or error exponents in hypothesis testing \citep{TZ09,CT12}.
As a motivating example, suppose a data analyst wants to
test whether the data is generated from $\PP_0$ or $\PP_1$
based on privatized views $Y_1,\ldots,Y_n$.
According to Chernoff-Stein's lemma,
for a bounded type I error probability, the best type II error probability
scales as $e^{-n\, \Dkl(\MM_0||\MM_1)}$.
Naturally, we are interested in finding a privatization mechanism $\Q$
that minimizes the probability of error by solving the following constraint maximization problem
\begin{equation}
\label{eq:optKL}
\begin{aligned}
& \underset{\Q}{\text{maximize}}
& & \Dkl(\MM_0||\MM_1)  \\
& \text{subject to}
& & \Q \in \mathcal{\D}_{\varepsilon}
\end{aligned},
\end{equation}
where $\mathcal{\D}_{\varepsilon}$ is the set of all $\varepsilon$-locally differentially private mechanisms
 satisfying (\ref{eq:deflocalDP1}).

In the information preservation setting, $\X$ is generated from an underlying distribution $\PP$.
We are interested in quantifying how much information can be preserved when releasing a private view of the data. In other words, the data provider would like to release an $\varepsilon$-locally differentially private view $\Y$ of $\X$ that preserves the amount of information in $\X$ as much as possible.
The utility in this case is measured by the mutual information between $\X$ and $\Y$
\begin{equation}
I\left(\X;\Y\right)=\sum_{\cX}\sum_{ \cY} \pp\left(\x\right)\Q\left(\y|\x\right)\log\left(\frac{\Q\left(\y|\x\right)}{\sum_{l \in \cX}\pp\left(l\right)\Q\left(\y|l\right)}\right).
\label{eq:defmi}
\end{equation}
Mutual information, as the name suggests, measures the mutual dependence between two random variables.  It has been used as a criterion for feature selection and as a measure of similarity between two different clusterings of a data set, in addition to many other applications in signal processing and machine learning. To characterize the fundamental tradeoff between privacy and mutual information, we solve the following constrained maximization problem
\begin{equation}
\label{eq:optest}
\begin{aligned}
& \underset{\Q}{\text{maximize}}
& & I(\X;\Y)  \\
& \text{subject to}
& & \Q \in \mathcal{\D}_{\varepsilon}
\end{aligned},
\end{equation}
where $\mathcal{\D}_{\varepsilon}$ is the set of all $\varepsilon$-locally differentially private mechanisms
 satisfying (\ref{eq:deflocalDP1}).

Motivated by such applications in statistical analysis,
our goal is to provide a general framework for finding
optimal privatization mechanisms
that maximize information theoretic utilities
under local differential privacy.  We demonstrate the power of our techniques
in a very general setting that includes both hypothesis testing and information preservation.

\subsection{Our contributions}
We study the fundamental tradeoff between local differential privacy and a rich class of convex utility functions. This class of utilities includes several information theoretic quantities such as mutual information and $f$-divergences. The privacy-utility tradeoff is posed as a constrained maximization problem: maximize utility subject to local differential privacy constraints. This maximization problem is (a) nonlinear: the utility functions we consider are convex in $\Q$; (b) non-standard: we are maximizing instead of minimizing a convex function; and (c) infinite dimensional: the space of all differentially private mechanisms is infinite dimensional. We show, in Theorem \ref{thm:sc}, that for all utility functions considered and any privacy level $\varepsilon$, a {\em finite} family of {\em extremal} mechanisms (a finite subset of the corner points of the space of differentially private mechanisms), which we call {\em staircase} mechanisms, contains the optimal privatization mechanism. We further prove, in Theorem \ref{thm:lp}, that solving the original privacy-utility problem is equivalent to solving a finite dimensional linear program, the outcome of which is the optimal mechanism. However, solving this linear program can be computationally expensive because it has $2^{|\cX|}$ variables. To account for this, we show that two simple staircase mechanisms (the binary and randomized response mechanisms) are optimal in the high and low privacy regimes, respectively, and well approximate the intermediate regime.  This contributes an important progress in the differential privacy area, where the privatization mechanisms have been few and almost no exact optimality results are known. As an application, we show that the effective sample size reduces from $n$  to $\varepsilon^2 n$ under local differential privacy in the context of hypothesis testing.

We also study the fundamental tradeoff between utility and approximate differential privacy, a generalized notion of privacy that was first introduced in \cite{DKM06}. The techniques we develop for differential privacy do not generalize to approximate differential privacy. To account for this, we use an operational interpretation of approximate differential privacy (developed in \cite{KOV14}) to prove that a simple mechanism maximizes utility for all levels of privacy when the data is binary.

\subsection{Related work}

Our work is closely related to the recent work of \cite{DJW13}
where an upper bound on $\Dkl(\MM_0||\MM_1)$ was derived under the same local differential privacy setting.
Precisely, Duchi et.\ al.\ proved that the KL-divergence maximization problem
in \eqref{eq:optKL} is at most $4(e^\varepsilon-1)^2\|P_1-P_2\|_{TV}^2$.
This bound was further used to provide a minimax bound on statistical estimation
using information theoretic converse techniques such as Fano's and Le Cam's inequalities.
Such tradeoffs also provide tools for comparing various notions of privacy \citep{BD14}.

In a similar spirit, we are also interested in
maximizing information theoretic quantities under local differential privacy.
We generalize the results of \cite{DJW13}, and provide stronger results in the sense that we
$(a)$ consider a broader class of information theoretic utilities;
$(b)$ provide explicit constructions for the optimal mechanisms; and
$(c)$ recover the existing result of \cite[Theorem 1]{DJW13} (with a stronger condition on $\varepsilon$).

Our work also provides a formal connection to an information theoretic notion of privacy called information leakage \citep{CCG10,SRP13}. Given a privatization mechanism $\Q$, the information leakage is measured
  by the mutual information between the private data $\X$ and the released output $\Y$, i.e. $I(\X;\Y)$.
  Information leakage has been widely studied as a practical notion of privacy.
    However, connections to differential privacy have been studied only indirectly through comparisons to how much distortion
    is incurred under the two notions of privacy \citep{SS14, WYZ14}. We show that under $\varepsilon$-local differential privacy, $I(X;Y)$ is upper bounded by $0.5 \varepsilon^2 \max_{A\subseteq \cX} P(A) P(A^c) + O(\varepsilon^3)$ for small $\varepsilon$. Moreover, we provide an explicit privatization mechanism that achieves this bound.

While there is a vast literature on differential privacy,
exact optimality results are only known for a few cases.
The typical recipe is to propose a differentially private mechanism inspired by the work of \cite{Dwo06,DMNS06,MT07} and \cite{HR10},
and then establish its near-optimality by comparing the achievable utility to a converse,
for example in linear dynamical systems \citep{YZMD14}, principal component analysis \citep{CSS12,BBDS12,HR12,KT13}, linear queries \citep{HT10,HLM12}, logistic regression \citep{CM08}
and histogram release \citep{Lei11}.
In this paper, we take a different route and solve the utility maximization problem {\em exactly}.

Optimal differentially private mechanisms are known only in a few cases.
\cite{GRS12} showed that the
geometric noise adding mechanism is optimal (under a Bayesian setting)
for monotone utility functions under count queries (sensitivity one).
This was generalized by Geng et.\ al.\ (for a worst-case input setting)  who proposed a family of mechanisms
and proved its optimality for monotone utility functions under queries with arbitrary sensitivity \citep{GV12,GV13,GV13multi}.
The family of optimal mechanisms was called {\em staircase mechanisms} because for any $\y$ and any neighboring $\x$ and $\x'$,
the ratio of $\Q(\y|\x)$ to $\Q(\y|\x')$ takes one of three possible values $e^\varepsilon$, $e^{-\varepsilon}$, or $1$. Since the optimal mechanisms we develop also have an identical property, we retain the same nomenclature.

\subsection{Organization}

The remainder of the paper is organized as follows. In Section \ref{sec:result}, we introduce the family of staircase mechanisms, prove its optimality for a broad class of convex utility functions, and study its combinatorial structure. In Section \ref{sec:hyp}, we study the problem of private hypothesis testing and prove that two staircase mechanisms, the binary and randomized response mechanisms, are optimal for KL-divergence in the high and low privacy regimes, respectively, and (nearly) optimal the intermediate regime. We show, in Section \ref{sec:est}, similar results for mutual information. In Section \ref{sec:general}, we study approximate local differential privacy, a more general notion of local privacy. Finally, we conclude this paper in Section \ref{sec:discussion} with a few interesting and nontrivial extensions.

\section{Main Results}
\label{sec:result}

In this section, we first present a formal definition for staircase mechanisms
and prove that they are the optimal solutions to optimization
problems of the form \eqref{eq:opt}.
We then provide a combinatorial representation for staircase mechanisms that allows us to reduce the infinite dimensional nonlinear program of \eqref{eq:opt} to
a finite dimensional linear program with $2^{|\cX|}$ variables.
For any given privacy level $\varepsilon$ and utility function $U(\cdot)$, one can solve this linear program to obtain the optimal privatization mechanism, albeit with significant computational challenges
since the number of variables scales exponentially in the alphabet size. To address this issue, we prove, in Sections \ref{sec:hyp} and \ref{sec:est}, that two simple staircase mechanisms, which we call the binary mechanism and the randomized response mechanism, are optimal in the high and low privacy regimes, respectively, and well approximate the intermediate regime.
\subsection{Optimality of staircase mechanisms}
For an input alphabet $\cX$ with $|\cX|=\kk$, we represent the set of
$\varepsilon$-locally differentially private mechanisms that lead to output
alphabets $\cY$ with $|\cY|=\n$ by
\begin{equation*}
\mathcal{\D}_{\varepsilon,\n}=\mathcal{\Q}_{\kk\times\n}\cap\left\{\Q\;:\;\forall \; \x, \x'\in\mathcal{\X}, S \subseteq \mathcal{\Y},   \;\Big|\,\ln\frac{\Q\left(S|x\right)}{\Q\left(S|x'\right)}\,\Big| \leq\varepsilon\right\},
\end{equation*}
where $\mathcal{\Q}_{\kk\times\n}$ denotes the set of all $\kk \times \n$ dimensional
conditional distributions. The set of all $\varepsilon$-locally differentially private mechanisms is given by
\begin{equation}
\label{eq:defdp}
\mathcal{\D}_{\varepsilon}= \cup_{\n \in \mathbb{N}} \mathcal{\D}_{\varepsilon,\n}.
\end{equation}
The set of all conditional distributions acting on $\cX$ is given by $\mathcal{\Q}=\cup_{\n \in \mathbb{N}} \mathcal{\Q}_{\kk,\n}$.

We consider two
types of utility functions, one for the hypothesis testing setup and another for the information preservation setup. In the hypothesis testing setup, the utility is a
function of the privatization mechanism and two priors defined on the input alphabet.
Namely, $U\left(\Po, \Pt, \Q\right): \mathbb{S}^{\kk} \times\mathbb{S}^{\kk}
\times \mathcal{\Q} \rightarrow \mathbb{R}_{+}$, where $\Po$ and $\Pt$ are
positive priors defined on $\cX$, and $\mathbb{S}^{\kk}$ is the $(\kk-1)$-dimensional probability simplex. $\PP_{\nu}$ is said to be positive if
$\pp_{\nu}\left(\x\right) >0$ for all $\x \in \cX$. In the information preservation setup, the utility is a function of the
privatization mechanism and a prior defined on the input alphabet. Namely,
$U\left(\PP, \Q\right): \mathbb{S}^{\kk} \times \mathcal{\Q} \rightarrow
\mathbb{R}_{+}$, where $\PP$ is a positive prior defined on $\mathcal{\X}$. For notational convenience,
we will use $U\left(\Q\right)$ to refer to both $U\left(\PP, \Q\right)$ and
$U\left(\po, \pt, \Q\right)$.

\begin{definition}[Sublinear Functions] A function $\mu\left(z\right): \mathbb{R}^{\kk} \rightarrow \mathbb{R}$ is said to be sublinear if the following two conditions are met
\begin{enumerate}
\item $\mu\left(\gamma z\right)=\gamma \mu\left(z\right)$ for all $\gamma
    \in \mathbb{R}_{+}$.
\item $\mu\left(z_1 + z_2\right) \leq \mu\left(z_1\right) +
    \mu\left(z_2\right)$ for all $z_1,z_2 \in \mathbb{R}^{\kk}$.
\end{enumerate}
\end{definition}

Let $\Q_\y$ be the column of $\Q$ corresponding to $\Q(\y|\cdot)$ and $\mu$ be any sublinear function.
We are interested in utilities that can be decomposed into a sum of
sublinear functions.
We study the \emph{fundamental tradeoff between privacy and utility} by
solving the following constrained maximization problem
\begin{equation}
\label{eq:opt}
\begin{aligned}
& \underset{\Q}{\text{maximize}}
& & U\left(\Q\right)=\sum_{\y \in  \cY} \mu(\Q_\y) \\
& \text{subject to}
& & \Q \in \mathcal{\D}_{\varepsilon}
\end{aligned}.
\end{equation}
This includes maximization over information theoretic quantities
of interest in statistical estimation and hypothesis testing such as mutual information,
total variation, KL-divergence, and $\chi^2$-divergence \citep{TZ09}.
Since sub-linearity implies convexity,
\eqref{eq:opt} is in general a complicated nonlinear program:
we are maximizing (instead of minimizing) a convex function in $Q$; further,
the dimension of $Q$ might be unbounded: the optimal privatization mechanism $\Q^*$ might produce an infinite output alphabet $\cY$. The following theorem proves that
one never needs an output alphabet larger than
the input alphabet in order to achieve the maximum utility,
and provides a combinatorial representation for the optimal solution.
\begin{theorem}
\label{thm:sc}
	For any sublinear function $\mu$ and any $\varepsilon\geq0$,
	there exists an optimal mechanism $\Q^*$
	maximizing the utility in \eqref{eq:opt}
	over all $\varepsilon$-locally differentially private mechanisms,
	such that
	\begin{itemize}
	\item[$(a)$] the output alphabet size is at most the input alphabet size, i.e. $|\cY|\leq|\cX|$; and
	\item[$(b)$] 	for all  $\y\in\cY$, and $\x,\x'\in\cX$
\begin{eqnarray}
	\Big| \ln \frac{\Q^*(\y|\x)}{\Q^*(\y|\x')} \Big| \in \{0, \varepsilon\} \;.
	\label{eq:defsc}
	\end{eqnarray}
	\end{itemize}
\end{theorem}
The first claim of bounded alphabet size is more generally true for
any general utility $U\left(\Q\right)$ that is convex in $\Q$ (not necessarily
decomposing into a sum of sublinear functions as in \eqref{eq:opt}).
The second claim establishes that
there is an optimal mechanism with an extremal structure;
 the absolute value of the log-likelihood ratios can only take one of the two extremal values: 0 or $e^\varepsilon$  (see Figure \ref{fig:mechanism} for example).
 We refer to such a mechanism as a staircase mechanism,
and define the {\em family of staircase mechanisms} formally as
\begin{eqnarray*}
	\cS_\varepsilon \;\equiv\; \{Q \,|\, \text{ satisfying }\eqref{eq:defsc}\}\;.
\end{eqnarray*}
For all choices of $U\left(\Q\right) = \sum_{\cY} \mu(\Q_\y)$ and any $\varepsilon \geq 0$, Theorem \ref{thm:sc} implies that the family of staircase mechanisms includes the optimal solutions to maximization problems of the form \eqref{eq:opt}.
Notice that  staircase mechanisms are $\varepsilon$-locally differentially private,
since any $\Q$ satisfying \eqref{eq:defsc} implies that $\Q(y|x)/\Q(y|x')  \leq e^\varepsilon$.

\begin{figure}[h]
	\begin{center}
	\includegraphics[width=.32\textwidth]{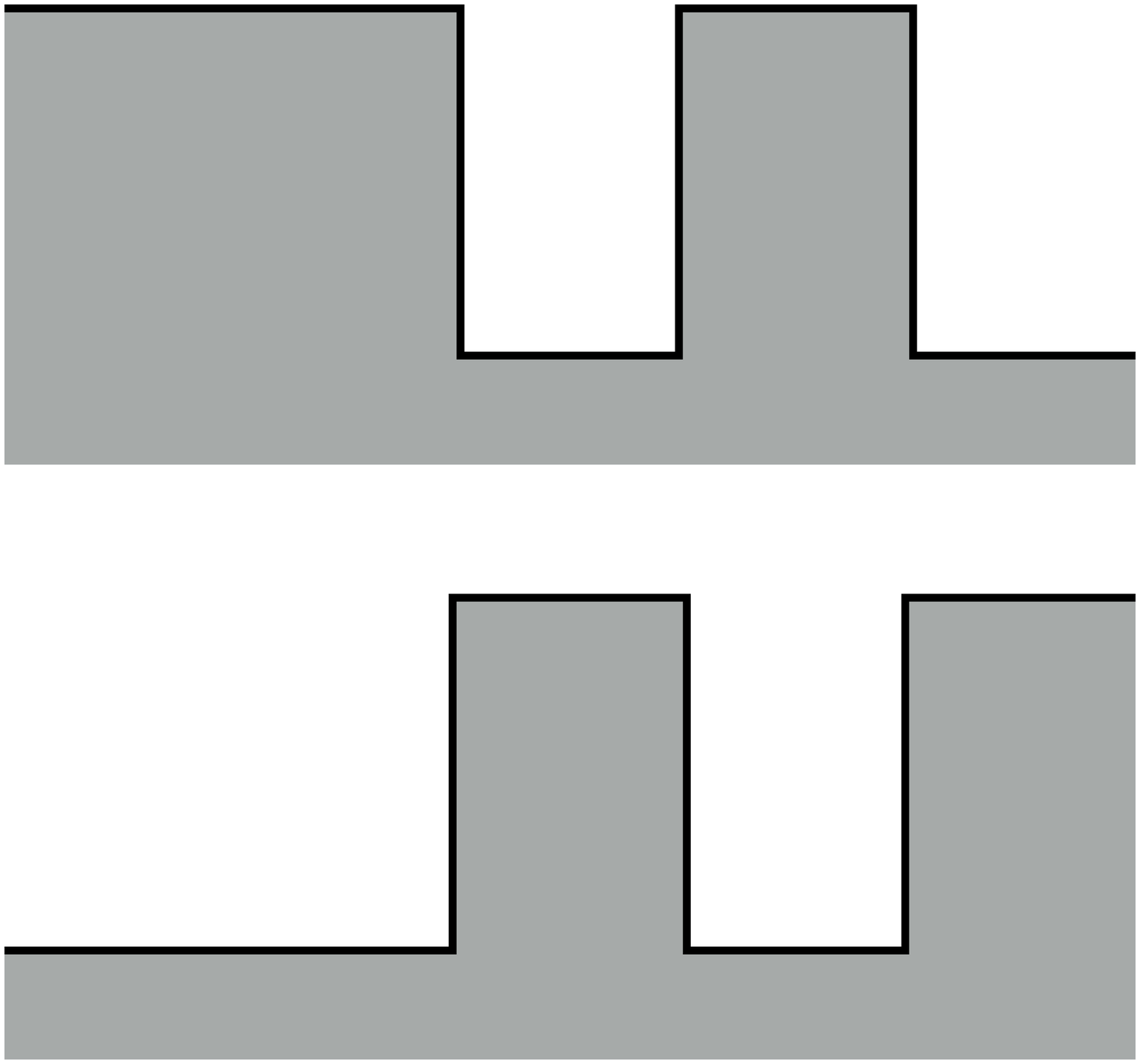}
	\put(-149,75){\small$y=1$}
	\put(-130,26){\small$2$}
	\put(-133,-10){\small$x=1$}
	\put(-93,-10){\small$2$}
	\put(-72,-10){\small$3$}
	\put(-52,-10){\small$4$}
	\put(-31,-10){\small$5$}
	\put(-56,108){\small$\frac{e^\varepsilon}{1+e^\varepsilon}$}
	\put(-36,76){\small$\frac{1}{1+e^\varepsilon}$}
	\put(-163,150){$Q^T=\frac{1}{1+e^\varepsilon}\begin{bmatrix}
		e^\varepsilon&e^\varepsilon&1&e^\varepsilon&1\\
		1&1&e^\varepsilon&1&e^\varepsilon
		\end{bmatrix}$}
		\hspace{1.2cm}
	\includegraphics[width=.32\textwidth]{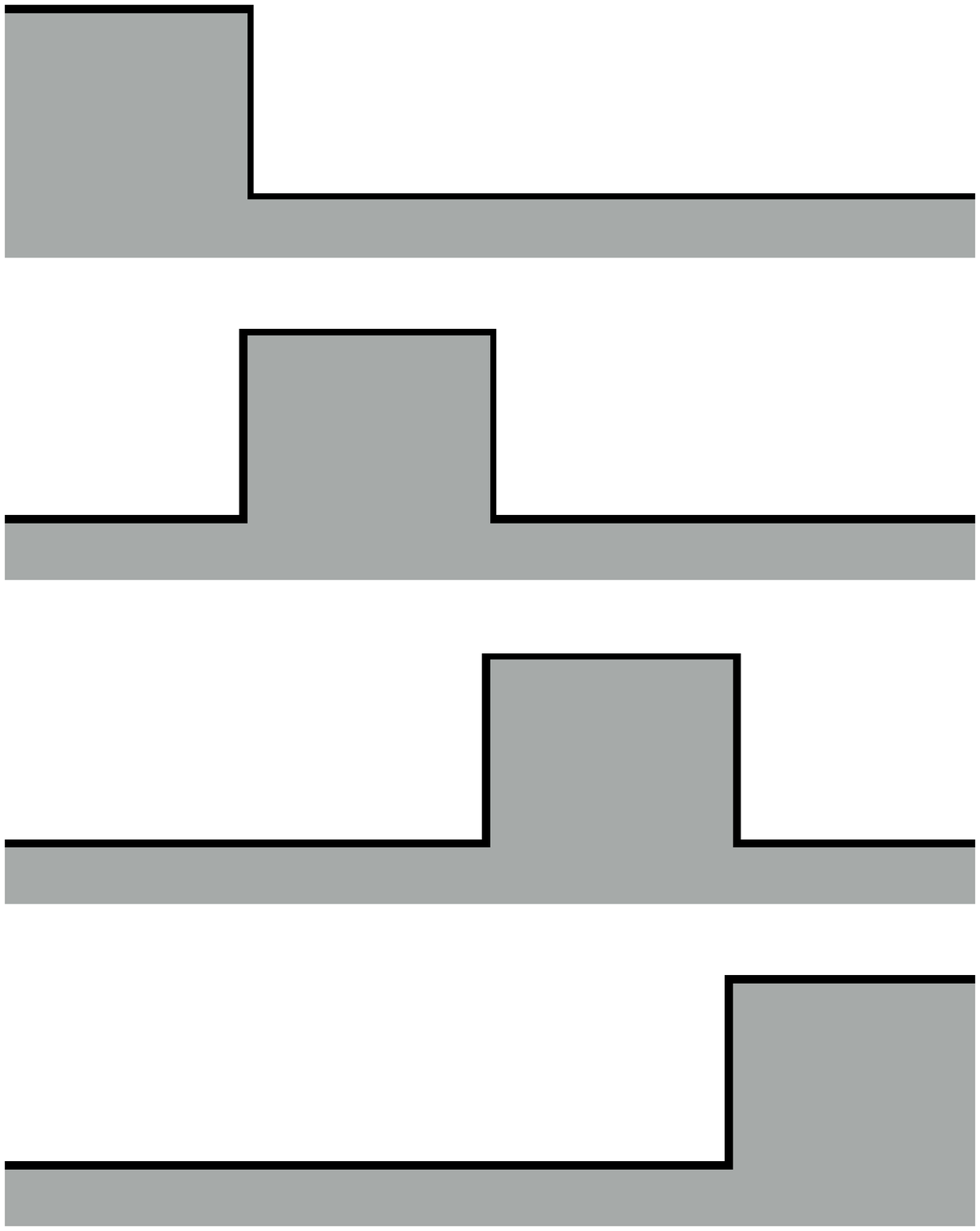}
	\put(-142,92){\small$y=1$}
	\put(-125,64){\small$2$}
	\put(-125,38){\small$3$}
	\put(-125,10){\small$4$}
	\put(-125,-10){\small$x=1$}
	\put(-88,-10){\small$2$}
	\put(-65,-10){\small$3$}
	\put(-43,-10){\small$4$}
	\put(-112,112){\small$\frac{e^\varepsilon}{3+e^\varepsilon}$}
	\put(-90,96){\small$\frac{1}{3+e^\varepsilon}$}
	\put(-163,150){$Q^T=\frac{1}{3+e^\varepsilon}\begin{bmatrix}
		e^\varepsilon&1&1&1\\
		1&e^\varepsilon&1&1\\
		1&1&e^\varepsilon&1\\
		1&1&1&e^\varepsilon
		\end{bmatrix}$}
	\end{center}
	\caption{Examples of staircase mechanisms: the binary (left) and
	the randomized response (right) mechanisms.}
	\label{fig:mechanism}
\end{figure}

For global differential privacy,
we can generalize the definition of staircase mechanisms to hold for all neighboring database queries $\x,\x'$ (or equivalently within some sensitivity),
and recover all known existing optimal mechanisms.
Precisely, the geometric mechanism shown to be optimal in \cite{GRS12},
and the mechanisms shown to be optimal in \cite{GV12,GV13} (also called staircase mechanisms)
are special cases of the staircase mechanisms defined above.
We believe that the characterization of these extremal mechanisms
and the analysis techniques developed in this paper can be of independent interest to
researchers interested  in optimal mechanisms for global privacy and more general utilities.

\subsection{Combinatorial representation of the staircase mechanisms}
Now that we know that staircase mechanisms are optimal, we can try to combinatorially search for the
best staircase mechanism for an instance of the function $\mu$ and  a fixed $\varepsilon$.
To this end, we give a simple representation for  all staircase mechanisms,  exploiting the fact that they are scaled copies of a finite number of patterns.

Let $Q\in\reals^{|\cX|\times |\cY|}$ be a staircase mechanism, and $k=|\cX|$ denote the size of the input alphabet.
Then, from the definition of staircase mechanisms,
$Q(y|x)/Q(y|x') \in \{e^{-\varepsilon},1,e^{\varepsilon} \}$ and
each column $Q(y|\cdot)$ must be proportional to one of the canonical staircase patterns we define next.

\begin{definition}[Staircase Pattern Matrix]
\label{def:staircase_matrix} Let $b_j$ be the $\kk$-dimensional binary vector corresponding to the binary representation of $j$ for $j \leq 2^\kk-1$. A matrix $\s^{(\kk)} \in \{1,e^\varepsilon\}^{\kk\times2^\kk}$ is called a staircase pattern matrix if the $j$-th column of $\s^{(\kk)}$ is $\s^{(\kk)}_j= \left(e^\varepsilon -1\right)b_{j-1} + \unity$, for $j\in\{1,\ldots,2^\kk\}$. Each column of $\s^{(\kk)}$ is a staircase pattern.
\end{definition}
When $k=3$, there are $2^k=8$ staircase patterns and the staircase pattern matrix
is given by
\begin{eqnarray*}
	S^{(3)} = \begin{bmatrix}
	1&1& 1&1&e^{\varepsilon} &e^{\varepsilon}&e^{\varepsilon}&e^{\varepsilon}\\
	1&1& e^{\varepsilon}&e^{\varepsilon}&1&1&e^{\varepsilon}&e^{\varepsilon}\\
	1&e^{\varepsilon}& 1& e^{\varepsilon} & 1& e^{\varepsilon} & 1& e^{\varepsilon}\\
	\end{bmatrix}\;.
\end{eqnarray*}
For all values of $k$, there are exactly $2^k$ such patterns, and
any column $Q(y|\cdot)$ of $Q$, a staircase mechanism,  is a scaled version of one of the columns of $S^{(k)}$.
%
%
Using this pattern matrix, we can show that any staircase mechanism
$Q$ can be represented as
\begin{eqnarray}
	Q&=&S^{(k)} \B\;,
	\label{eq:canonical}
\end{eqnarray}
where $\B={\rm diag}(\bb)$ is a $2^k\times2^k$ diagonal matrix and $\bb$ is a $2^\kk$-dimensional vector representing the scaling of the columns of $S^{(k)}$.
We can now formulate the problem of maximizing the utility
as a linear program and prove their equivalence.
\begin{theorem}
	\label{thm:lp}
	For any sublinear function $\mu$ and any $\varepsilon\geq0$,
	the nonlinear program of \eqref{eq:opt} and the following linear program
	have the same optimal value
\begin{eqnarray}
      \label{eq:optsc}
      \underset{\bb \in\reals^{2^k}}{\text{maximize}} && \sum_{j=1}^{2^k} \mu(S^{(k)}_j) \bb_{j} = \mu^T\bb \\
      \text{subject to}&& S^{(k)}\bb \, =\, \unity\;\nonumber\\
	&& \bb\geq 0\;, \nonumber
\end{eqnarray}
 and the optimal solutions are related by \eqref{eq:canonical}.
 \end{theorem}
Thus, the infinite dimensional nonlinear  program of \eqref{eq:opt} is now reduced
to a finite dimensional linear program. The constraints in \eqref{eq:optsc} ensure that we
get a valid probability matrix $Q=S^{(k)}\Theta$ with rows that sum to one.
One could potentially solve this LP with $2^k$ variables but its
computational complexity scales exponentially in the alphabet size $k=|\cX|$.
For practical values of $k$ this might not always be possible. However, in
the following sections, we prove that in the high privacy regime
($\varepsilon\leq\varepsilon^*$ for some positive $\varepsilon^*$), there is
a single optimal mechanism, which we call the {\em binary mechanism}, which
dominates over all other mechanisms in a very strong sense for all utility
functions of practical interest.

In order to understand the above theorem,
observe that both the objective function and differential privacy constraints
are invariant under {\em permutations} (or relabelling) of the columns of a privatization mechanism $\Q$. In other words, shuffling the columns of an $\varepsilon$-locally differentially private mechanism $\Q$ results in a new $\varepsilon$-locally differentially private mechanism $\Q'$ that achieves the same utility. Similarly, both the objective function and differential privacy constraints
are invariant under {\em merging/splitting} of outputs with the same pattern.
To be specific, consider a privatization mechanism $\Q$ and suppose that there exist two outputs $y$ and $y'$ that have the same pattern, i.e.
 $Q(y|\cdot) = C\,Q(y'|\cdot)$ for some positive constant $C$.
Then, we can consider a new mechanism $Q'$ by merging the two columns corresponding to $y$ and $y'$.
Let $y''$ denote this new output.
It follows that $Q'$ satisfies the differential privacy constraints and the resulting utility is also preserved.
Precisely, using the fact that $Q(y|\cdot) = C\,Q(y'|\cdot)$, it follows that
\begin{eqnarray*}
	 \mu(Q_y) + \mu(Q_{y'})  &=&  \mu((1+C)Q_y)
		\;=\; \mu(Q'_{y''}) \;,
\end{eqnarray*}
 by the homogeneity property of $\mu$. Therefore, we can naturally define equivalence classes for staircase mechanisms
that are equivalent up to a permutation of columns and merging/splitting of columns with the same pattern:
\begin{eqnarray*}
	[Q] = \{ Q' \in \cS_\varepsilon \,|\,\exists \text{ a sequence of permutations and merge/split of columns from }Q' \text{ to }Q \}\;.
\end{eqnarray*}
To represent an equivalence class,
we use a mechanism in $[Q]$ that is ordered and merged to match the patterns of the pattern matrix $S^{(k)}$.
For any staircase mechanism $Q$, there exists a possibly different staircase mechanism $Q'\in[Q]$
 such that $Q'=S^{(k)}\B$ for some diagonal matrix $\B$ with nonnegative entries.
Therefore, to solve optimization problems of the form \eqref{eq:opt},
we can restrict our attention to such representatives of equivalent classes.
Further, for privatization mechanisms of the form $\Q=S^{(k)}\B$, the objective function takes the form $\sum_j \mu(S^{(k)}_j) \bb_{j}$,
 a simple linear function of $\B$.

\section{Hypothesis Testing}
\label{sec:hyp}

In this section, we study the fundamental tradeoff between
local differential privacy and hypothesis testing. In this setting,
there are $n$ individuals each with data $\X_i$ sampled from a distribution $\PP_\nu$
for a fixed $\nu\in\{0,1\}$. Let $\Q$ be a non-interactive privatization
mechanism guaranteeing $\varepsilon$-local differential privacy. The output
of the privatization mechanism $Y_i$ is distributed
according to the induced marginal $\MM_\nu$ defined in
\eqref{eq:defM2}. With a slight abuse of notation, we will use
$\MM_\nu$ and $\PP_\nu$ to represent both probability distributions and
probability mass functions.
The power to discriminate data sampled from $\PP_0$ to data sampled from $\PP_1$
depends on the `distance' between the marginals $\MM_0$ and $\MM_1$.
To measure the ability of such statistical discrimination,
our choice of utility of a privatization mechanism $Q$ is an information theoretic quantity called Csisz\'ar's $f$-divergence defined as
\begin{eqnarray}
	D_f(M_0||M_1) = \sum_{\cY} M_1(y) f\Big(\,\frac{M_0(y)}{M_1(y)}\,\Big) \,=U\left(\Po,\Pt,\Q\right)=U\left(\Q\right)\;, \label{eq:defdiv}
\end{eqnarray}
for some convex function $f$ such that $f(1)=0$. The Kullback-Leibler (KL)
divergence $\Dkl(\MM_0||\MM_1)$ is a special case of $f$-divergence with
$f(x)=x\log x$. The total variation distance $\|\MM_0-\MM_1\|_{\rm TV}$ is also special
case with $f(x)=(1/2)|x-1|$. Note that in general, the $f$-divergence is not necessarily a distance metric
since it need not be symmetric or satisfy triangular inequality. We are
interested in characterizing the optimal solution to
\begin{equation}
\label{eq:optf}
\begin{aligned}
& \underset{\Q}{\text{maximize}}
& & D_f(\MM_0||\MM_1)  \\
& \text{subject to}
& & \Q \in \mathcal{\D}_{\varepsilon}
\end{aligned},
\end{equation}
where $\mathcal{\D}_{\varepsilon}$ is the set of all $\varepsilon$-locally
differentially private mechanisms defined in \eqref{eq:defdp}.

A motivating example for this choice of utility is the Neyman-Pearson hypothesis testing framework \citep{CT12}.
Given the privatized views $\{Y_i\}_{i=1}^{n}$,
the data analyst wants to test whether they are generated from $\MM_0$ or $\MM_1$.
Let the null hypothesis be $H_0:\text{$Y_i$'s are generated from }\MM_0$, and
the alternative hypothesis $H_1:\text{$Y_i$'s are generated from }\MM_1$.
For a choice of rejection region $R\subseteq \cY^n$, the probability of false alarm (type I error)
is $\alpha=\MM_0^n(R)$ and the probability of miss detection (type II error) is $\beta=\MM_1^n( \cY^n\setminus R)$.
Let $\beta^{\alpha^*}=\min_{R\subseteq\cY^n,\alpha<\alpha^*} \beta$ denote the
minimum type II error achievable while keeping the type I error rate at most $\alpha^*$.
According to Chernoff-Stein lemma \citep{CT12}, we know that
\begin{eqnarray*}
	\lim_{n\to\infty} \frac{1}{n}\log \beta^{\alpha^*} = -\Dkl(\MM_0||\MM_1)\;.
\end{eqnarray*}
Suppose the analyst knows $\PP_0$, $\PP_1$, and $Q$.
Then in order to achieve optimal asymptotic error rate,
one would want to maximize the KL divergence of the induced marginals,
over all $\varepsilon$-locally differentially private mechanisms $Q$.
The results we present in this section (Theorems \ref{thm:hypbin} and \ref{thm:hyprr} to be precise) provide an explicit construction of optimal mechanisms
in high and low privacy regimes. Using these optimality results, we prove a fundamental limit on
the achievable error rates under differential privacy.
Precisely, with data collected from an
$\varepsilon$-locally differentially privatization mechanism,
one cannot achieve an asymptotic type II error smaller than
\begin{eqnarray*}
	\lim_{n\to\infty} \frac{1}{n}\log \beta^{\alpha^*} \geq - \frac{(1+\delta)(e^\varepsilon-1)^2}{(e^\varepsilon+1)}\|\PP_0-\PP_1\|_{\rm TV}^2 \;\geq - \frac{(1+\delta)(e^\varepsilon-1)^2}{2(e^\varepsilon+1)} \Dkl(\PP_0||\PP_1) \; \;,
\end{eqnarray*}
whenever $\varepsilon\leq \varepsilon^*$, where $\varepsilon^*$
is dictated by Theorem \ref{thm:hypbin} and $\delta>0$ is some arbitrarily small but positive constant. In the equation above, the second inequality follows from Pinsker's inequality.
Since $(e^{\varepsilon}-1)^2=O(\varepsilon^2)$ for small $\varepsilon$,
the effective sample size is now reduced from $n$ to $\varepsilon^2 n$.
This is the price of privacy.
In the low privacy regime where
$\varepsilon\geq\varepsilon^*$, for $\varepsilon^*$ dictated by Theorem \ref{thm:hyprr}, one cannot achieve an asymptotic type II error smaller than
\begin{eqnarray*}
	\lim_{n\to\infty} \frac{1}{n}\log \beta^{\alpha^*} \geq - \Dkl(\PP_0||\PP_1) + (1-\delta)G(\PP_0,\PP_1)e^{-\varepsilon}\;.
\end{eqnarray*}
\subsection{Optimality of staircase mechanisms}
\label{sec:hypopt} From the definition of $D_f(\Mo||\Mt)$, we have that
\begin{equation*} D_f(\Mo||\Mt)=\sum_{\cY}(\PP_1^T\Q_\y) f(\PP_0^T\Q_\y/\PP_1^T\Q_\y)=\sum_{\cY}\mu\left(\Q_\y\right), \end{equation*}
where $\PP_\nu^T\Q_\y =\sum_{\cX}
\PP_\nu\left(\x\right)\Q\left(\y|\x\right) $ and
$\mu\left(\Q_\y\right)=(\PP_1^T\Q_\y) f(\PP_0^T\Q_\y/\PP_1^T\Q_\y)$. For any
$\gamma >0$,
\begin{eqnarray*}
\mu\left(\gamma \Q_\y\right) &=& \left(\PP_1^T(\gamma\Q_\y)\right) f\left(\PP_0^T\left(\gamma\Q_\y\right)/\PP_1^T\left(\gamma\Q_\y\right)\right) \nonumber \\
&=& \gamma \left(\PP_1^T\Q_\y\right) f\left(\PP_0^T\Q_\y/\PP_1^T\Q_\y\right) \nonumber \\
&=& \gamma \mu\left(\Q_\y\right).
\end{eqnarray*}
Moreover, since the function $\phi(z,t)=tf\left(\frac{z}{t}\right)$ is convex
in $(z,t)$ for $0\leq z,t \leq 1$, then $\mu$ is convex in $\Q_\y$. Convexity
and homogeniety together imply sublinearlity. Therefore, Theorems
\ref{thm:sc} and \ref{thm:lp} apply to $D_f(\Mo||\Mt)$ and we have that
staircases are optimal.

\subsection{Optimality of the binary mechanism}
For a given $\PP_0$ and $\PP_1$, the {\em binary mechanism} is defined as a staircase mechanism with only two outputs
$\Y\in\{0,1\}$ satisfying (see Figure \ref{fig:mechanism})
\begin{eqnarray}
\Q(0|x) \,=\, \left\{
\begin{array}{rl}
	\dfrac{e^\varepsilon}{1+e^\varepsilon}& \text{ if } \PP_0(x)\geq \PP_1(x)\;,\\
	\dfrac{1}{1+e^\varepsilon}& \text{ if } \PP_0(x)< \PP_1(x)\;.\\
\end{array}
\right.\;\;\;
\Q(1|x) \,=\, \left\{
\begin{array}{rl}
	\dfrac{e^\varepsilon}{1+e^\varepsilon}& \text{ if } \PP_0(x)< \PP_1(x)\;,\\
	\dfrac{1}{1+e^\varepsilon}& \text{ if } \PP_0(x)\geq \PP_1(x)\;. \\
\end{array}
\right.
\label{eq:defhypbin}
\end{eqnarray}
Although this mechanism is extremely simple, perhaps surprisingly, we will
establish that it is the optimal mechanism when a high level of privacy is
required. Intuitively, the output should be very noisy in the high privacy regime,
and we are better off sending just one bit of information that tells you
whether your data is more likely to have come from $\PP_0$ or $\PP_1$.
\begin{theorem}
\label{thm:hypbin}
	For any pair of distributions $\PP_0$ and $\PP_1$,
	there exists a positive $\varepsilon^*$ that depends on  $\PP_0$ and $\PP_1$ such that
	for any $f$-divergences and any positive $\varepsilon\leq \varepsilon^*$,
	the binary mechanism maximizes the $f$-divergence between the induced marginals over all
	$\varepsilon$-locally differentially private mechanisms.
\end{theorem}
This implies that in the high privacy regime,
which is a typical setting studied in much of the differential privacy literature,
the binary mechanism is universally optimal for all $f$-divergences.
In particular this threshold $\varepsilon^*$ is {\em universal}, in that it does not depend on the particular choice of
which $f$-divergence we are maximizing. It is only a function of $\PP_0$ and $\PP_1$.
This is established by proving a very strong statistical dominance
using Blackwell's celebrated result on comparisons of statistical experiments  \cite{Bla53}.
In a nutshell, we prove that any $\varepsilon$-locally differentially private mechanism can be simulated from the output of the binary mechanism for sufficiently small $\varepsilon$. Hence, the binary mechanism
dominates over all other mechanisms and at the same time achieves the maximum divergence. A similar idea has been used previously in \citep{KOV14b}
to exactly characterize how much privacy degrades under composition attacks.

The optimality of binary mechanisms is not just for high privacy regimes. The
next theorem shows that it is {\em the} optimal solution  of \eqref{eq:optf}
for all $\varepsilon$, when the objective function is
 the total variation distance: $D_f(\MM_0||\MM_1) = \|\MM_0-\MM_1\|_{\rm TV}$.
\begin{theorem}
\label{thm:hyptv}
	For any pair of distributions $\PP_0$ and $\PP_1$,
	and any $\varepsilon\geq 0$,
	the binary mechanism maximizes the total variation distance between the
	induced marginals $\MM_0$ and $\MM_1$ among all $\varepsilon$-locally differentially private mechanisms.
\end{theorem}

When maximizing the  KL divergence between the induced marginals,
we show that the binary mechanism still achieves good performance for $\varepsilon\leq C$
where $C$ is a constant that does not depend on $\PP_0$ and $\PP_1$.
For the special case of KL divergence,
let ${\rm OPT}$ denote the maximum value of $\eqref{eq:optf}$
and ${\rm BIN}$ denote the KL divergence when the binary mechanism is used.
The next theorem shows that
\begin{eqnarray*}
	{\rm BIN} & \geq & \frac{1}{2(e^\varepsilon+1)^2} {\rm OPT}\;.
\end{eqnarray*}
\begin{theorem}
\label{thm:hypappx}
	For any $\varepsilon$ and any pair of distributions
	$\PP_0$ and $\PP_1$, the binary mechanism is an $1/(2(e^\varepsilon+1)^2)$ approximation of the
	 maximum KL divergence between the
	induced marginals $\MM_0$ and $\MM_1$ among all $\varepsilon$-locally differentially private mechanisms.
\end{theorem}

Observe that $2(e^\varepsilon+1)^2 \leq 32$ for $\varepsilon\leq 1$. Therefore,
for any $\varepsilon\leq 1$, the simple binary mechanism is at most a constant factor away from the optimal mechanism.

\subsection{Optimality of the randomized response mechanism}

The {\em randomized response mechanism} (see Figure \ref{fig:mechanism}) is a staircase mechanism with $\cY=\cX$ satisfying
\begin{eqnarray}
	Q(y|x) \,=\, \left\{
\begin{array}{rl}
	\dfrac{e^\varepsilon}{|\cX|-1+e^\varepsilon}& \text{ if } y=x\;,\\
	\dfrac{1}{|\cX|-1+e^\varepsilon}& \text{ if } y\neq x\;.\\
\end{array}
\right.
	\label{eq:defrr}
\end{eqnarray}
In other words, the randomized response is a simple randomization over the same alphabet where the true data is released with probability ${e^\varepsilon}/\left({|\cX|-1+e^\varepsilon}\right)$. We view it as a multiple choice generalization to the
randomized response method proposed by  \cite{War65}. We now establish that for the special case of optimizing the KL divergence between the induced marginals, the randomized response mechanism is the optimal solution of \eqref{eq:optf}
 in the low privacy regime (i.e., $\varepsilon\geq\varepsilon^*$ for some threshold $\varepsilon^*$ that depends on $\PP_0$ and $\PP_1$).
\begin{theorem}
\label{thm:hyprr}
	There exists a positive $\varepsilon^*$ that depends on $\PP_0$ and $\PP_1$  such that
	for all $\PP_0$ and $\PP_1$, and all $\varepsilon\geq \varepsilon^*$,
	the {randomized response mechanism}  maximizes the KL divergence between the induced marginals
	among all $\varepsilon$-locally differentially private mechanisms.
\end{theorem}
The randomized response mechanism is particularly important because it does not depend on $\PP_0$ or $\PP_1$. Thus, even if the data providers and analysts do not have access to the priors, they can still use the randomized response mechanism to achieve the optimal (or near-optimal) utility in the moderate to low privacy regimes.
\subsection{Numerical experiments}
\label{subsec:num_exp}

A typical approach for achieving $\varepsilon$-local differential privacy is to add geometric noise with appropriately chosen variance.
For an input with alphabet size $|\cX|=k$, this amounts to relabelling the inputs as integers
$\{1,\ldots,k\}$ and adding geometric noise,
i.e., $\Q(\y|\x)= ((1-\varepsilon^{1/(k-1)})/(1+\varepsilon^{1/(k-1)})) \varepsilon^{|y-x|/(k-1)}$ for $y \in \mathbb{Z}$. The output is then truncated at $1$ and $k$ to preserve the support.

For $100$ instances of randomly chosen $P_0$ and $P_1$ defined over an
input alphabet of size $|\cX|=6$, we compare the average
performance of the binary, randomized response, and geometric mechanisms
to the average performance of the optimal staircase mechanism for various values of $\varepsilon$. The optimal staircase mechanism is computed by solving the linear program in Equation \eqref{eq:optsc} for each fixed pair $(P_0,P_1)$ and $\varepsilon$. The left panel of Figure \ref{fig:eps} shows the average performance measured by the normalized divergence $D_{\rm kl}(M_0||M_1)/D_{\rm kl}(P_0||P_1)$ for all 4 mechanisms. The average is taken over the 100 instances of $P_0$ and $P_1$.
In the low privacy (large $\varepsilon$) regime, the randomized response achieves optimal performance (which converges exponentially in $\varepsilon$ to 1) as predicted. In the high privacy regime (small $\varepsilon$), the binary mechanism
achieves optimal performance (which converges quadratically in $\varepsilon$ to 0) as predicted. In all regimes, both the binary and randomized response mechanisms provide significant gains over the geometric mechanism.

\begin{figure}[h]
	\begin{center}
	\includegraphics[width=.47\textwidth]{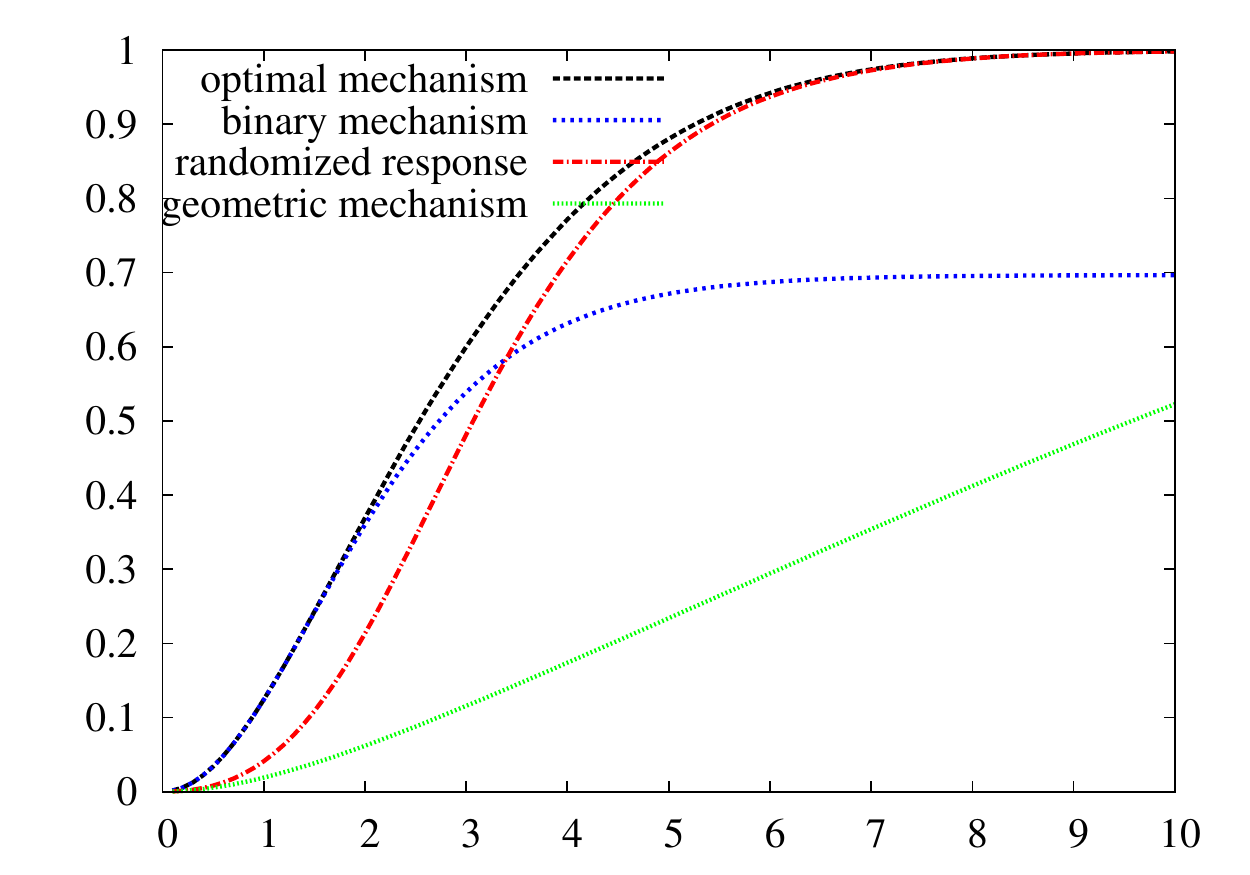}
	\put(-210,144){$\frac{D_{\rm kl}(M_0||M_1)}{D_{\rm kl}(P_0||P_1)}$}
	\put(-100,-10){$\varepsilon$}
	\includegraphics[width=.47\textwidth]{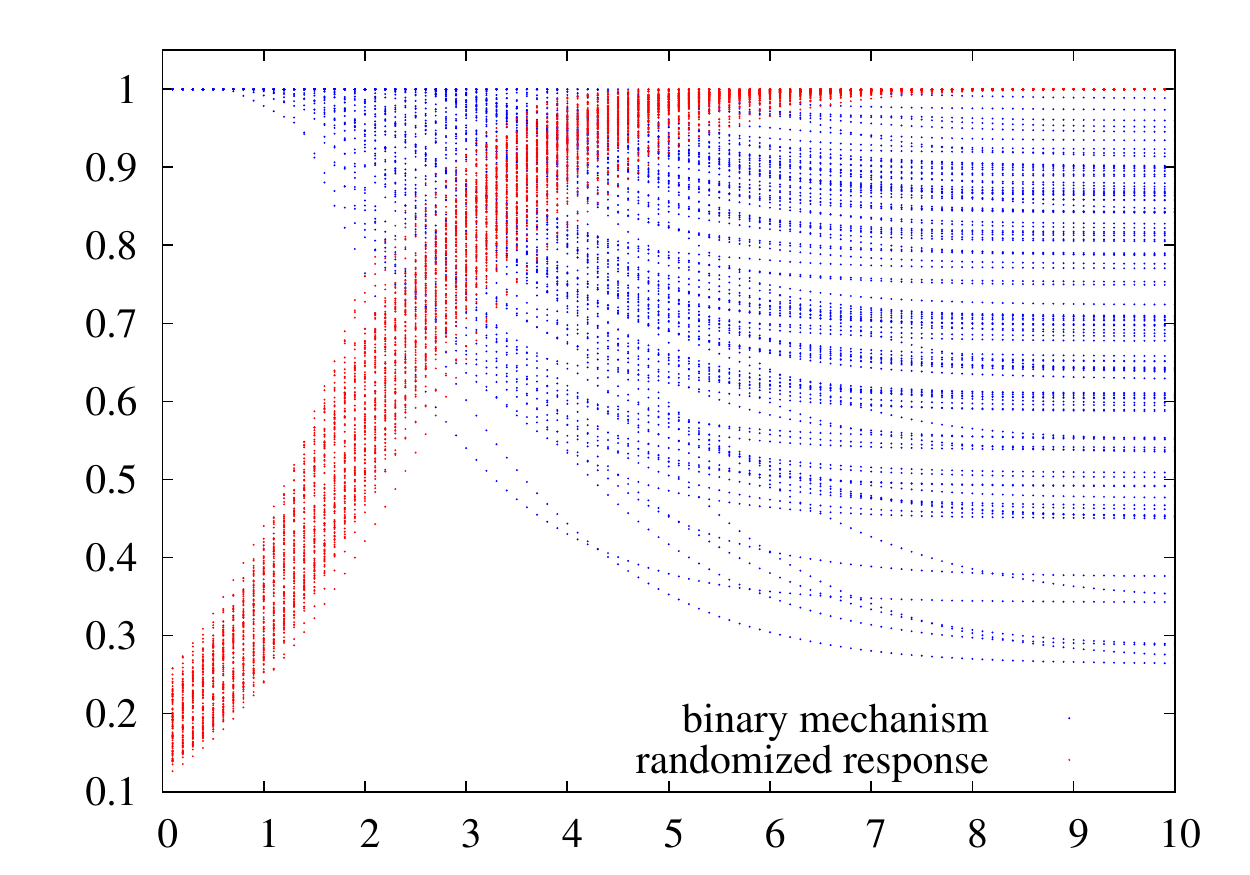}
	\put(-200,144){$\frac{D_{\rm kl}(M_0||M_1)}{OPT}$}
	\put(-100,-10){$\varepsilon$}
	\put(-33.6,18.9){\Huge \textcolor{blue}{$\cdot$}}
	\put(-33.6,12.9){\Huge \textcolor{red}{$\cdot$}}
	\end{center}
	\caption{The binary and randomized response mechanisms are optimal in the high-privacy (small $\varepsilon$) and low-privacy (large $\varepsilon$) regimes, respectively, and improve over the geometric mechanism significantly (left). When the regimes are mismatched, $D_{\rm kl}(M_0||M_1)$ under these mechanisms can be as bad as $10\%$ of the optimal one (right).}
	\label{fig:eps}
\end{figure}
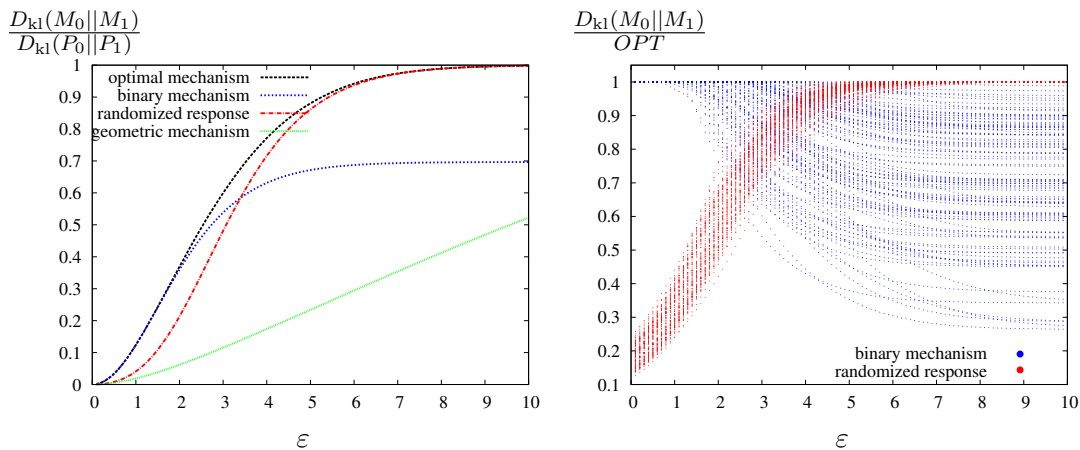
To illustrate how much worse the binary and the randomized response mechanisms can be relative to the optimal staircase mechanism,
we plot in the right panel of Figure \ref{fig:eps} the divergence under each mechanism normalized by the divergence under the optimal mechanism. This is done for all 100 instances of $P_0$ and $P_1$.
In all instances, the binary mechanism is optimal for small $\varepsilon$ and the randomized response mechanism is optimal for large $\varepsilon$.  However, $D_{\rm kl}(M_0||M_1)$ under the randomized response mechanism can be as bad as $10\%$ of the optimal one (for small $\varepsilon$). Similarly, $D_{\rm kl}(M_0||M_1))$ under the binary mechanism can be as bad as $25\%$ of the optimal one (for large $\varepsilon$).
To overcome this issue, we propose the following simple strategy: use the better among these two mechanisms. The performance of this strategy is illustrated in
Figure \ref{fig:max}. For various input alphabet size $|\cX|\in\{3,4,6,12\}$, we plot the performance of this mixed strategy for each value of $\varepsilon$ and each
of the $100$ randomly generated instances of $P_0$ and $P_1$.  This mixed strategy achieves at least $70\%$ for $|\cX| = 6$ (and $55\%$ for $|\cX| = 12$) of the optimal divergence for all instances. Figure \ref{fig:max} shows that this mixed strategy is not too sensitive to the size of the alphabet $k$. Therefore, this strategy provides a good mechanism that can be readily used in practice for any value of $\varepsilon$.

\begin{figure}[h]
	\begin{center}
	\includegraphics[width=.47\textwidth]{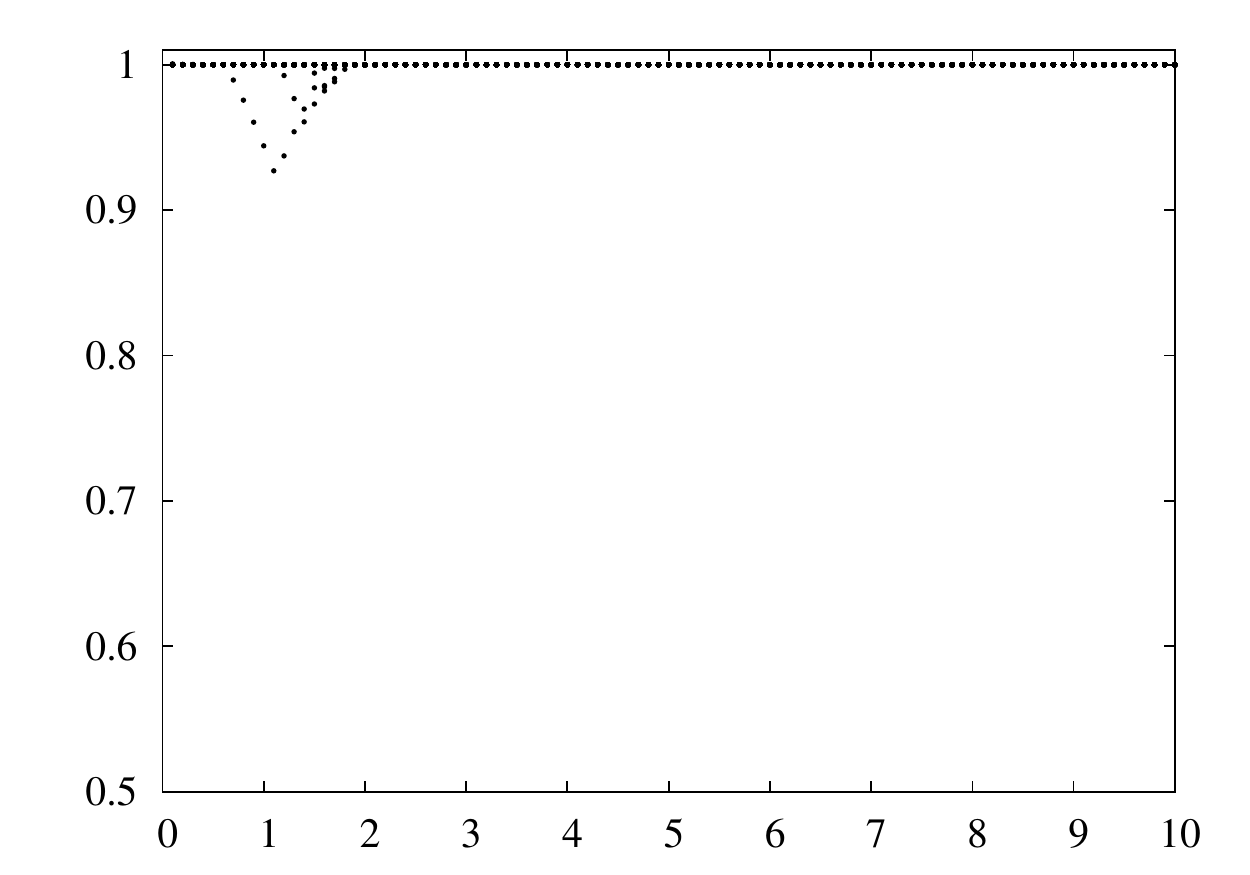}
	\put(-200,144){$\frac{D_{\rm kl}(M_0||M_1)}{OPT}$}
	\put(-100,-5){$\varepsilon$}
	\put(-105,146){$|\cX|=3$}
	\includegraphics[width=.47\textwidth]{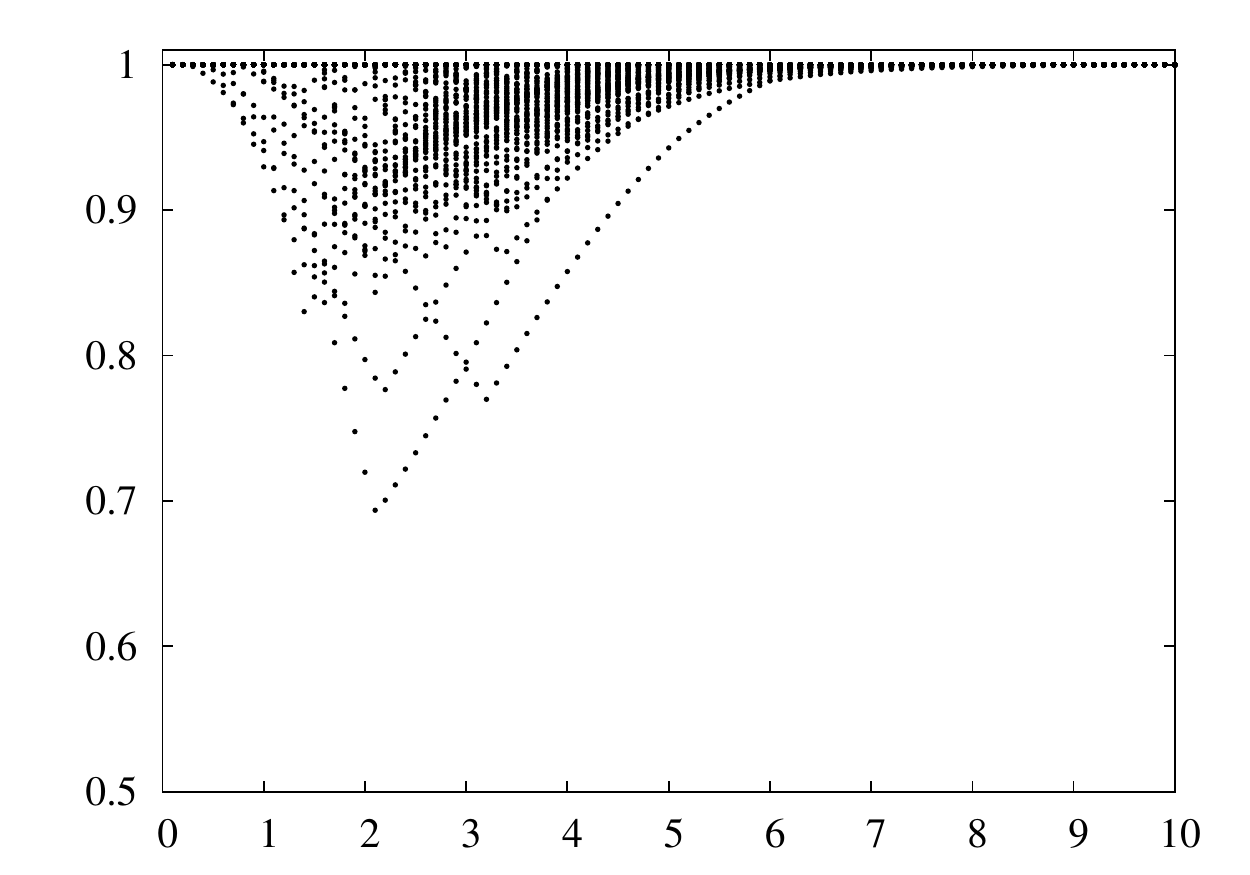}
	\put(-200,144){$\frac{D_{\rm kl}(M_0||M_1)}{OPT}$}
	\put(-100,-5){$\varepsilon$}
	\put(-105,146){$|\cX|=4$}
	\\
	\includegraphics[width=.47\textwidth]{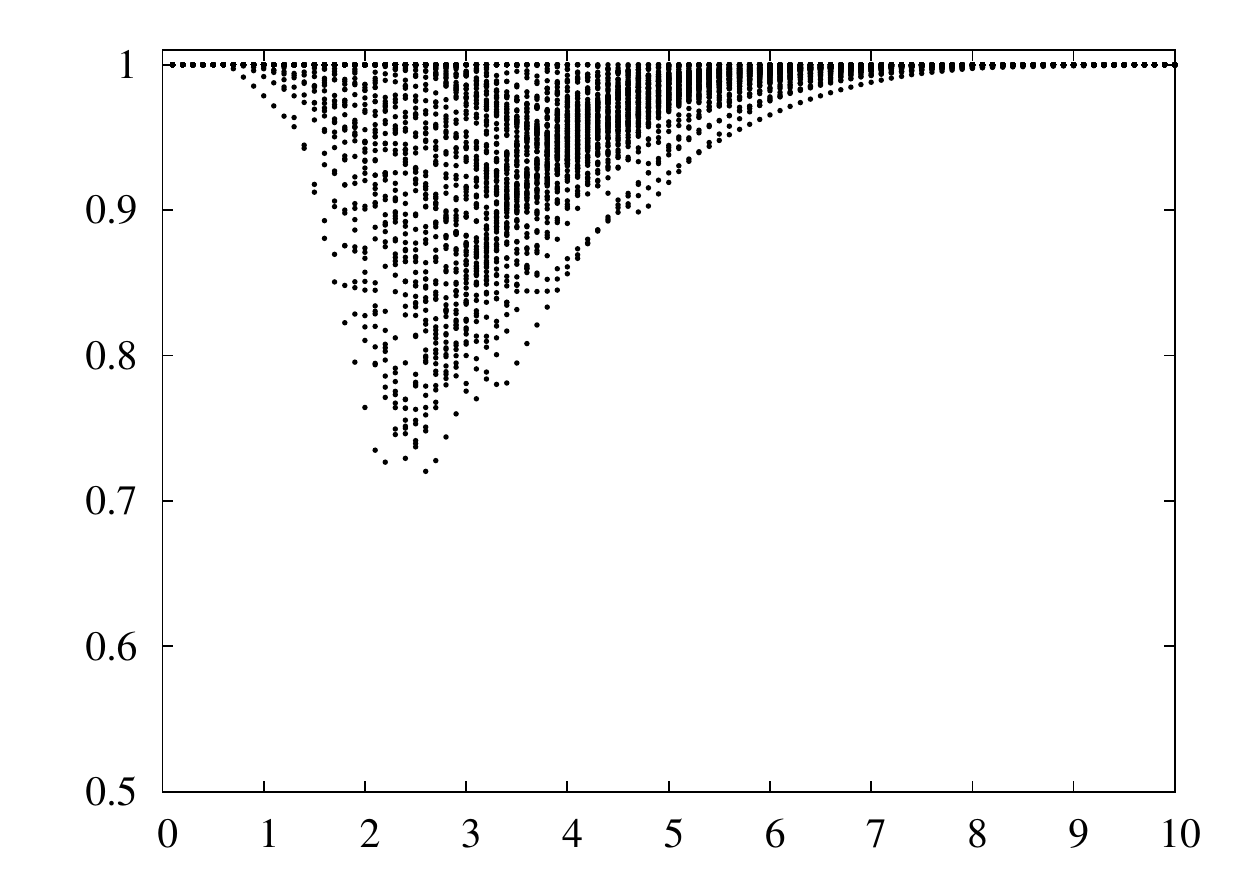}
	\put(-200,144){$\frac{D_{\rm kl}(M_0||M_1)}{OPT}$}
	\put(-100,-5){$\varepsilon$}
	\put(-105,146){$|\cX|=6$}
	\put(-105,160){\phantom{a}}
	\includegraphics[width=.47\textwidth]{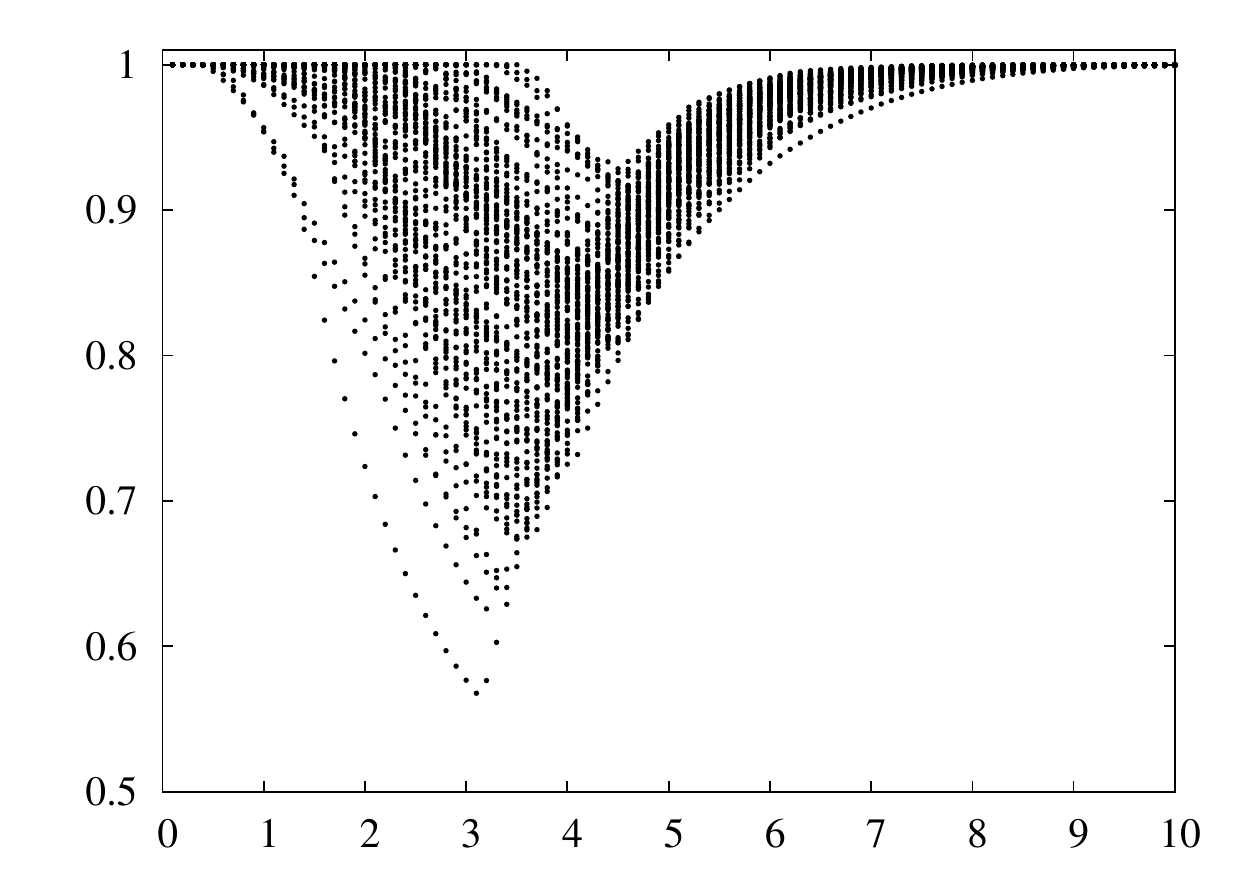}
	\put(-200,144){$\frac{D_{\rm kl}(M_0||M_1)}{OPT}$}
	\put(-100,-5){$\varepsilon$}
	\put(-105,146){$|\cX|=12$}
	\end{center}
	\caption{For varying input alphabet size $|\cX|\in\{3,4,6, 12\}$, at least $55\%$ of the optimal divergence can be achieved by taking the better one between the binary and the randomized response mechanisms.}
	\label{fig:max}
\end{figure}

\subsection{Lower bounds}
\label{sec:hyplb}
In this section, we provide converse results on the fundamental limit of differentially private mechanisms; these results
follow from our main theorems and
are of independent interest in other applications
where lower bounds in statistical analysis are studied \citep{BNO08,HT10,CH12,De12}.
For example,
a bound similar to \eqref{eq:converseKLlow} was used to provide converse results on the sample complexity for
statistical estimation with differentially private data in \cite{DJW13}.

\begin{corollary}
\label{coro:hyphigh}
For any $\varepsilon\geq 0$, let $Q$ be any conditional distribution that guarantees $\varepsilon$-local differential privacy.
Then, for any pair of distributions $\PP_0$, $\PP_1$ and any positive $\delta>0$,
there exists a positive $\varepsilon^*$ that depends on $\PP_0$, $\PP_1$ and $\delta$ such that for any $\varepsilon\leq \varepsilon^*$
the induced marginals $\MM_0$ and $\MM_1$ satisfy the bound
\begin{eqnarray}
	\Dkl\big(\MM_0||\MM_1\big) + \Dkl\big(\MM_1||\MM_0\big) &\leq& \frac{2(1+\delta)(e^{\varepsilon}-1)^2}{(e^\varepsilon+1)} \,\big\| \PP_0-\PP_1 \big\|_{\rm TV}^2\;.
	\label{eq:converseKLlow}
\end{eqnarray}
\end{corollary}
This follows from Theorem \ref{thm:hypbin} and observing that the binary mechanism achieves
\begin{align*}
	&\Dkl\big(\MM_0||\MM_1\big) \\
	&\hspace{0.5cm}=\; \frac{(e^\varepsilon-1)\PP_0(T)+1}{e^\varepsilon+1}\log\Big(\frac{1+(e^\varepsilon-1)\PP_0(T)}{1+(e^\varepsilon-1)\PP_1(T)}\Big) \\
&~~~~~~~~~~~~~~~~~~~~~~~~~~~~~~~~~~
 	+ \frac{(e^\varepsilon-1)\PP_0(T^c)+1}{e^\varepsilon+1}\log\Big(\frac{1+(e^\varepsilon-1)\PP_0(T^c)}{1+(e^\varepsilon-1)\PP_1(T^c)}\Big) \\
	&\hspace{0.5cm}=\; \frac{(e^\varepsilon-1)^2}{e^\varepsilon+1} (\PP_0(T)-\PP_1(T)) + O(\varepsilon^3)\\
	&\hspace{0.5cm}=\; \frac{(e^{\varepsilon}-1)^2}{e^\varepsilon+1}\,\big\| \PP_0-\PP_1 \big\|_{\rm TV}^2  + O(\varepsilon^3)\;,
\end{align*}
where $T\subseteq\cX$ is the set of $\x$ such that $\PP_0(\x)\geq\PP_1(\x)$.
Compared to \cite[Theorem 1]{DJW13}, we recover their bound of $4(e^\varepsilon-1)^2 \|\PP_0-\PP_1\|_{\rm TV}^2$
with a smaller constant.
We want to note that Duchi et al.'s bound holds for all values of $\varepsilon$ and uses a different technique of
bounding the KL divergence directly, however no achievable mechanism has been provided.
We instead provide an explicit mechanism, that is optimal in high privacy regime.

Similarly, in the low privacy regime, we can show the following converse result.
\begin{corollary}
\label{coro:hyplow}
For any $\varepsilon\geq 0$, let $Q$ be any conditional distribution that guarantees $\varepsilon$-local differential privacy.
Then, for any pair of distributions $\PP_0$ and $\PP_1$ and any positive $\delta>0$,
there exists a positive $\varepsilon^*$ that depends on $\PP_0$ and $\PP_1$ and $\delta$ such that for any $\varepsilon\geq \varepsilon^*$
the induced marginals $\MM_0$ and $\MM_1$ satisfy the bound
\begin{eqnarray}
	\Dkl\big(\MM_0||\MM_1\big) + \Dkl\big(\MM_1||\MM_0\big) &\leq&  \Dkl(\PP_0||\PP_1) - (1-\delta)G(P_0,P_1) e^{-\varepsilon} \;.\label{eq:rr_cor}
\end{eqnarray}
where $G(\PP_0,\PP_1)  = \sum_{\cX} (1-\PP_0(x)) \log (\PP_1(x)/\PP_0(x))$.
\end{corollary}
This follows directly from Theorem \ref{thm:hyprr}
and observing that the randomized response mechanism achieves
$
\Dkl(M_0||M_1) = \Dkl(P_0||P_1) - G(P_0,P_1) e^{-\varepsilon} + O( e^{-2\varepsilon})\;.$

Figure \ref{fig:scatter} illustrates the gap between the divergence achieved by
the geometric mechanism described in the previous section and
the optimal mechanisms (the binary mechanism for the high privacy regime and the randomized response mechanism for the low privacy regime).
For each instance of the $100$ randomly generated $P_0$ and $P_1$ defined over input alphabets of size $k=6$,
we plot the resulting divergence $D_{\rm kl}(M_0||M_1)$
as a function of $\|P_0-P_1\|_{\rm TV}$ for $\varepsilon=0.1$,
and as a function of $D_{\rm kl}(P_0||P_1)$ for $\varepsilon=10$.
The binary and the randomized response mechanisms exhibit the scaling predicted by Equation \eqref{eq:converseKLlow} and \eqref{eq:rr_cor}, respectively.
\begin{figure}[h]
	\begin{center}
	\includegraphics[width=.47\textwidth]{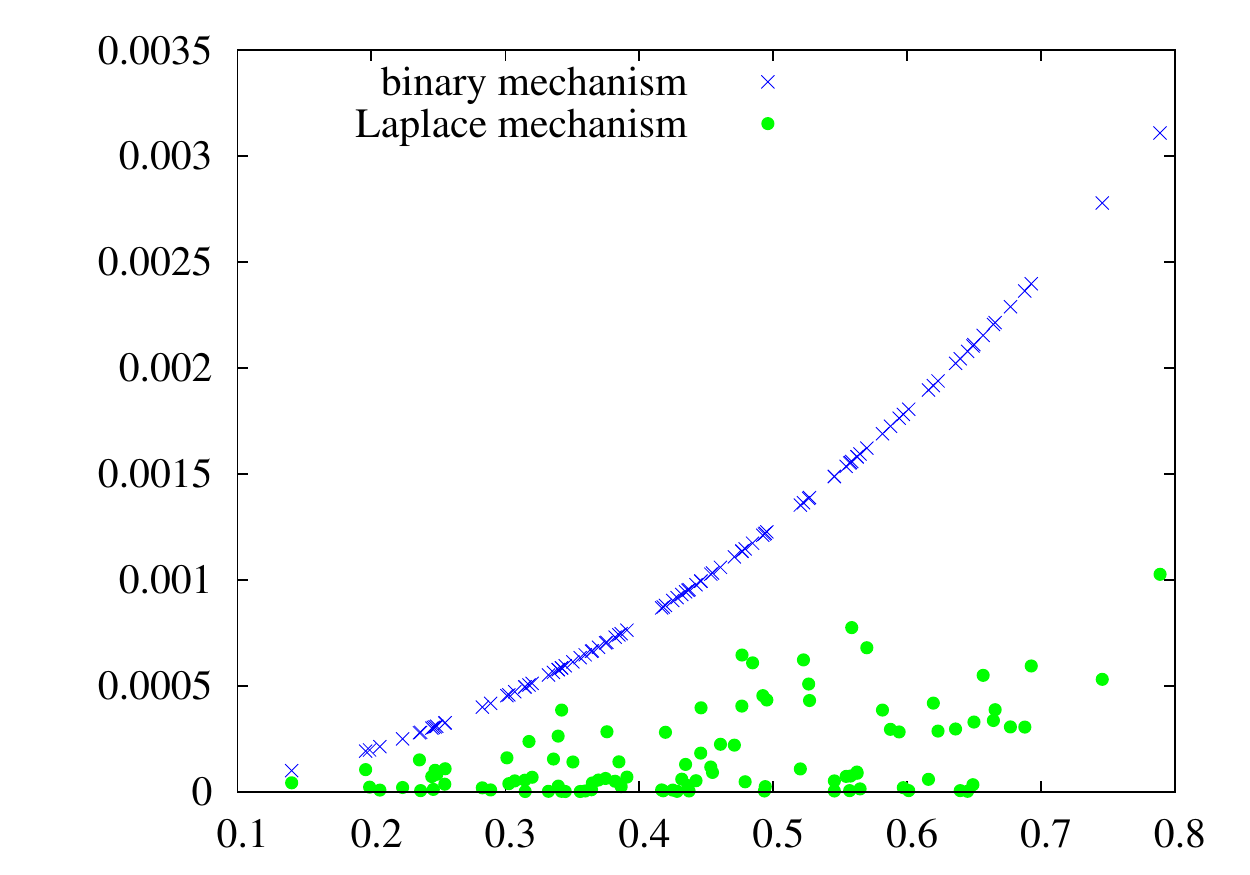}
	\put(-200,144){$D_{\rm kl}(M_0||M_1)$}
	\put(-120,-10){$\|P_0-P_1\|_{\rm TV}$}
	\includegraphics[width=.47\textwidth]{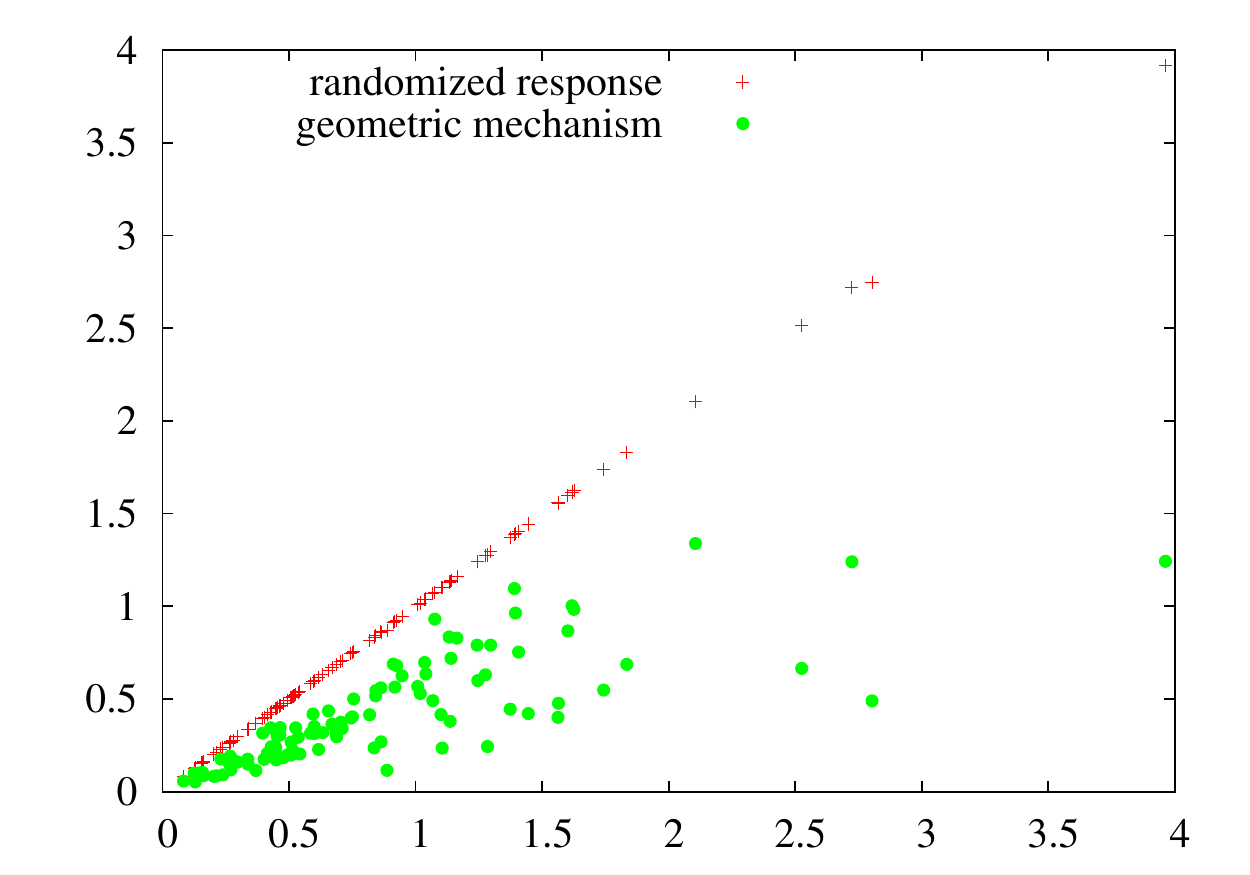}
	\put(-200,144){$D_{\rm kl}(M_0||M_1)$}
	\put(-120,-10){$D_{\rm kl}(P_0||P_1)$}
	\end{center}
	\caption{For small $\varepsilon=0.1$ (left) the binary mechanism achieves the optimal KL divergence, which scales as Equation \eqref{eq:converseKLlow}. For large $\varepsilon=10$ (right) the randomized response achieves the optimal KL divergence, which scales as Equation \eqref{eq:rr_cor}. Both mechanisms improve significantly over the geometric mechanism.}
	\label{fig:scatter}
\end{figure}
Similarly, for total variation, we can get the following converse result.
\begin{corollary}
\label{coro:hyptv}
For any $\varepsilon\geq 0$, let $Q$ be any conditional distribution that guarantees $\varepsilon$-local differential privacy. Then, for any pair of distributions $\PP_0$ and $\PP_1$, the induced marginals $\MM_0$ and $\MM_1$ satisfy the bound
$\big\| \MM_0-\MM_1\big\|_{\rm TV} \leq ( ({e^{\varepsilon}-1)}/({e^\varepsilon+1})) \,\big\| \PP_0-\PP_1 \big\|_{\rm TV}\;,$
and equality is achieved by the binary mechanism.
\end{corollary}
This follows from Theorem \ref{thm:hyptv} and explicitly computing the total variation achieved by the binary mechanism.

\section{Information Preservation}
\label{sec:est}

In this section, we study the fundamental tradeoff between
local privacy and mutual information. Consider a random variable $\X$ distributed according to $\PP$. The information content of $\X$ is captured by information theoretic quantity called entropy
\begin{equation*}
H\left(\X\right)=-\sum_{\cX}\pp\left(\x\right)\log
\pp\left(\x\right).
\end{equation*}
We are interested in releasing a differentially private version of $\X$ represented by $\Y$. The random variable $\Y$
should preserve the information content of $\X$ as much as possible while
meeting the local differential privacy constraints.
Similar to the hypothesis testing setting,
we will show that a variant of the binary mechanism is optimal in the high privacy regime and that the
randomized response mechanism is optimal in the low privacy regime.

Let $\Q$ be a non-interactive
privatization mechanism guaranteeing $\varepsilon$-local differential
privacy. The output of the privatization mechanism $\Y$ is distributed
according to the induced marginal $\MM$ given by
\begin{eqnarray*}
	\MM(S) &=& \sum_{\x\in\cX} \Q(S|\x)  \PP(\x) \;,
\end{eqnarray*}
for $S \subseteq \cY$. With a
slight abuse of notation, we will use $\MM$ and $\PP$ to
represent both probability distributions and probability mass functions. The
information content in $\Y$ about $\X$ is captured by the well celebrated
information theoretic quantity called mutual information. The mutual
information between $\X$ and $\Y$ is given by
\begin{equation*}
I\left(\X;\Y\right)=\sum_{\cX}\sum_{ \cY} \pp\left(\x\right)\Q\left(\y|\x\right)\log\left(\frac{\Q\left(\y|\x\right)}{\sum_{l \in \cX}\pp\left(l\right)\Q\left(\y|l\right)}\right)=U\left(\PP,\Q\right)=U\left(\Q\right).
\end{equation*}
Notice that $I\left(\X;\Y\right)\ \leq
H\left(\X\right)$ and $I\left(\X;\Y\right)$ is convex in $\Q$ \citep{CT12}. To preserve the information context of $\X$, we wish to choose a privatization
mechanism $\Q$ such that the mutual information between $\X$ and $\Y$ is
maximized subject to differential privacy constraints.
In other words, we are interested in characterizing the optimal solution to
\begin{equation}
\label{eq:opti}
\begin{aligned}
& \underset{\Q}{\text{maximize}}
& & I\left(\X;\Y\right)  \\
& \text{subject to}
& & \Q \in \mathcal{\D}_{\varepsilon}
\end{aligned},
\end{equation}
where $\mathcal{\D}_{\varepsilon}$ is the set of all $\varepsilon$-locally
differentially private mechanisms defined in \eqref{eq:defdp}. The above
mutual information maximization problem can be thought of as a conditional
entropy minimization problem since $I\left(\X;\Y\right) = H\left(\X\right)-
H\left(\X|\Y\right)$.
\subsection{Optimality of staircase mechanisms}
\label{sec:estopt}
From the definition of $I\left(\X;\Y\right)$, we have that
\begin{equation*}
I\left(\X;\Y\right)=\sum_{\cY} \sum_{\cX} \pp\left(\x\right)\Q\left(\y|\x\right)\log\left(\frac{\Q\left(\y|\x\right)}{\pp^T\Q_\y}\right)= \sum_{\cY} \mu\left(\Q_\y\right),
\end{equation*}
where $\pp^T\Q_\y=\sum_{\cX}\pp(\x)\Q\left(\y|\x\right)$ and $\mu\left(\Q_\y\right)=\sum_{\cX} \pp\left(\x\right)\Q\left(\y|\x\right)\log\left(\Q\left(\y|\x\right)/\pp^T\Q_\y\right)$. Notice that
$\mu\left(\gamma\Q_\y\right)=\gamma\mu\left(\Q_\y\right)$, and by the log-sum
inequality, $\mu$ is convex. Convexity and homogeneity together imply sublinearity.
Therefore, Theorems \ref{thm:sc} and \ref{thm:lp} apply to
$I\left(\X;\Y\right)$ and we have that staircase mechanisms are optimal.

\subsection{Optimality of the binary mechanism}
For a given $\PP$, the {\em binary mechanism for mutual information}
is a staircase mechanism with only two outputs
$\Y\in\{0,1\}$ (see Figure \ref{fig:mechanism})
\begin{eqnarray}
\Q(0|x) \,=\, \left\{
\begin{array}{rl}
	\dfrac{e^\varepsilon}{1+e^\varepsilon}& \text{ if } \x \in \T\;,\\
	\dfrac{1}{1+e^\varepsilon}& \text{ if } \x\notin \T\;,\\
\end{array}
\right.\;\;\;
\Q(1|x) \,=\, \left\{
\begin{array}{rl}
	\dfrac{e^\varepsilon}{1+e^\varepsilon}& \text{ if } \x\notin \T \;,\\
	\dfrac{1}{1+e^\varepsilon}& \text{ if } \x\in \T \;, \\
\end{array}
\right.
\label{eq:defestbin}
\end{eqnarray}
with
\begin{eqnarray}
\label{eq:tbin}
	\T &\in& \underset{A\subseteq \cX}{\arg\min} \;\; \Big| \PP(A) - \frac12 \Big|\;.
\end{eqnarray}
Observe that there are multiple valid choices for $\T$. Indeed, for any minimizing set $\T$, $\T^c$ is also a minimizing set since $|\PP(\T)-1/2|=|\PP(\T^c)-1/2|$.
When there are multiple pairs, any pair $(\T,\T^c)$ can be chosen to define the binary mechanism. All valid binary mechanisms are equivalent from a utility maximization perspective.

In what follows, we will establish that this simple mechanism is optimal in the high privacy regime. Intuitively, in the high privacy regime,
we cannot release more than one bit of information, and hence, the input alphabet is
reduced to a binary output alphabet. In this case, it makes sense to partition the original alphabet $\cX$ in a way that preserves the information content of $X$ as much as possible. Indeed, our choice of $\T$ in Equation \eqref{eq:tbin} maximizes the information
contained in the released bit because $\T \in \underset{A\subseteq \cX}{\arg\max}|\PP(A)-1/2| =
\underset{A\subseteq \cX}{\arg\max}\big(-\PP(A)\log \PP(A) - \PP(A^c)\log \PP(A^c)\big) $ (see Section \ref{sec:estbin} for a proof).
\begin{theorem}
\label{thm:estbin}
	For any distribution $\PP$,
	there exists a positive $\varepsilon^*$ that depends on  $\PP$ such that
	for any positive $\varepsilon\leq \varepsilon^*$,
	the binary mechanism maximizes the mutual information between the input and the output of a privatization mechanism
	over all $\varepsilon$-locally differentially private mechanisms.
\end{theorem}
This implies that in the high privacy regime,
the binary mechanism is the optimal solution for \eqref{eq:opti}.

Next, we show that the binary mechanism achieves near-optimal performance for all $(\cX,\PP)$ and $\varepsilon \leq 1$ even when $\varepsilon^*<1$.
Let ${\rm OPT}$ denote the maximum value of $\eqref{eq:opti}$
and ${\rm BIN}$ denote the mutual information achieved by the binary mechanism given in \eqref{eq:defestbin}.
The next theorem shows that
\begin{eqnarray*}
	{\rm BIN} & \geq & \frac{1}{1+e^\varepsilon} {\rm OPT}\;.
\end{eqnarray*}
\begin{theorem}
\label{thm:estappx}
	For any $\varepsilon \leq 1$ and any distribution $\PP$, the binary mechanism is an $(1+e^\varepsilon)$-approximation of the
	 maximum mutual information between the input and the output of a privatization mechanism
	among all $\varepsilon$-locally differentially private mechanisms.
\end{theorem}
Note that $1+e^\varepsilon \leq 4$ for $\varepsilon\leq 1$ which is a commonly studied regime in differential privacy applications. Therefore,
we can always use the simple binary mechanism and the resulting mutual information is at most a constant factor away from the optimal value.

\subsection{Optimality of the randomized response mechanism}
In the low privacy regime ($\varepsilon\geq\varepsilon^*$), the {\em randomized response mechanism} defined in\eqref{eq:defrr} is optimal.
\begin{theorem}
\label{thm:estrr}
	There exists a positive $\varepsilon^*$ that depends on $\PP$  such that
	for any distribution $\PP$  and all $\varepsilon\geq \varepsilon^*$,
	the {randomized response mechanism}  maximizes the mutual information between the input and the output of a privatization mechanism
	over all $\varepsilon$-locally differentially private mechanisms.
\end{theorem}
Observe that the randomized response is not a function of $P$. Therefore, it can be used even when the distribution $P$ is unknown.
\subsection{Numerical experiments}
For $100$ instances of randomly chosen $\PP$ defined over an
input alphabet of size $|\cX|=6$, we compare the average
performance of the binary, randomized response, and the geometric mechanisms
to the average performance of the optimal mechanism. We plot (in Figure \ref{fig:epsI}, left) the average performance measured by the normalized mutual information ${I\left(\X;\Y\right)}/{H\left(\X\right)}$ for all 4 mechanisms. The average is taken over the 100 random instances of $\PP$. In the low privacy (large $\varepsilon$) regime, the randomized response achieves optimal performance as predicted, which converges to one. In the high privacy regime (small $\varepsilon$), the binary mechanism achieves optimal performance as predicted. In all regimes, both mechanisms significantly improve over the geometric mechanism.
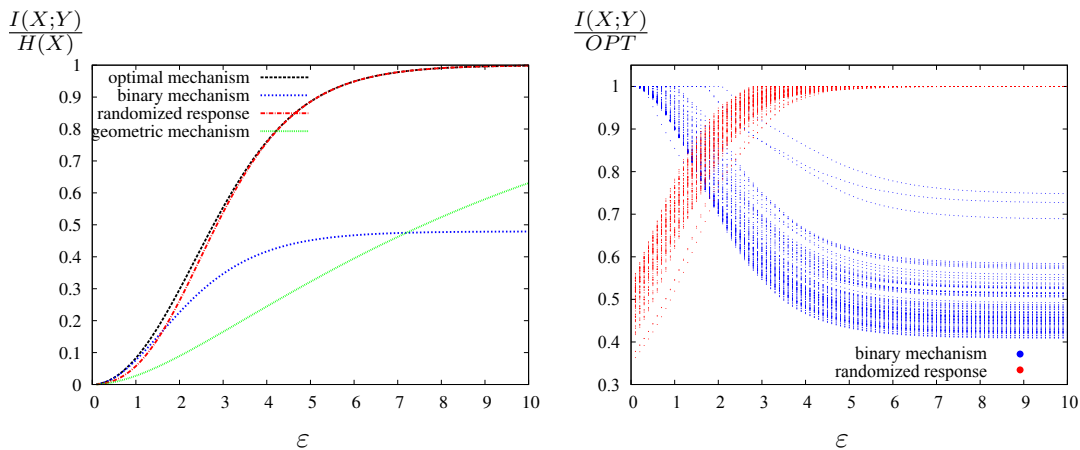
\begin{figure}
	\begin{center}
	\includegraphics[width=.47\textwidth]{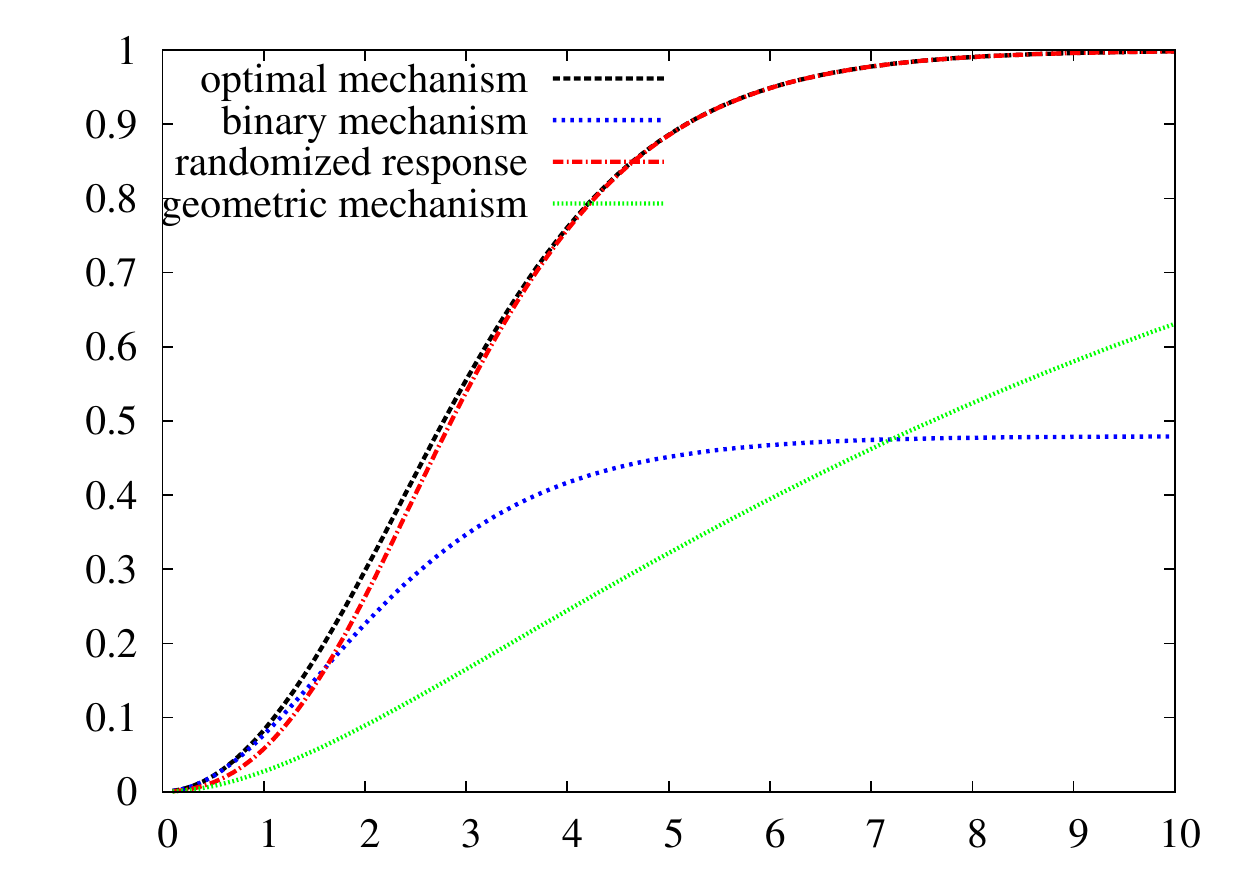}
	\put(-210,144){$\frac{I\left(\X;\Y\right)}{H\left(\X\right)}$}
	\put(-100,-10){$\varepsilon$}
	\includegraphics[width=.47\textwidth]{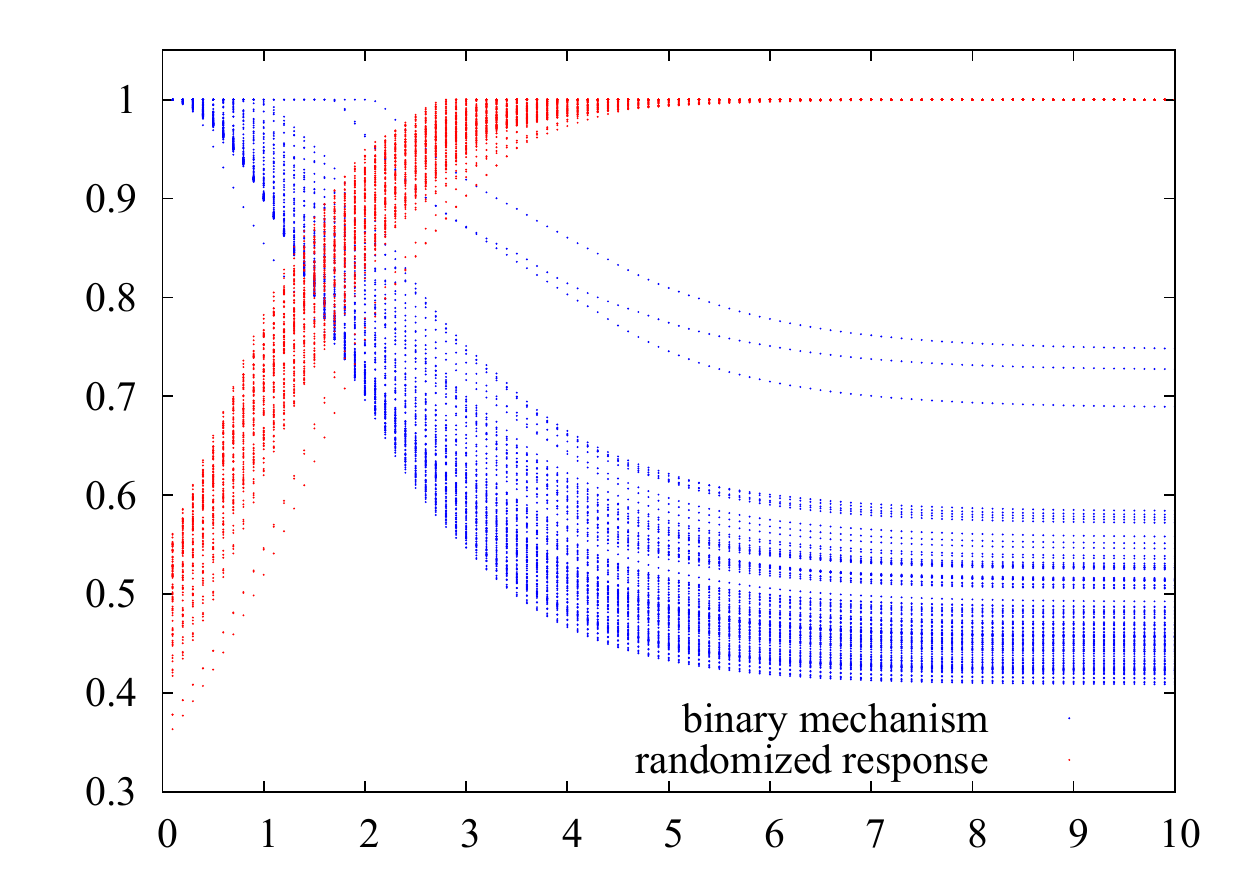}
	\put(-200,144){$\frac{I\left(\X;\Y\right)}{OPT}$}
	\put(-100,-10){$\varepsilon$}
	\put(-33.6,18.9){\Huge \textcolor{blue}{$\cdot$}}
	\put(-33.6,12.9){\Huge \textcolor{red}{$\cdot$}}
	\end{center}
	\caption{The binary and randomized response mechanisms are optimal in the high-privacy (small $\varepsilon$) and low-privacy (large $\varepsilon$) regimes, respectively, and improve over the geometric mechanism significantly (left). When the regimes are mismatched, $I\left(\X;\Y\right)$ under these mechanisms can each be as bad as $35\%$ of the optimal one (right).}
	\label{fig:epsI}
\end{figure}
To illustrate how much worse the binary and randomized response mechanisms can be (relative to the optimal staircase mechanism), we plot (in Figure \ref{fig:epsI}, right) the mutual information under each mechanism normalized by the mutual information under the optimal staircase mechanism. This is done for all 100 instances of $\PP$. In all instances, the binary mechanism is optimal for small $\varepsilon$ and the randomized response mechanism is optimal for large $\varepsilon$. However, $I\left(\X;\Y\right)$ under the randomized response mechanism can be as bad as $35\%$ of the optimal one (for small $\varepsilon$). Similarly, $I\left(\X;\Y\right)$ under the binary mechanism can be as bad as $40\%$ of the optimal one (for large $\varepsilon$).

For $|\cX|\in\{3,4,6, 12\}$, we plot (in Figure \ref{fig:maxI}) the performance of the better between the binary and randomized response mechanisms normalized by the optimal mechanism for all $100$ randomly generated instances of $\PP$. This mixed strategy achieves at least $75\%$ for $|\cX| = 6$ (and $65\%$ for $|\cX| = 12$) of the optimal mutual infirmation for all instances of $\PP$. Moreover, it is not sensitive to the size of the alphabet $|\cX|$.
\begin{figure}[tb!]
	\begin{center}
	\includegraphics[width=.47\textwidth]{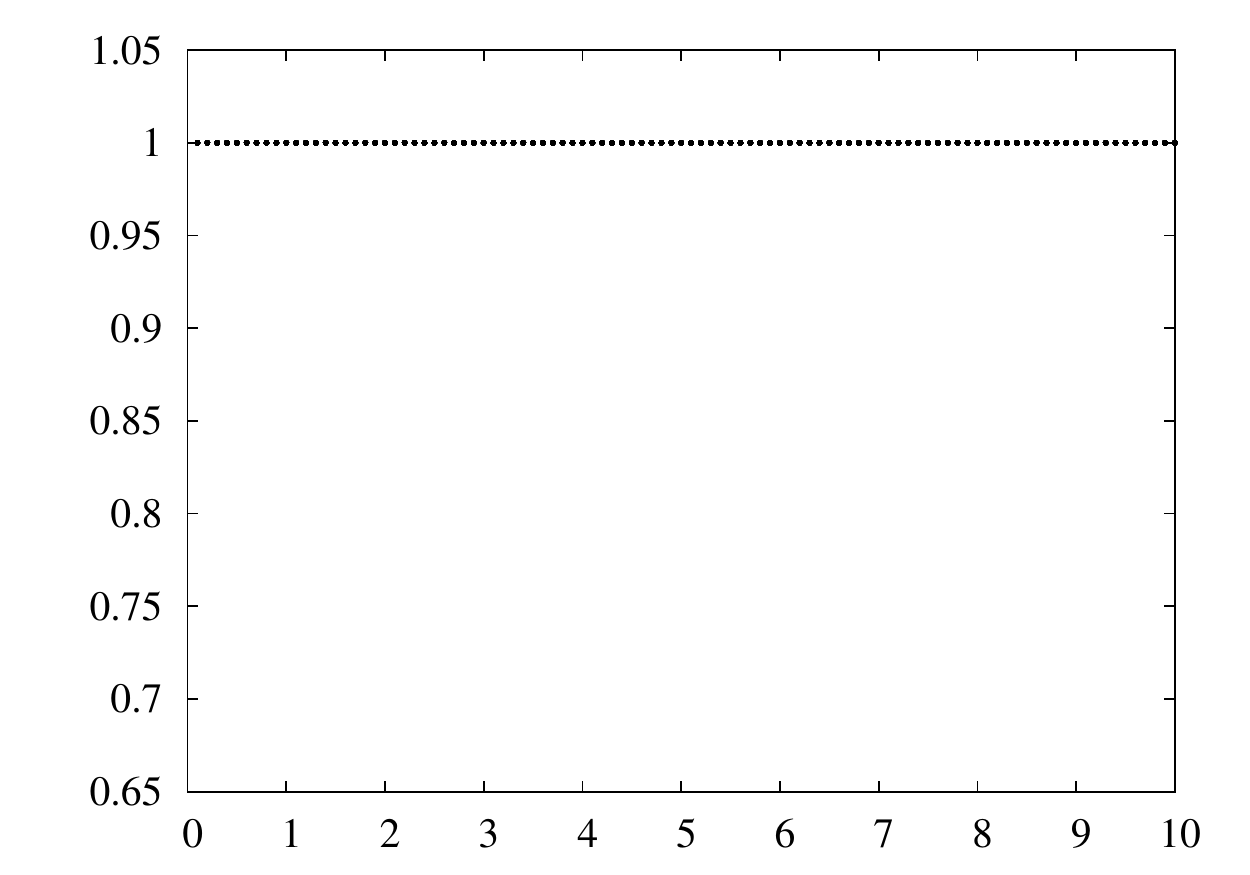}
	\put(-200,144){$\frac{I\left(\X;\Y\right)}{OPT}$}
	\put(-100,-5){$\varepsilon$}
	\put(-105,146){$|\cX|=3$}
	\includegraphics[width=.47\textwidth]{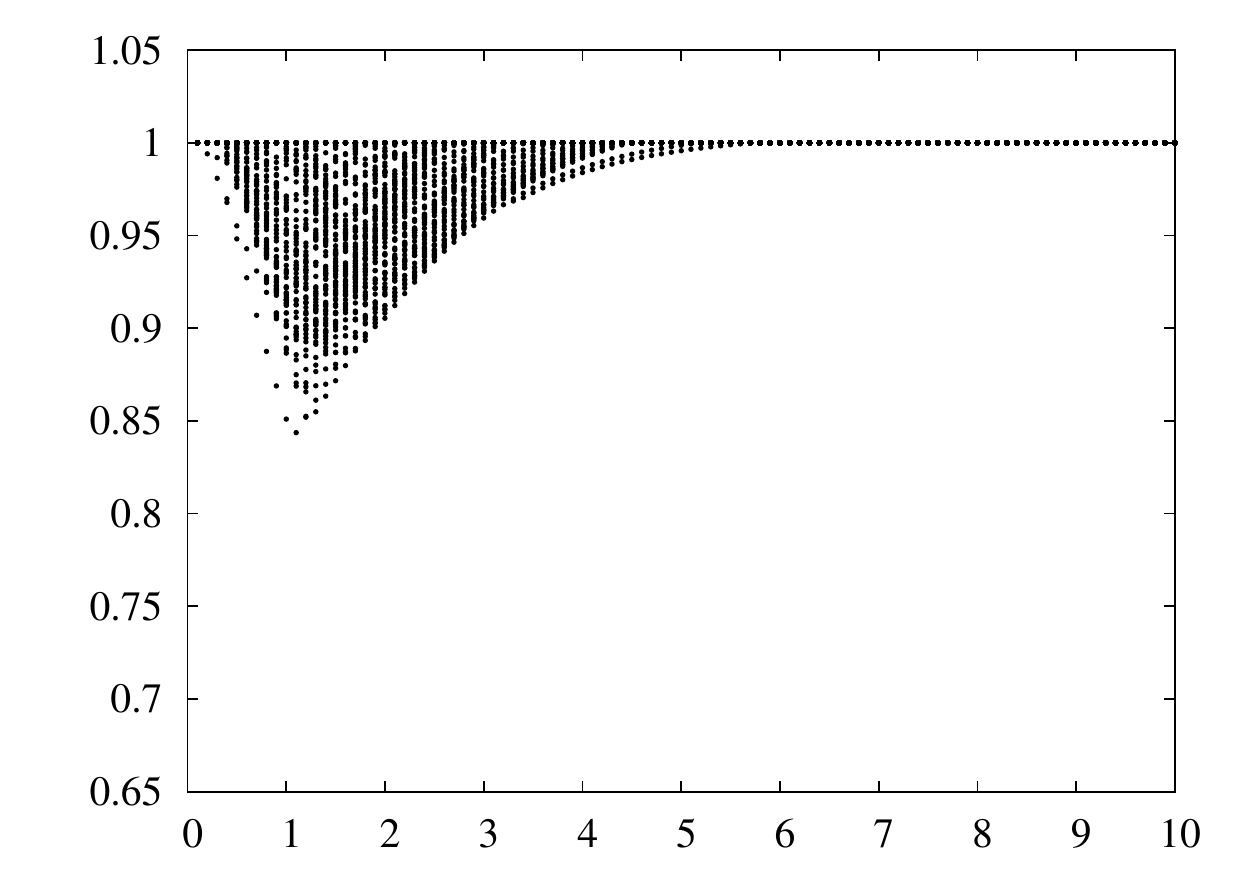}
	\put(-200,144){$\frac{I\left(\X;\Y\right)}{OPT}$}
	\put(-100,-5){$\varepsilon$}
	\put(-105,146){$|\cX|=4$}
	\\
	\includegraphics[width=.47\textwidth]{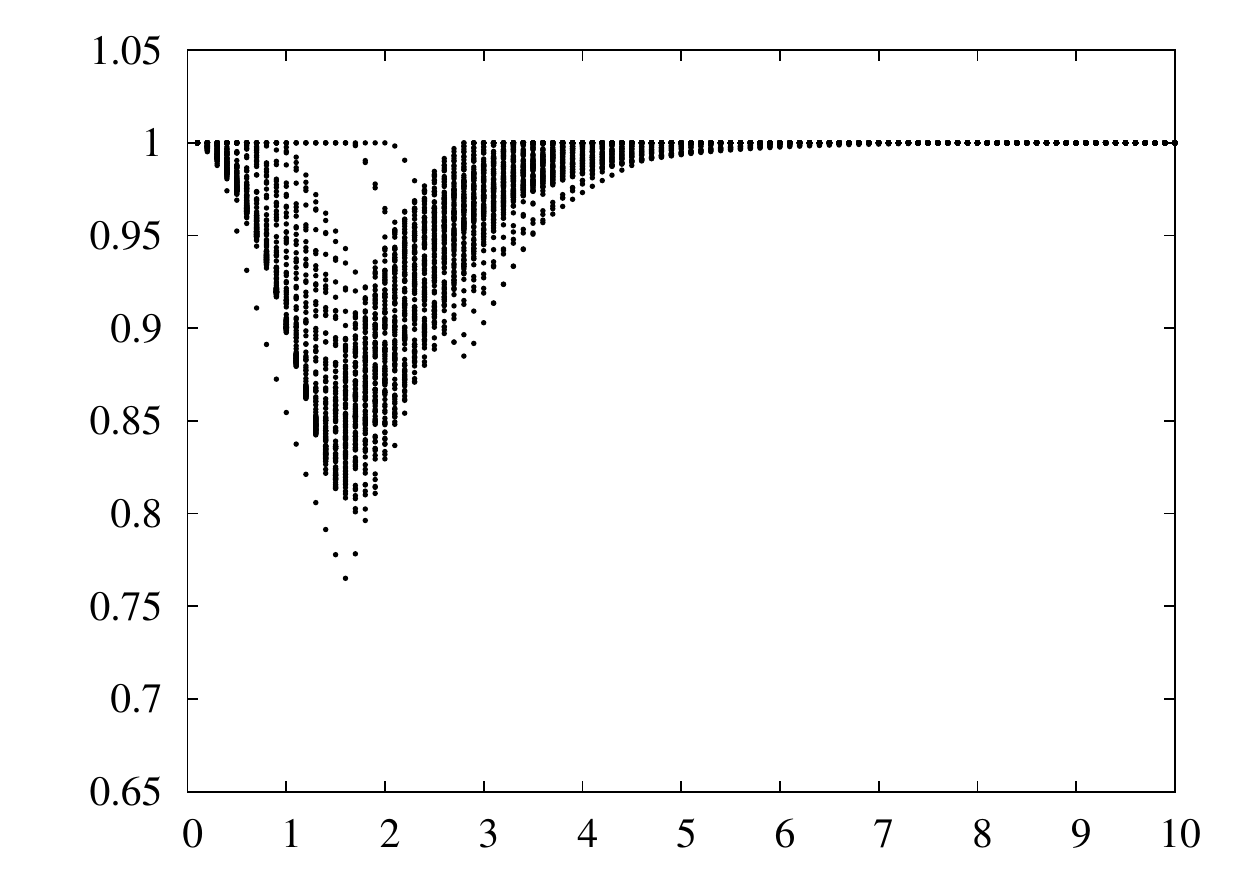}
	\put(-200,144){$\frac{I\left(\X;\Y\right)}{OPT}$}
	\put(-100,-5){$\varepsilon$}
	\put(-105,146){$|\cX|=6$}
	\put(-105,160){\phantom{a}}
	\includegraphics[width=.47\textwidth]{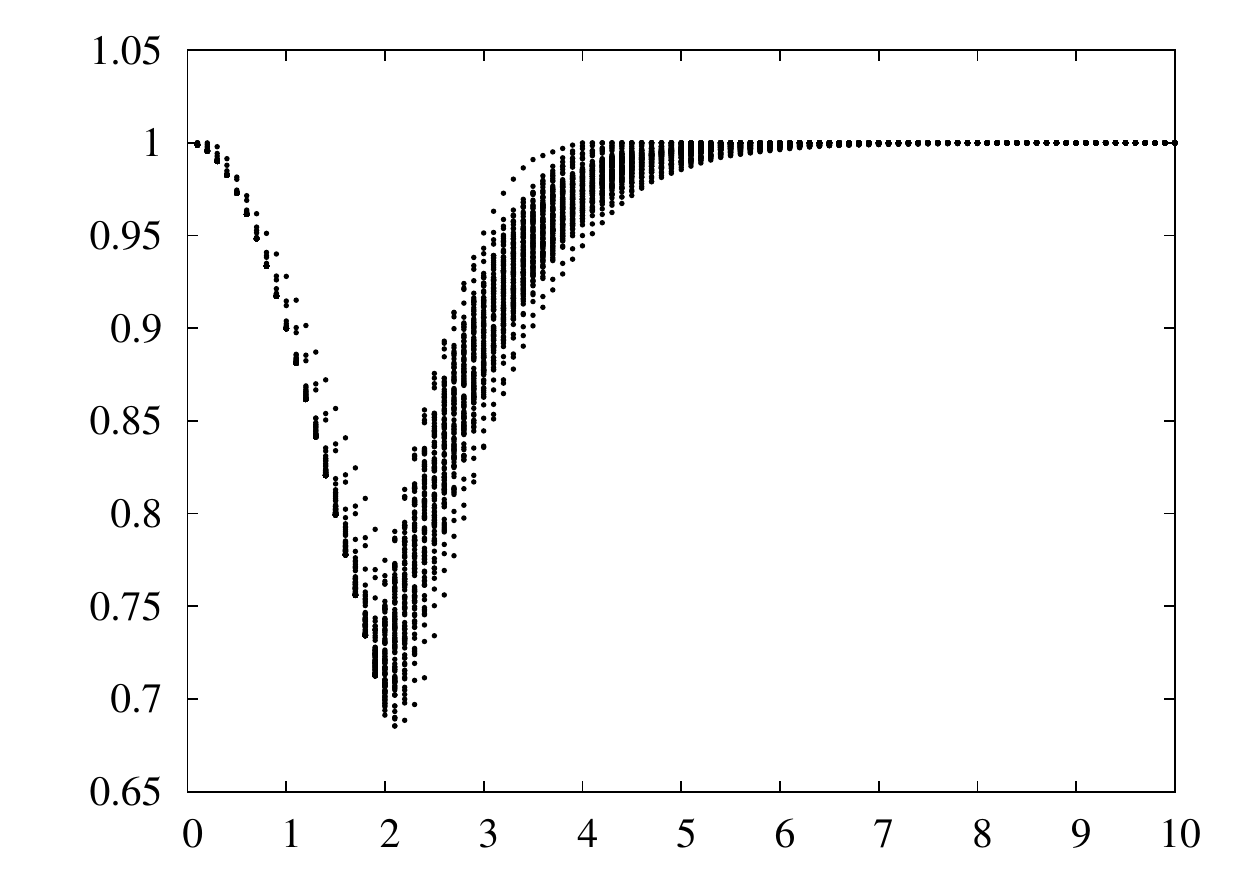}
	\put(-200,144){$\frac{I\left(\X;\Y\right)}{OPT}$}
	\put(-100,-5){$\varepsilon$}
	\put(-105,146){$|\cX|=12$}
	\end{center}
	\caption{For varying input alphabet size $|\cX|\in\{3,4,6, 12\}$, at least $65\%$ of the maximum $I\left(\X;\Y\right)$ can be achieved by taking the better one between the binary and the randomized response mechanisms.}
	\label{fig:maxI}
\end{figure}

\subsection{Lower bounds}
\label{sec:estlb}
In this section, we provide converse results on the fundamental limit of locally differentially private mechanisms
when utility is measured via mutual information.
\begin{corollary}
\label{coro:esthigh}
For any $\varepsilon\geq 0$, let $\Q$ be any conditional distribution that guarantees $\varepsilon$-local differential privacy.
Then, for any distribution $\PP$ and any positive $\delta>0$,
there exists a positive $\varepsilon^*$ that depends on $\PP$ and $\delta$ such that for any $\varepsilon\leq \varepsilon^*$ the following bound holds
\begin{eqnarray}
	I\left(\X;\Y\right) &\leq& (1+\delta)\frac{1}{2} \PP\left(\T\right)\PP\left(\T^c\right)\varepsilon^2,
	\label{eq:converseIlow}
\end{eqnarray}
where $\T$ is defined in \eqref{eq:tbin}.
\end{corollary}
This follows from Theorem \ref{thm:estbin} (optimality of the binary mechanism) and observing that the binary mechanism achieves
 \begin{eqnarray*}
I\left(X;Y\right)&=&\frac{1}{e^{\varepsilon}+1}\left\{P\left(\T \right)e^{\varepsilon}\log\frac{e^{\varepsilon}}{P\left(\T^c\right)+e^{\varepsilon}P\left(\T \right)}+P\left(\T^c\right)\log\frac{1}{P\left(\T^c\right)+e^{\varepsilon}P\left(\T \right)}\right\}   \nonumber \\
&&+\frac{1}{e^{\varepsilon}+1}\left\{P\left(\T^c\right)e^{\varepsilon}\log\frac{e^{\varepsilon}}{P\left(\T \right)+e^{\varepsilon}P\left(\T^c\right)}+P\left(\T \right)\log\frac{1}{P\left(\T \right)+e^{\varepsilon}P\left(\T^c\right)}\right\} \nonumber \\
&=&  \frac{1}{2} \PP\left(\T\right)\PP\left(\T^c\right)\varepsilon^2 + O\left(\varepsilon^3\right).
\end{eqnarray*}
Similarly, in the low privacy regime, we can show the following converse result.
\begin{corollary}
\label{coro:est_low}
For any $\varepsilon\geq 0$, let $Q$ be any conditional distribution that guarantees $\varepsilon$-local differential privacy.
Then, for any distribution $\PP$ and any positive $\delta>0$,
there exists a positive $\varepsilon^*$ that depends on $\PP$ and $\delta$ such that for any $\varepsilon\geq \varepsilon^*$ the following bound holds
\begin{eqnarray*}
	I\left(\X;\Y\right) &\leq& H\left(\X\right)  - (1-\delta)\left(\kk-1\right)\varepsilon e^{-\varepsilon}.
\end{eqnarray*}
\end{corollary}
This follows directly from Theorem \ref{thm:estrr} (optimality of the randomized response mechanism)
and observing that the randomized response mechanism achieves
\begin{equation}
\label{eq:rr_corII}
I\left(\X;\Y\right) = H\left(\X\right) - \left(\kk-1\right)\varepsilon e^{-\varepsilon} + O( e^{-2\varepsilon}).
\end{equation}

Figure \ref{fig:scatterI} illustrates the gap between the mutual information achieved under
the geometric and
the optimal mechanisms (the binary mechanism for the high privacy regime and the randomized response mechanism for the low privacy regime).
For each instance of the $100$ randomly generated $\PP$ over an input of size $k=6$,
we plot the resulting mutual information $I\left(\X;\Y\right)$
as a function of $\PP\left(\T \right)\PP\left(\T^c\right)$ for $\varepsilon=0.1$,
and as a function of $H\left(\X\right)$ for $\varepsilon=10$.
The binary and the randomized response mechanisms exhibit the scaling predicted by Equations \eqref{eq:converseIlow} and \eqref{eq:rr_corII}, respectively.
\begin{figure}[tb!]
	\begin{center}
	\includegraphics[width=.47\textwidth]{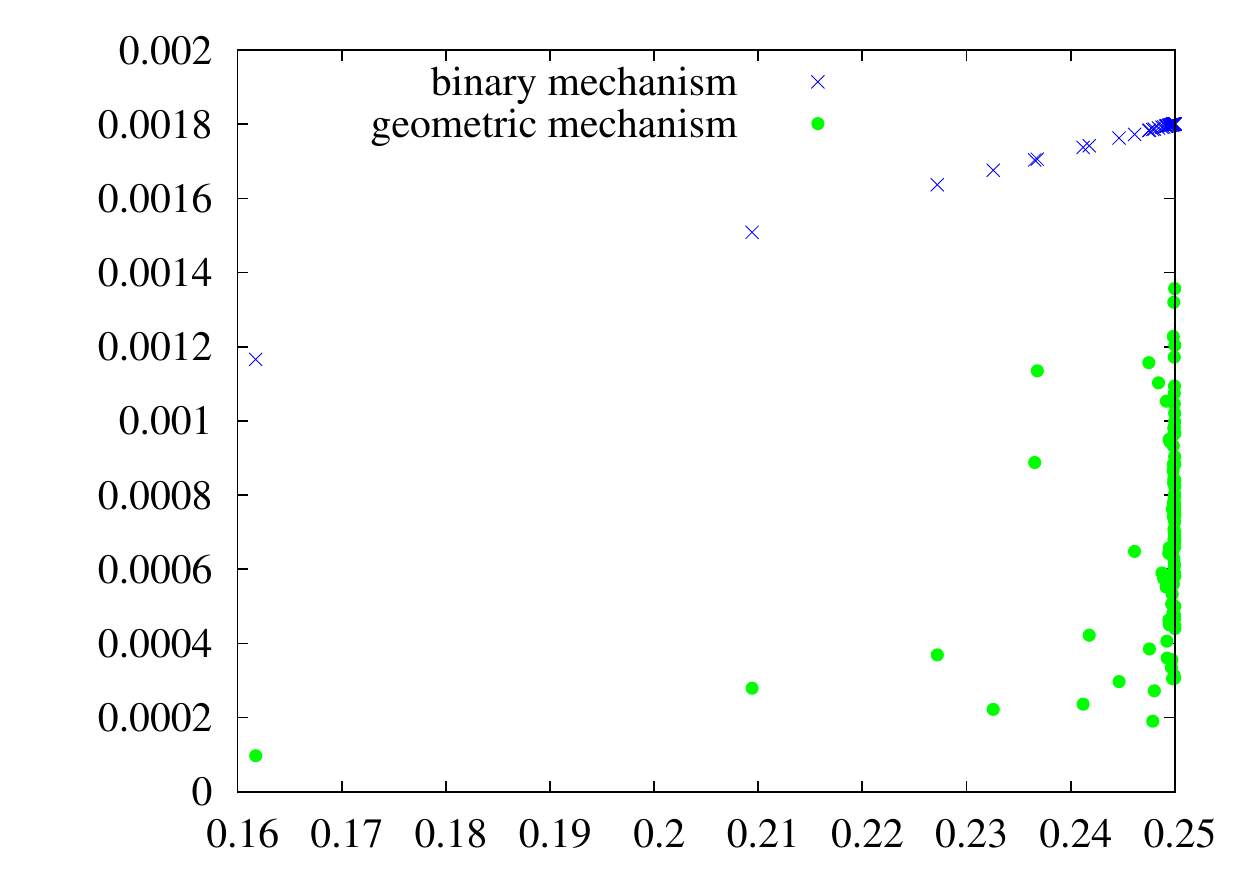}
	\put(-200,144){$I\left(\X;\Y\right)$}
	\put(-120,-10){$\PP\left(\T \right)\PP\left(\T^c\right)$}
	\includegraphics[width=.47\textwidth]{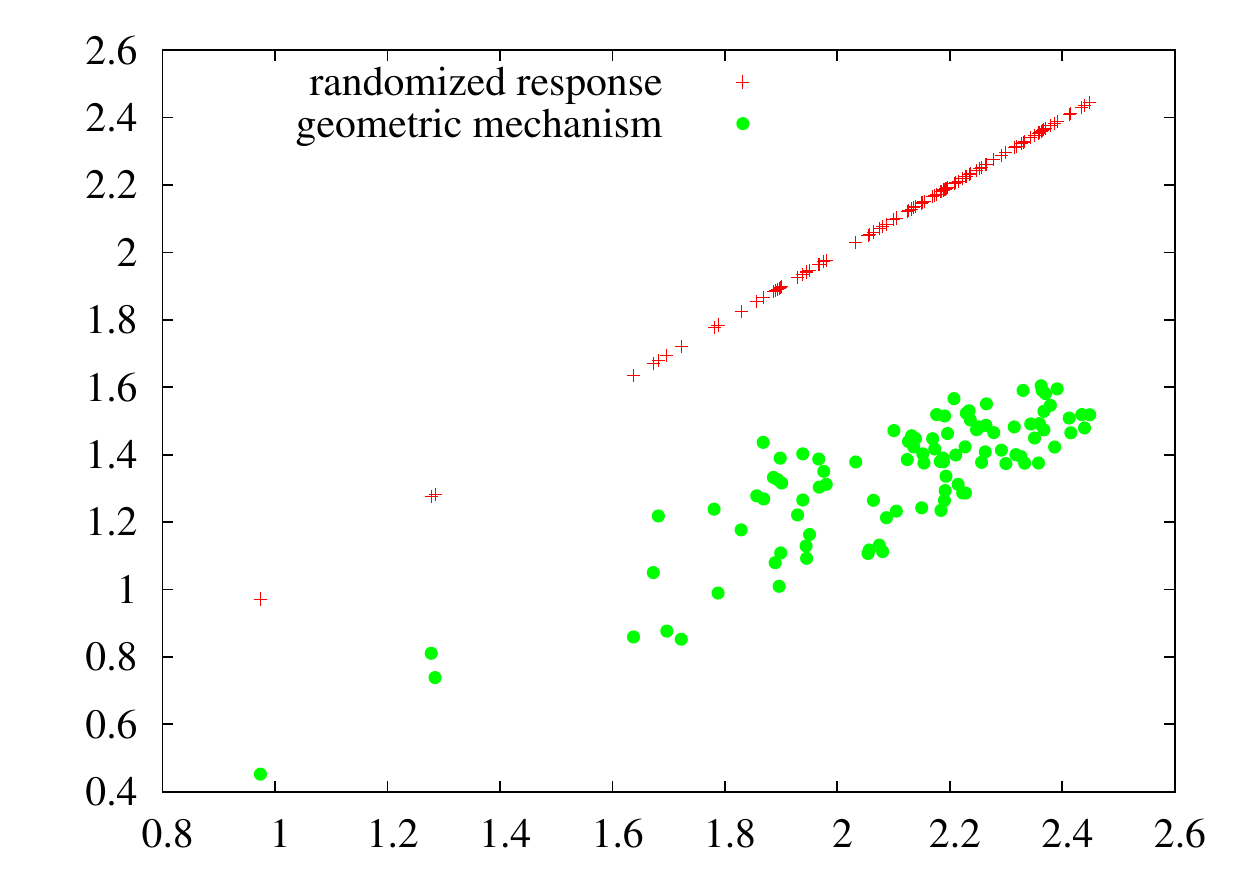}
	\put(-200,144){$I\left(\X;\Y\right)$}
	\put(-120,-10){$H\left(\X\right)$}
	\end{center}
	\caption{For $\varepsilon=0.1$ (left) the binary mechanism achieves the maximum $I\left(\X;\Y\right)$, which scales as Equation \eqref{eq:converseIlow}. For $\varepsilon=10$ (right) the randomized response mechanism achieves the optimal mutual information, which scales as Equation \eqref{eq:rr_corII}.}
	\label{fig:scatterI}
\end{figure}

\section{Generalizations to approximate differential privacy}
\label{sec:general}

In this section, we generalize the results of the previous sections in the
following ways.

\begin{enumerate}

\item We consider the class of utility functions that obey the data
    processing inequality. Consider the composition of two privatization mechanisms $QW =Q\circ W$ where
    the output of the first mechanism $Q$ is applied to another mechanism $W$.
    We say that a utility function $U(\cdot)$ obeys the data processing inequality if the following inequality holds for all $Q$ and $W$
    \begin{eqnarray*}
    	U(QW)& \leq& U(Q) \;.
    \end{eqnarray*}
    The following proposition, proved in Section
    \ref{sec:dp_prop}, shows that the class of utilities obeying the data processing inequality includes all the utility
    functions we studied in Section \ref{sec:result}.
    \begin{proposition}
\label{dp_prop} Any utility function that can be written in the form of $U\left(\Q\right)=\sum_{\cY} \mu(\Q_\y)$, where $\mu$ is any sublinear
function, obeys the data processing inequality.
\end{proposition}

\item We consider $\left(\varepsilon,\delta\right)$-differential
    privacy which generalizes the notion of $\varepsilon$-differential privacy. $\left(\varepsilon,\delta\right)$-differential
    privacy is commonly referred to as approximate differential privacy and it was first introduced in \cite{DKM06}. For the
    release of a random variable $\X \in \cX$, we say that a mechanism $\Q$
    is $\left(\varepsilon,\delta\right)$-locally differentially private if
    \begin{equation}
    \label{eq:approx_diff_privacy}
    \Q\left(S| \x\right) \leq e^{\varepsilon} \Q\left(S| \x'\right) + \delta,
    \end{equation}
    for all $S \subseteq \cY$ and all $\x, \x' \in \cX$.
Note that $\varepsilon$-local differential privacy is a special case of
$\left(\varepsilon,\delta\right)$-local differential privacy where
$\delta=0$.

\item We prove that the {\em quaternary mechanism}, defined in Equation
    \eqref{eq:quat_mech}, is optimal for any $\varepsilon$ and any $\delta$.
    This is different from the treatment conducted in the previous sections
    where we proved the optimality of the binary (randomized response)
    mechanism for sufficiently small (large) $\varepsilon$ and $\delta =
    0$.
\end{enumerate}
The treatment in this section, even though more general than the one in
previous sections in the ways described above, holds only for
binary alphabets (i.e., $|\cX|=2$).
Finding optimal privatization mechanisms under
$(\varepsilon,\delta)$-local differential privacy for larger input alphabets (i.e., $|\cX| > 2$) is an interesting open question.
Unlike $\varepsilon$-local differential privacy, the privacy constraints under $(\varepsilon,\delta)$-local differential privacy
no longer decompose into separate constraints on each output $y$. This makes it difficult to generalize the techniques developed in previous sections of this paper.
However, for the special case of binary input alphabets, we can prove the optimality of one mechanism for all values of $(\varepsilon,\delta)$ and all utility functions that obey the data processing inequality.

For a binary random variable $\X \in \cX = \{0,1\}$, the {\em quaternary
mechanism} maps $\X$ to a quaternary random variable $\Y \in \cY = \{0, 1, 2,
3\}$ and is defined as
\begin{eqnarray}
\label{eq:quat_mech}
\Q_{\rm QT}(0|\x) \,=\, \left\{
\begin{array}{rl}
	\delta & \text{ if } \x= 0\;,\\
	0 & \text{ if } \x = 1\;.\\
\end{array}
\right.\;\;\;
\Q_{\rm QT}(1|\x) \,=\, \left\{
\begin{array}{rl}
	0 & \text{ if } \x = 0\;,\\
	\delta & \text{ if } \x = 1\;. \\
\end{array}
\right.
\end{eqnarray}

\begin{eqnarray*}
\Q_{\rm QT}(2|\x) \,=\, \left\{
\begin{array}{rl}
	(1-\delta)\frac{1}{1+e^\varepsilon}& \text{ if } \x = 0\;,\\
	(1-\delta)\frac{e^\varepsilon}{1+e^\varepsilon}& \text{ if } \x = 1\;. \\
\end{array}
\right.\;\;\;
\Q_{\rm QT}(3|\x) \,=\, \left\{
\begin{array}{rl}
	(1-\delta)\frac{e^\varepsilon}{1+e^\varepsilon}& \text{ if } \x = 0\;,\\
	(1-\delta)\frac{1}{1+e^\varepsilon}& \text{ if } \x = 1\;. \\
\end{array}
\right.
\end{eqnarray*}
%
%
In other words, the quaternary mechanism passes $X$ unchanged with probability $\delta$ and applies the binary mechanism (defined in previous sections) with probability $1-\delta$. The main result of this section can be stated formally as follows.

\begin{theorem}
\label{thm:quat_optmlty}
If $|\cX|=2$, then for any $\varepsilon$, any $\delta$, and any utility $U\left(\Q\right)$ that obeys the data
processing inequality, the quaternary mechanism
maximizes $U\left(\Q\right)$ subject to $\Q \in
\mathcal{\D}_{(\varepsilon,\delta)}$, the set of all
$(\varepsilon,\delta)$-locally differentially private mechanism.
\end{theorem}

The proof of Theorem \ref{thm:quat_optmlty} depends on an {\em operational definition} of differential privacy which we describe next.
Consider a privatization mechanism $\Q$ that maps $\X \in \{0,1\}$ stochastically to $\Y \in \cY$. Given $\Y$, construct a binary hypothesis test on whether $\X=0$ or $\X=1$. Any binary hypothesis test is completely described by a, possibly randomized, decision rule $\hat{\X}: \Y \rightarrow \{0,1\}$. The two types of error associated with $\hat{\X}$ are {\em false alarm:} $\hat{\X}=1$ when $\X=0$, and {\em miss detection:} $\hat{\X}=0$ when $\X=1$. The probability of false alarm is given by $P_{\rm FA} = \prob(\hat{\X}=1|\X=0)$ while the probability of miss detection is given by $P_{\rm MD} = \prob(\hat{\X}=0|\X=1)$. For a fixed $\Q$, the convex hull of all pairs $(P_{\rm MD},P_{\rm FA})$ for all decision rules $\hat{\X}$ defines a two-dimensional \emph{error region} where $P_{\rm MD}$ is plotted against $P_{\rm FA}$. For example, the quaternary mechanism given in Figure \ref{fig:priv_mech} has an error region $\mathcal{R}_{\Q_{\rm QT}}$ shown in Figure \ref{fig:error_quat}.

\begin{figure}
    \centering
    \begin{subfigure}[b]{0.5\textwidth}
	\begin{center}
        \includegraphics[width=.4\textwidth]{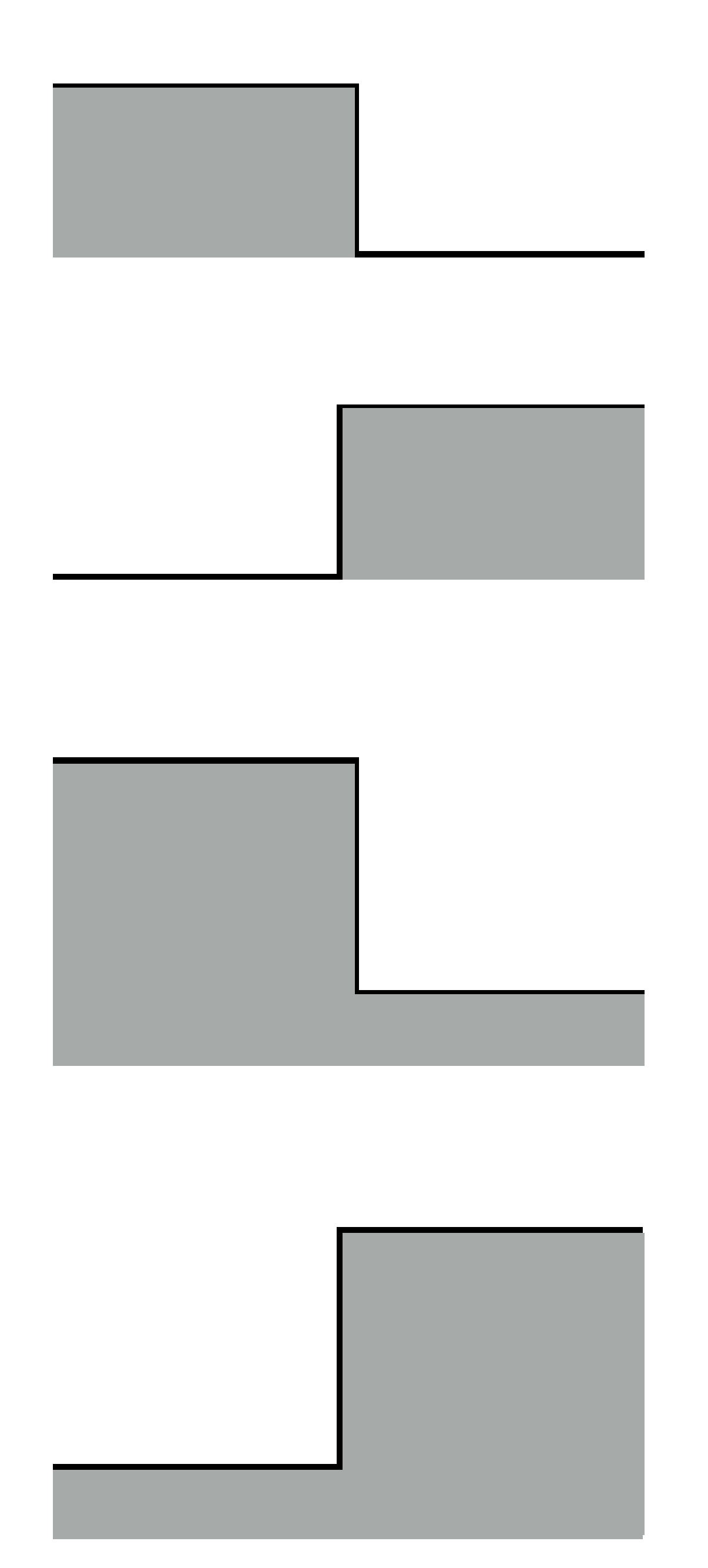}
        \put(-68,170){$\delta$}
		\put(-30,130){$\delta$}
		\put(-68,130){$0$}
		\put(-30,170){$0$}
		\put(-76,80){$\frac{(1-\delta)e^\varepsilon}{1+e^\varepsilon}$}
		\put(-35,80){$\frac{(1-\delta)}{1+e^\varepsilon}$}
		\put(-40,22){$\frac{(1-\delta)e^\varepsilon}{1+e^\varepsilon}$}
		\put(-76,22){$\frac{(1-\delta)}{1+e^\varepsilon}$}
		\put(-76,-10){$x=0$}
		\put(-40,-10){$x=1$}
		\put(-115,22){$y=3$}
		\put(-115,80){$y=2$}
		\put(-115,130){$y=1$}
		\put(-115,170){$y=0$}
	\end{center}
        \caption{Privatization mechanism}
        \label{fig:priv_mech}
    \end{subfigure}%
    ~
    \begin{subfigure}[b]{0.5\textwidth}
	\begin{center}
		\includegraphics[width=0.9\textwidth]{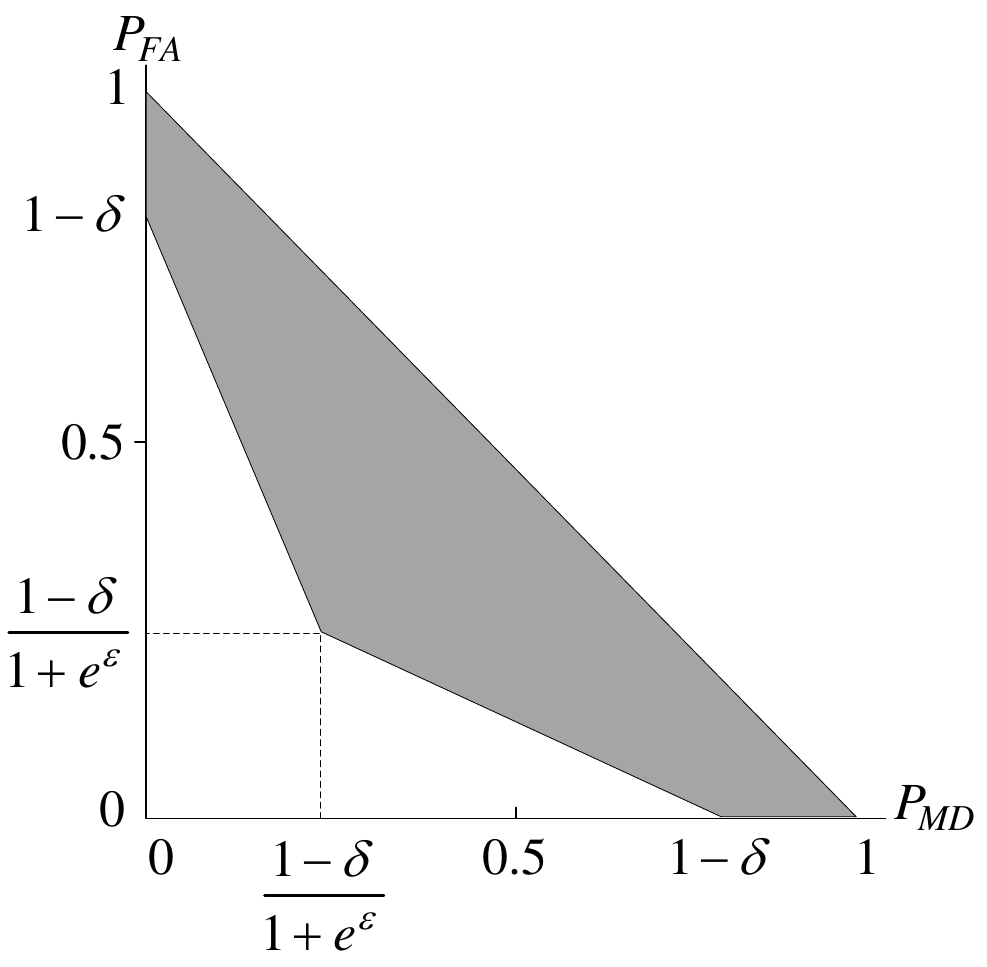}
\put(-135,70){$\mathcal{R}_{\Q_{\rm QT}}= \mathcal{R}_{\varepsilon,\delta}$}

\end{center}
\caption{Error region}
\label{fig:error_quat}
    \end{subfigure}
    \caption{The quaternary mechanism}\label{fig:quat_mech}
\end{figure}
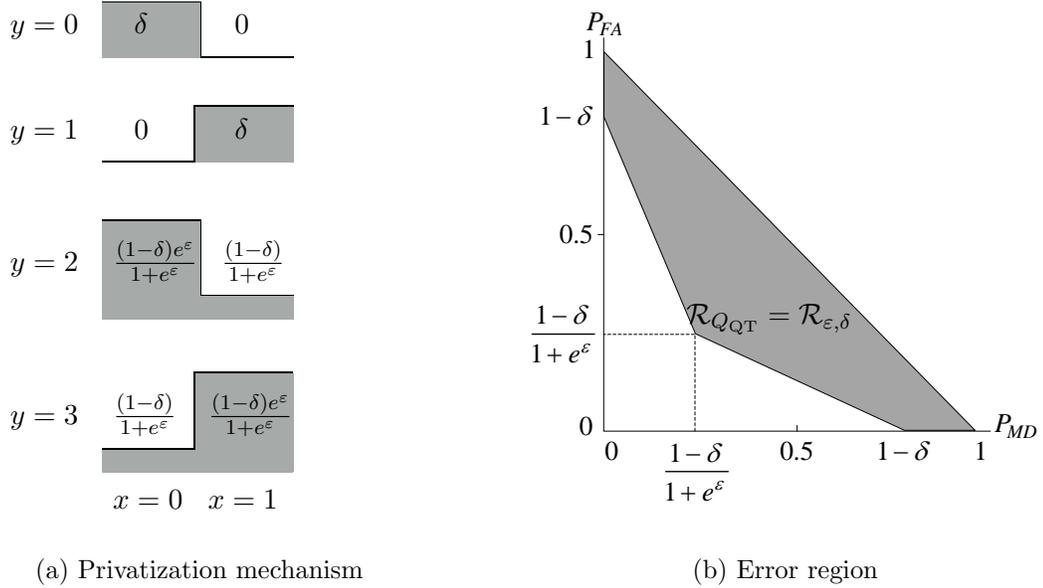

It turns out that $\left(\varepsilon,\delta\right)$-local differential privacy imposes the following conditions on the error region of all  $\left(\varepsilon,\delta\right)$-locally differentially private mechanisms
\begin{eqnarray*}
 	P_{\rm FA} + e^\varepsilon P_{\rm MD} &\geq & 1-\delta \;,\;\;\text{ and } \;\;\;\;  e^\varepsilon P_{\rm FA} + P_{\rm MD}\;\geq\; 1-\delta\;,
\end{eqnarray*}
for any decision rule $\hat{\X}$. These two conditions define an error region $\mathcal{R}_{\varepsilon,\delta}$ shown in Figure \ref{fig:error_quat}. Interestingly, the next theorem shows that the converse result is also true.
\begin{theorem}
\label{thm:op_diff_prvcy}
A mechanism $\Q$ is $\left(\varepsilon,\delta\right)$-locally differentially private if and only if $\mathcal{R}_{\Q} \subseteq \mathcal{R}_{\varepsilon,\delta}$.
\end{theorem}
The proof of the above theorem can be found in \cite{KOV14b}. Notice that it is no coincidence that $\mathcal{R}_{\Q_{\rm QT}}=\mathcal{R}_{\varepsilon,\delta}$. This property will be essential in proving the optimality of the quaternary mechanism.

Theorem \ref{thm:op_diff_prvcy} allows us to benefit from the data processing inequality (DPI) and its converse, which follows from a celebrated result by \cite{Bla53}. These inequalities, while simple by themselves, lead to surprisingly strong technical results. Indeed, there is a long line of such a tradition in the information theory literature (see Chapter 17 of \cite{CT12}). Consider two privatization mechanisms, $Q^{(1)}$ and $Q^{(2)}$. Let $Y$ and $Z$ denote the output of the mechanisms $Q^{(1)}$ and $Q^{(2)}$, respectively. We say that $Q^{(1)}$ dominates $Q^{(2)}$ if there exists a coupling of $Y$ and $Z$ such that $X \text{--} Y \text{--}Z$ forms a Markov chain. In other words, we say $Q^{(1)}$ dominates $Q^{(2)}$ if there exists a stochastic mapping $Q$ such that $Q^{(2)} = Q^{(1)} \circ \Q$.
\begin{theorem}
\label{thm:dpi}
A mechanism $\Q^{(1)}$ dominates a mechanism $\Q^{(2)}$ if and only if $\mathcal{R}_{\Q^{(2)}} \subseteq \mathcal{R}_{\Q^{(1)}}$.
\end{theorem}
The proof of the above theorem can be found in \cite{Bla53}. Observe that by Theorems \ref{thm:dpi} and \ref{thm:op_diff_prvcy}, and the fact that $\mathcal{R}_{\Q_{\rm QT}}=\mathcal{R}_{\varepsilon,\delta}$, the quaternary mechanism dominates any other differentially private mechanism. In other words, for any differentially private mechanism $\Q$, there exists a stochastic mapping $W$ such that $\Q = W \circ \Q_{\rm QT}$. Therefore, for any $(\varepsilon,\delta)$ and any utility function $U(.)$ obeying the data processing inequality, we have that $U(\Q) \leq U(\Q_{\rm QT})$. This finishes the proof of Theorem \ref{thm:quat_optmlty}.

\section{Discussion}
\label{sec:discussion}

In this paper, we have considered a broad class of convex utility functions and assumed a setting where individuals cannot collaborate (communicate with each other) before releasing their data. It turns out that the techniques developed in this work can be generalized to find optimal privatization mechanisms in a setting where different individuals can collaborate interactively and each individual can be an analyst \citep{KOV14a}.

 Binary hypothesis testing and information preservation are two canonical problems with a wide range of applications. However, there are a number of non-trivial and interesting extensions to our work.

 \bigskip
 \noindent{\bf Correlation among data.}
In some scenarios the $\X_i$'s could be correlated (e.g., when different individuals observe different functions of the same random variable). In this case, the data analyst is interested in inferring whether the data was generated from $\PP^{n}_0$ or $\PP^{n}_1$, where $\PP^{n}_\nu$ is one of two possible joint priors on $\X_1,...,\X_n$. This is a challenging problem because knowing $\X_i$ reveals information about $\X_j$, $j \neq i$. Therefore, the utility maximization problems for different individuals are coupled in this setting.

 \bigskip
 \noindent{\bf Robust and $m$-ary hypothesis testing.}
In some cases the data analyst need not have access to $\PP_0$ and $\PP_1$, but rather to two classes of prior distribution $\PP_{\theta_0}$ and $\PP_{\theta_1}$ for $\theta_0 \in \Lambda_0$ and $\theta_1 \in \Lambda_1$. Such problems are studied under the rubric of universal hypothesis testing and robust hypothesis testing. One possible direction is to select the privatization mechanism that maximizes the worst case utility:  $\Q^* = \arg \max_{\Q \in \mathcal{\D}_{\varepsilon}} \min_{\theta_0 \in \Lambda_0,\theta_1 \in \Lambda_1} D_f(\MM_{\theta_0}||\MM_{\theta_1})$, where $\MM_{\theta_\nu}$ is the induced marginal under $\PP_{\theta_\nu}$.

 The more general problem of private $m$-ary hypothesis testing is also an interesting but challenging one. In this setting, the $\X_i$'s can follow one of $m$ distributions $\PP_0$, $\PP_1$, ..., $\PP_{m-1}$. Consequently, the $\Y_i$'s can follow one of $m$ distributions $\MM_0$, $\MM_1$, ..., $\MM_{m-1}$. In this case, the utility can be defined as the average $f$-divergence between any two distributions: $1/(m(m-1)) \sum_{i \neq j}D_f(\MM_{i}||\MM_{j})$, or the worst case one: $\min_{i \neq j} D_f(\MM_{i}||\MM_{j})$.

\bigskip
\noindent
{\bf Non-exchangeable utility functions.}
The utility studied in this paper was measured by functions that are exchangeable, i.e.
the utility does not depend on the naming (labelling) or topology of the private and privatized data ($X$ and $Y$).
This made sense for statistical learning applications that depend on information theoretic quantities such as $f$-divergences and mutual information.
However, in some other applications, the utility might be defined over $\cX \cup \cY$ in a metric space,
where there exists a natural measure of distance (or distortion) between the data points.
In this case, we can formulate the problem as a distortion minimization one
\begin{eqnarray*}
	\displaystyle \text{minimize}_{Q\in {\cal D}_\varepsilon } &&  \sum_{x,y} d(x,y) P(x) Q(y|x)\;,
\end{eqnarray*}
where $d(\x,\y)$ is some distortion metric.
\cite{WYZ14} studied this problem, and showed that the mechanism
$Q(y|x) \propto  e^{\varepsilon(1-d(x,y))}/(k-1+e^{\varepsilon})$ achieves near optimal
performance when $\varepsilon$ is large enough, which is the low privacy regime.
Notice that when Hamming distance is used, $d(x,y)={\mathbb I}(x\neq y)$,
this recovers the randomized response mechanism exactly.
This provides a starting point for generalizing the search for optimal mechanisms
under non-exchangeable utility functions.

%
%
\section{Proof of Theorems \ref{thm:sc} and \ref{thm:lp}}
We start the proof with a few definitions, a lemma, and a general result that applies to any convex utility function that obeys a mild assumption.

Recall that for an input alphabet $\cX$ with $|\cX|=\kk$, we represent the set of
$\varepsilon$-locally differentially private mechanisms that lead to output
alphabets $\cY$ with $|\cY|=\n$ by $\mathcal{\D}_{\varepsilon,\n}$. The set of all $\varepsilon$-locally differentially private mechanisms is given by $\mathcal{\D}_{\varepsilon}= \cup_{\n \in \mathbb{N}} \mathcal{\D}_{\varepsilon,\n}$. A utility function $U\left(\Q\right)$ is convex in $\Q$ if
$U\left(\lambda\Q^{(1)}+\left(1-\lambda\right)\Q^{(2)}\right) \leq \lambda
U\left(\Q^{(1)}\right)+\left(1-\lambda\right)U\left(\Q^{(2)}\right)$ for any
$\lambda \in \left(0,1\right)$. Convex utility functions are ubiquitous in
information theory and statistics.
\begin{assumption}
\label{asmptn:well_behaved} If a $\kk \times \n$ privatization mechanism
$\Q^{(1)} \in \mathcal{\D}_{\varepsilon,\n}$ is obtained by deleting an all-zero column
of a $\kk \times \n +1$ privatization mechanism $\Q^{(2)} \in
\mathcal{\D}_{\varepsilon,\n+1}$, then $U\left(\Q^{(1)}\right)=U\left(\Q^{(2)}\right)$.
\end{assumption}
Naturally, one would expect that if we delete the zero columns of a
privatization mechanism $\Q^{(2)}$ to obtain a new privatization mechanism
$\Q^{(1)}$, we would still get the same utility. This is because a
``reasonable'' utility function should not depend on output alphabets with
zero probability.

\begin{theorem}
\label{thm:optmlty_randmzd_rspns} If $U\left(\Q\right)$ is a convex utility
function that satisfies Assumption \ref{asmptn:well_behaved}, then the following holds
\begin{equation}
\max_{\Q \in \mathcal{\D}_{\varepsilon}}
U\left(\Q\right)= \max_{\Q \in \cup_{\n=1}^{\kk} \mathcal{\D}_{\varepsilon,\n}} U\left(\Q\right).
\end{equation}
\end{theorem}
This theorem implies that among all $\varepsilon$-locally differentially private mechanisms, we only need to consider those that lead to output alphabets of size $\n \leq \kk$. In other words, enlarging the input alphabet cannot further maximize the utility. The proof of Theorem \ref{thm:optmlty_randmzd_rspns} is given in Section \ref{subsec:polytope}.

\begin{lemma}
\label{lemma:diff_priv_reps} A $\kk\times \n$ conditional distribution $\Q$ is $\varepsilon$-locally differentially private
 if and only if it can be written as
$\Q=\s\B$, where $\s$ is a $\kk \times \n$ matrix with $S_{ij} \in
[1,e^\varepsilon]$ and $\B=\mbox{diag}\left(\bb_1,\ldots,\bb_\n\right)$ with its
diagonal entries in $\mathbb{R}_{+}$.
\end{lemma}
The proof of Lemma \ref{lemma:diff_priv_reps} is provided in Section \ref{subsec:dp_rep}. With the above results, we are now ready to prove Theorems \ref{thm:sc} and \ref{thm:lp}. By Lemma \ref{lemma:diff_priv_reps}, for any $\Q \in \mathcal{\D}_{\varepsilon,\n}$ we have that $\Q_j=\bb_j\s_j$. Suppose $U\left(\Q\right)=\sum_{j\in [\n]} \mu(Q_j)$, where $\mu$ is a sublinear function. Since $\mu$ is sublinear, it is convex and
$\mu\left(\bb_j\s_j\right)=\bb_j\mu\left(\s_j\right)$. $U\left(\Q\right)$ is convex in $\Q$ because
\begin{eqnarray}
U\left(\lambda\Q^{(1)}+\left(1-\lambda\right)\Q^{(2)}\right)&=&
\sum_{j\in [\n]}\mu\left(\lambda\bb^{(1)}_j\s^{(1)}_j+\left(1-\lambda\right)\bb^{(2)}_j\s^{(2)}_j\right) \nonumber \\
&\leq& \sum_{j\in [\n]}\lambda\mu\left(\bb^{(1)}_j\s^{(1)}_j\right)+\left(1-\lambda\right)\mu\left(\bb^{(2)}_j\s^{(2)}_j\right)\nonumber \\
&=& \lambda U\left(\Q^{(1)}\right) + \left(1-\lambda\right) U\left(\Q^{(2)}\right),
\end{eqnarray}
for any $\lambda \in (0,1)$. Furthermore, $U\left(\Q\right)$ satisfies
Assumption \ref{asmptn:well_behaved} because $\mu\left(\Q_j\right)=0$
whenever $\bb_j=0$. Let $\Q^*=\s^*\B^*\in \arg \max_{\Q \in \cup_{\n=1}^{\kk} \mathcal{\D}_{\varepsilon,\n}}
U\left(\Q\right)$ and note that by
Theorem \ref{thm:optmlty_randmzd_rspns}, $U\left(\Q^*\right)= \max_{\Q \in
\mathcal{\D}_{\varepsilon}} U\left(\Q\right)$. Suppose that $\Q^*$ is of dimensions $\kk
\times \n$, where $\n \leq \kk$. Each of the $\n$ columns of $\Q^*$ can be
expressed as a convex combination of the columns of $\s^{(\kk)}$, the staircase pattern matrix. This is
because the $2^\kk$ columns of $\s^{(\kk)}$ are the corner points of the cube
$[1,e^\varepsilon]^\kk$ and each $\s^*_j \in [1,e^\varepsilon]^\kk$.
Therefore, $\s^*_j = \sum_{i=1}^{2^{\kk}} \lambda_{ij} \s^{(\kk)}_{i}$, where
$\lambda_{ij} \geq 0$ for all $i$ and $j$, and $\sum_{i=1}^{2^{\kk}}
\lambda_{ij}=1$ for all $j$. Create the $2^{\kk}$-dimensional vector
$\tilde{\bb}$ such that $\tilde{\bb}_i = \sum_{j=1}^{\n}
\lambda_{ij}\bb^*_j$ and let $\tilde{\Q}=\s^{(\kk)}\tilde{\B}$.
\begin{eqnarray}
U\left(\Q^*\right)-U(\tilde{\Q})&=& \sum_{j=1}^{\n}\mu\left(\s^*_j\right)\bb^*_j - \sum_{i=1}^{2^\kk}\mu\left(\left(\sum_{j=1}^{\n}
\lambda_{ij}\bb^*_j\right)\s^{(\kk)}_j\right) \nonumber \\
&=& \sum_{j=1}^{\n}\mu\left(\sum_{i=1}^{2^{\kk}} \lambda_{ij} \s^{(\kk)}_{i}\right)\bb^*_j - \sum_{i=1}^{2^\kk}\sum_{j=1}^{\n}
\lambda_{ij}\bb^*_j\mu\left(\s^{(\kk)}_j\right) \nonumber \\
&=& \sum_{j=1}^{\n}\bb^*_j\left\{\mu\left(\sum_{i=1}^{2^{\kk}} \lambda_{ij} \s^{(\kk)}_{i}\right) - \sum_{i=1}^{2^\kk}
\lambda_{ij}\mu\left(\s^{(\kk)}_j\right)\right\} \nonumber \\
&\leq& 0,
\end{eqnarray}
by the convexity of $\mu\left(z\right)$ and the non-negativity of
$\bb^*_j$'s. Moreover, observe that since $\s^{(\kk)}\tilde{\bb} =\unity$,
$\tilde{\bb}$ is a valid choice for the linear program of \eqref{eq:optsc}. This implies that
\begin{equation}
\max_{\s^{(\kk)}\bb = \unity, \bb \geq 0}
\sum_{j=1}^{2^{\kk}}\mu\left(\s^{(\kk)}_j\right)\bb_j \geq U(\tilde{\Q}) \geq U\left(\Q^*\right) = \max_{\Q \in \mathcal{\D}_{\varepsilon}}
U\left(\Q\right)
\end{equation}
On the other hand, for any $\tilde{\Q}=\s^{(\kk)}\tilde{\B}$, where $\tilde{\bb}$ is valid for the linear program of \eqref{eq:optsc}, we have that
$\tilde{\Q} \in \mathcal{\D}_{\varepsilon,2^\kk} \subset \mathcal{\D}_{\varepsilon}$ and therefore,
$\max_{\s^{(\kk)}\bb = \unity, \bb \geq 0}
\sum_{j=1}^{2^{\kk}}\mu\left(\s^{(\kk)}_j\right)\bb_j \leq \max_{\Q \in \mathcal{\D}}
U\left(\Q\right)$. Thus, $\max_{\s^{(\kk)}\bb = \unity, \bb \geq 0}
\sum_{j=1}^{2^{\kk}}\mu\left(\s^{(\kk)}_j\right)\bb_j = \max_{\Q \in \mathcal{\D}}
U\left(\Q\right)$. This proves Theorem \ref{thm:lp}.

The polytope given by  $\s^{(\kk)}\bb = \unity$ and $\bb \geq 0$ is a closed and bounded one. Thus, the linear program of \eqref{eq:optsc} is bounded and has a solution, say $\bb^*$, at one of the corner points of the polytope. Since there are $\kk$ equality constraints given
by $\s^{(\kk)}\bb = \unity$ and $2^\kk$ inequality constraints given by $\bb
\geq 0$, any corner point, including $\bb^*$, cannot have more than $\kk$
non-zero entries. Form
$\tilde{\s}$ by deleting the columns of $\s^{(\kk)}$ corresponding to zero entries of
$\bb^*$. Similarly, form $\tilde{\bb}$ by deleting the zero entries of
$\bb^*$ and let $\tilde{\Q}=\tilde{\s}\tilde{\B}$, where $\tilde{\B}={\rm diag}\tilde{\bb}$. It is easy to verify that
$U(\tilde{\Q})= U(\Q^*) =\mu^T\bb^*$; hence, $\tilde{\Q}$ solves linear program of \eqref{eq:optsc}. Moreover,
$\tilde{\Q}$ has at most $\kk$ columns and $\tilde{\s}_{ij}=
\{1,e^\varepsilon\}$. Therefore, $\tilde{\Q}$ is a staircase mechanism of dimension $\kk \times \n$, where $\n \leq \kk$.

\subsection{Proof of Theorem \ref{thm:optmlty_randmzd_rspns}}
\label{subsec:polytope}
We start the proof of Theorem \ref{thm:optmlty_randmzd_rspns} with an important lemma the proof of which is presented in Section \ref{subsec:max_non_zero}.
\begin{lemma}
\label{lemma:max_non_zero} The set of all $\kk\times\n$, $\varepsilon$-locally differentially private mechanisms $\mathcal{\D}_{\varepsilon,\n}$ forms a closed and bounded polytope in $\mathbb{R}_{+}^{\kk\n}$. Moreover, let $\Q$ be a corner point of the polytope formed by $\mathcal{\D}_{\varepsilon,\n}$, then $\Q$ has at most $\kk$ non-zero columns.
\end{lemma}

Fix an $\n > \kk$. Since $U\left(\Q\right)$ is convex in $\Q$, it suffices to
consider the corner points of $\mathcal{\D}_{\varepsilon,\n}$ when maximizing
$U\left(\Q\right)$ subject to $\Q \in \mathcal{\D}_{\varepsilon,\n}$. By Lemma
\ref{lemma:max_non_zero}, any $\Q^{(1)}$, a $\kk\times \n$ corner point of
$\mathcal{\D}_{\varepsilon,\n}$, has at most $\kk$ non-zero columns. Therefore, the
privatization mechanism $\Q^{(2)}$, obtained by deleting the all-zero columns
of $\Q^{(1)}$, has at most $\kk$ columns. Notice that $\Q^{(2)} \in \cup_{i=1}^{\kk}\mathcal{\D}_{\varepsilon,i}$. Since $U\left(\Q\right)$ satisfies
Assumption \ref{asmptn:well_behaved}, we have that
$U\left(\Q^{(1)}\right)=U\left(\Q^{(2)}\right)$ and therefore, it suffices to
consider $\Q \in \cup_{i=1}^{\kk}\mathcal{\D}_{\varepsilon,i}$ when maximizing $U\left(\Q\right)$ subject to
$\Q \in \mathcal{\D}_{\varepsilon,\n}$. Thus,
\begin{eqnarray}
\sup_{\Q \in \mathcal{\D}_{\varepsilon}}
U\left(\Q\right)&=& \sup_{\n \in \mathbb{N}}\left\{ \max_{\Q \in \mathcal{\D}_{\varepsilon,\n}}
U\left(\Q\right)\right\} \nonumber \\
&=& \sup_{\n \in \mathbb{N}} \left\{\max_{\Q \in \cup_{i=1}^{\kk}\mathcal{\D}_{\varepsilon,i}}
U\left(\Q\right)\right\} \nonumber \\
&=& \max_{\Q \in \cup_{i=1}^{\kk}\mathcal{\D}_{\varepsilon,i}}
U\left(\Q\right),
\end{eqnarray}
which finishes the proof.

\subsection{Proof of Lemma \ref{lemma:diff_priv_reps}}
\label{subsec:dp_rep}
\begin{claim}
\label{claim:zero_column} Let $\Q \in \mathcal{\D}_{\varepsilon,\n}$. If $\Q_{ij}=0$ for
some $j \in \{1,...,\n\}$ then $\Q_{ij}=0$ for all $i \in \{1,...,\kk\}$.
\end{claim}
\begin{proof}
Assume that $\Q_{i_1j}=0$ and $\Q_{i_2j}\neq0$ for some $i_1,i_2
\in \{1,...,\kk\}$. It is obvious that $\q\left(\y_{j}|\x_{i_2}\right) \leq
\q\left(\y_{j}|\x_{i_1}\right)e^{\varepsilon}$ is not satisfied. Therefore, $\Q$ is not a locally differentially private mechanism.
\end{proof}

It is easy to check that any $\kk \times \n$ stochastic matrix $\Q=\s\B$, where $\B$ is a diagonal matrix with non-negative entries and $\s$ is a $\kk \times \n$ matrix with $\s_{ij} \in [1,e^\varepsilon]$, satisfies the local differential privacy constraints. Thus, $\Q \in \mathcal{\D}_{\varepsilon,\n}$. On the other hand, assume that $\Q \in \mathcal{\D}_{\varepsilon,\n}$. If $\Q_{ij}=0$ for some $j$ then by Claim \ref{claim:zero_column}
we have that $\Q_{ij}=0$ for all $i$ and therefore, we can set $\bb_j=0$ and
$\s_{ij}=1$ for all $i$. If $\Q_{ij} > 0$ then by Claim
\ref{claim:zero_column} we have that $\Q_{ij}>0$ for all $i$. In this case,
let $\bb_j=\min_{i}\Q_{ij}$ and observe that $\bb_j>0$ since $\Q_{ij}>0$ for
all $i$. Let $\s_{ij}=\Q_{ij}/\bb_i$, then it is clear (from the definition
of $\bb_i$) that $\s_{ij}\geq1$. On the other hand, from the differential
privacy constraints, we have that $ \Q_{ij} \leq \Q_{kj}e^{\varepsilon}
\mbox{ for all }k$ and thus, $\Q_{ij} \leq \min_{k}\Q_{kj} e^{\varepsilon}$
which proves that $\s_{ij}= \Q_{ij}/\min_{k}\Q_{kj} \leq e^{\varepsilon}$. This establishes that any  $\Q \in \mathcal{\D}_{\varepsilon,\n}$ can be written as $\Q=\s\B$, where $\B$ is a diagonal matrix with non-negative entries and $\s$ is a $\kk \times \n$ matrix with $\s_{ij} \in [1,e^\varepsilon]$.

\subsection{Proof of Lemma \ref{lemma:max_non_zero}}
\label{subsec:max_non_zero}
We start by showing that $\mathcal{\D}_{\varepsilon,\n}$ forms a closed
and bounded polytope in $\mathbb{R}_{+}^{\kk\n}$. We are interested in studying the corner points of the polytope formed by $\mathcal{\D}_{\varepsilon,\n}$ because
convex utility functions are maximized at one of these corner points whenever the space of privatization mechanisms is restricted to $\mathcal{\D}_{\varepsilon,\n}$.

\begin{claim}
\label{claim:local_dp_discrete} A privatization mechanism $\Q \in
\mathcal{\D}_{\varepsilon,\n}$ if and only if for all $\x,\x'\in\mathcal{\X}$
and all $\y \in \mathcal{\Y}$ we have that $\Q\left(\y|\x\right) \leq
\Q\left(\y|\x'\right)e^{\varepsilon}$.
\end{claim}
\begin{proof}
 By definition, $\Q$ is differentially
private if for all $\x$, $\x' \in\mathcal{\X}$ and all $B\subseteq
\mathcal{\Y}$ we have that $\Q\left(B|\x\right) \leq
\Q\left(B|\x'\right)e^{\varepsilon}$. By choosing
$B=\{\y\}$ for some $\y \in \mathcal{\Y}$ the first direction
of the above lemma is proven. In order to prove the other direction, assume
that for all $\x,\x'\in\mathcal{\X}$ and all $\y \in \mathcal{\Y}$ we
have that $\Q\left(\y|\x\right) \leq
\Q\left(\y|\x'\right)e^{\varepsilon}$. Then for any $B \subseteq
\mathcal{\Y}$, the following holds
\begin{eqnarray}
 \sum_{\y\in B}\Q\left(\y|\x\right) &\leq& \sum_{\y\in B}\Q\left(\y|\x'\right)e^{\varepsilon} \\ \nonumber
\Leftrightarrow ~~\Q\left(B|\x\right) &\leq& \Q\left(B|\x'\right)e^{\varepsilon}.
\end{eqnarray}
\end{proof}

Let $\Q \in \mathcal{\D}_{\varepsilon,\n}$, then by Claim
\ref{claim:local_dp_discrete}, it is easy to see that $\Q$ must satisfy
$\n\kk(\kk-1)$ inequalities of the form $\Q\left(\y|\x\right) \leq
\Q\left(\y|\x'\right)e^{\varepsilon}$. These inequalities can be compactly
represented by
\begin{equation}
\tilde{A}\q\leq 0,
\end{equation}
where $\q=\left[\Q\left(\y_1|\x_1\right) ,...,\Q\left(\y_1|\x_{\kk}\right)
,....,\Q\left(\y_{\n}|\x_1\right) ,...,\Q\left(\y_{\n}|\x_{\kk}\right)
\right]^{T}$ and $\tilde{A}$ is a $\n\kk(\kk-1)\times \kk\n$ matrix that
contains all the local differential privacy linear constraints. Observe that there is a one-to-one mapping between $\Q$ and $\q$. Here is an
example for the case when $\kk=\n=2$
\begin{equation}
\underbrace{\left[
  \begin{array}{cccc}
    1 & -e^{\varepsilon} & 0 & 0 \\
    -e^{\varepsilon} & 1 & 0 & 0 \\
    0 & 0 & 1 & -e^{\varepsilon} \\
    0 & 0 & -e^{\varepsilon} & 1 \\
  \end{array}\right]}_\textrm{$\tilde{A}$}\left[
               \begin{array}{c}
                 \Q\left(\y_1|\x_1\right) \\
                 \Q\left(\y_1|\x_2\right) \\
                 \Q\left(\y_2|\x_1\right) \\
                 \Q\left(\y_2|\x_2\right) \\
               \end{array}
             \right]\leq 0.
\end{equation}
Moreover, since $\Q$ is a row stochastic matrix (a conditional distribution)
it satisfies $\Q\unity=\unity$, where $\unity$ represents the all ones vector
of appropriate dimensions. This condition can be rewritten as
\begin{equation}
B\q=\unity,
\end{equation}
where $B$ is a $\kk\times \kk\n$ binary matrix. For the case when $\kk=\n=2$,
we have that
\begin{equation}
\underbrace{\left[
  \begin{array}{cccc}
    1 & 0 & 1 & 0 \\
    0 & 1 & 0 & 1 \\
  \end{array}\right]}_\textrm{$B$}\left[
               \begin{array}{c}
                 \Q\left(\y_1|\x_1\right) \\
                 \Q\left(\y_1|\x_2\right) \\
                 \Q\left(\y_2|\x_1\right) \\
                 \Q\left(\y_2|\x_2\right) \\
               \end{array}
             \right]= \left[        \begin{array}{c}
                 1 \\
                 1 \\\end{array}\right].
\end{equation}
Finally, observe that $\Q\left(\y|\x\right)\geq0$ for all $\x \in \mathcal{\X}$ and
$\y \in \mathcal{\Y}$. These constraints can be incorporated as follows. Let
$A=\left[\tilde{A}^{T}\;,\; -I_{\n\kk}\right]^{T}$, where $I_{\n\kk}$ is the
$\n\kk\times \n\kk$ identity matrix, then $A\q\leq0$. To summarize, $\Q \in
\mathcal{\D}_{\varepsilon,\n}$ if and only if
\begin{eqnarray}
\label{local_dp_polytope}
A\q&\leq&0 \\ \nonumber
B\q&=&\unity.
\end{eqnarray}
Therefore, the set of all $\kk \times \n$, $\varepsilon$-locally differentially private mechanisms
$\mathcal{\D}_{\varepsilon,\n}$ forms a convex polytope in $\mathbb{R}_{+}^{\kk\n}$.

We now proceed to proving that if $\Q$ is a corner point of the polytope formed by $\mathcal{\D}_{\varepsilon,\n}$, then $\Q$ has at most $\kk$ non-zero columns. This claim is obvious for all $\kk \times \n$ privatization mechanisms with $\n \leq \kk$. Therefore, we restrict our attention to the case where $\n > \kk$. Let $A_j$ be the matrix including all the inequality constraints imposed on
the $j^{th}$ column of $\Q$. Observe that the rows of $A_j$ form a subset of
the rows of $A$, defined in (\ref{local_dp_polytope}), and recall that there
are $\kk(\kk-1)$ differential privacy and $\kk$ non-negativity inequality
constraints imposed on the $j^{th}$ column of $\Q$. Therefore, $A_j$ is a
$\kk^2\times\kk$ matrix and we have that $A_j\Q_j\leq0$, where $\Q_j$
represents the $j^{th}$ column of $\Q$. By Claim \ref{claim:zero_column}, we
know that $\Q_j$ is either equal to zero or contains non-zero entries.
\begin{claim}
\label{cl:all_zero_cols} In what follows, the term linearly independent
inequality constraints refers to linear independent rows of $A_j$.
\begin{itemize}
\item If $\Q_j=0$, then $\kk$ linearly independent inequality constraints
    are achieved with equality.
\item If $\Q_j \neq 0$, then at most $\kk-1$ linearly independent
    inequality constraints can be achieved with equality.
    \end{itemize}
\end{claim}
\begin{proof}
In fact, the number of linearly independent inequality constraints (achieved
or not) cannot exceed $\kk$ because $A_j$ has a rank less than or equal to
$\kk$. If $\Q_j=0$, then the $\kk$ non-negativity inequality constraints are
    achieved with equality and it is easy to see that they are all linearly
    independent (in fact, they form an orthonormal basis to
    $\mathbb{R}^{\kk}$). This proves the first part of the claim. We now
    establish the second part of the claim by showing that if $\Q_j\neq0$, we cannot have
    $\kk$ linearly independent inequality constraints achieved with equality. Assume that $\Q_j \neq 0$ and
 let
$\tilde{A}_j$ be the matrix including the largest collection of linearly
independent rows of $A_j$ corresponding to the inequality constraints that
are achieved with equality. In other words, $\tilde{A}_j\Q_j=0$. If
$\tilde{A}_j$ contains $\kk$ rows, then its rank is equal to $\kk$. However,
this implies that $\Q_j=0$, a contradiction. Therefore, at most $\kk-1$
linearly independent inequality constraints can be achieved with equality
when $\Q_j \neq 0$.
\end{proof}

Suppose that $\Q$ is a corner point of $\mathcal{\D}_{\varepsilon,\n}$ and out of its
$\n$ columns, $\n_{>0}$ are non-zero and $\n_{=0}$ ($\n_{=0}=\n-\n_{>0}$) are
zero. Moreover, assume that the number of non-zero columns of $\Q$ is larger
than $\kk$ (i.e., $\n_{>0}>\kk$). In this case, from Claim
\ref{cl:all_zero_cols}, we can see that $\Q$ achieves at most $\n_{>0}(\kk-1)
+ (\n-\n_{>0})\kk$ linearly independent inequality constraints with equality.
Furthermore, at most $\kk$ additional linearly independent equality
constraints (linearly independent rows of the matrix $B$ defined in
(\ref{local_dp_polytope})) can be met by $\Q$. Therefore, the total number of
linearly independent constraints that $\Q$ achieves with equality is at most
$\n_{>0}(\kk-1) + (\n-\n_{>0})\kk+\kk = -\n_{>0} + (\n+1)\kk < \n\kk$, where the
last strict inequality follows from the fact that $\n_{>0}>\kk$. This implies
that $\Q$ cannot be a corner point of $\mathcal{\D}_{\varepsilon,\n}$. Therefore, any
corner point of $\mathcal{\D}_{\varepsilon,\n}$ must have at most $\kk$ non-zero columns.

%

\section{Proofs for Hypothesis Testing}
%
%
\subsection{Proof of Theorem \ref{thm:hypbin}}
Let $\T =\left\{\x: \po(\x)\geq \pt(\x)\right\}$. In other words, $\Po(\T) - \Pt(\T) = \max_{A \subseteq \mathcal{\X}} \Po(A) - \Pt(A)$. Recall that for a given $\Po$ and $\Pt$, the binary mechanism is defined as a staircase mechanism with only two outputs
$y\in\{0,1\}$ satisfying
\begin{eqnarray}
\Q(0|x) \,=\, \left\{
\begin{array}{rl}
	\frac{e^\varepsilon}{1+e^\varepsilon}& \text{ if } \PP_0(x)\geq \PP_1(x)\;,\\
	\frac{1}{1+e^\varepsilon}& \text{ if } \PP_0(x)< \PP_1(x)\;.\\
\end{array}
\right.\;\;\;
\Q(1|x) \,=\, \left\{
\begin{array}{rl}
	\frac{e^\varepsilon}{1+e^\varepsilon}& \text{ if } \PP_0(x)< \PP_1(x)\;,\\
	\frac{1}{1+e^\varepsilon}& \text{ if } \PP_0(x)\geq \PP_1(x)\;. \\
\end{array}
\right.
\label{eq:defbin2}
\end{eqnarray}

\begin{lemma}
\label{lemma:hyp_test_region_bin}
	 For any pair of distributions
$\PP_0$ and $\PP_1$,
	there exists a positive $\varepsilon^*$ that depends on  $\PP_0$ and $\PP_1$ such that
	for all $\y \in \mathcal{\Y}$, all $\n \in \mathbb{N}$, and all $\Q \in \mathcal{\D}_{\varepsilon,\n}$ with $\varepsilon\leq \varepsilon^*$, we have that
\begin{equation}
\frac{\left(e^{\varepsilon}-1\right)\Po\left(\T^c\right)+1}{\left(e^{\varepsilon}-1\right)\Pt\left(\T^c\right)+1} \leq \frac{\mo(\y)}{\mt(\y)}\leq \frac{\left(e^{\varepsilon}-1\right)\Po\left(\T\right)+1}{\left(e^{\varepsilon}-1\right)\Pt\left(\T\right)+1}.
\end{equation}
Moreover, the above upper and lower bounds are achieved by the binary
mechanism given in (\ref{eq:defbin2}).
\end{lemma}
Observe that because $\Po\left(\T\right)\geq \Pt\left(\T\right)$ and
$\Po\left(\T^c\right)\leq \Pt\left(\T^c\right)$, the direction of the above
inequalities makes sense.

Let $\tilde{M}_\nu$ be the induced marginal for the binary mechanism when $\PP_\nu$ is the original distribution.
Following the analysis techniques developed in \cite{KOV14b}, we define
hypothesis testing region $R(\tilde{M}_0,\tilde{M}_1)$ as the convex hull of all achievable probabilities
of missed detection and false alarm, when testing whether $\nu=0$ or $\nu=1$ based on $Y_{\rm bin}$
distributed as $\tilde{\MM}_{\nu}$:
\begin{eqnarray*}
	R(\tilde{M}_0,\tilde{M}_1) \equiv {\rm conv}\Big( \big\{ (\tilde{M}_1(S), \tilde{M}_0(S^c))\,:\, \forall S\subseteq \cY \big\} \Big)\;,
\end{eqnarray*}
where $S\in\cY$ is the accept region for hypothesis $\nu=0$.
For the binary mechanism, this ends up being a very simple triangular region.
The slopes defining the two sides of the triangular region are:
$- \max_S \tilde{M}_0(S)/\tilde{M}_1(S) = -((e^\varepsilon-1)\PP_0(T)+1)/((e^\varepsilon-1)\PP_1(T)+1)$ and
$- \min_S \tilde{M}_0(S^c)/\tilde{M}_1(S^c) = -((e^\varepsilon-1)\PP_0(T^c)+1)/((e^\varepsilon-1)\PP_1(T^c)+1)$.

  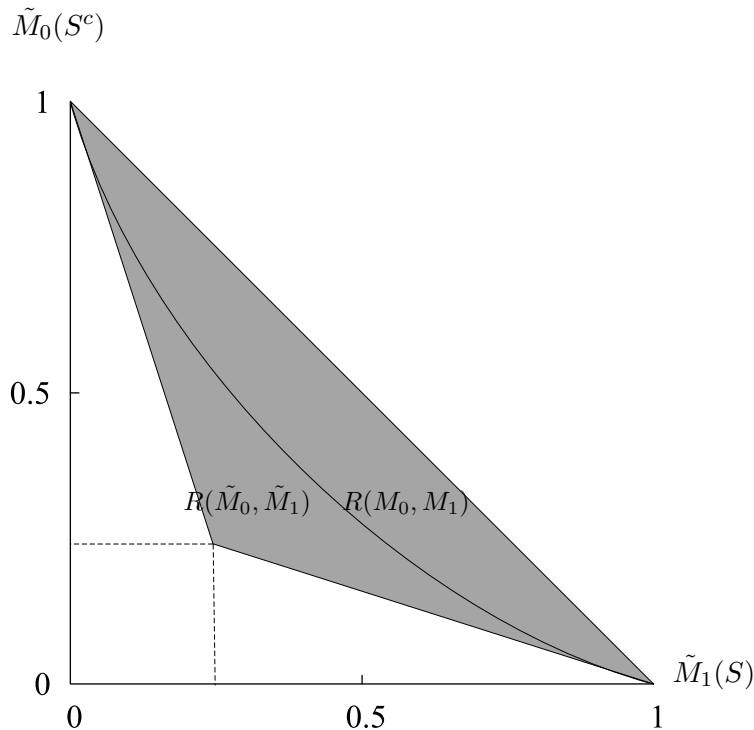
\begin{figure}[h]
  	\begin{center}
     \includegraphics[width=.6\textwidth]{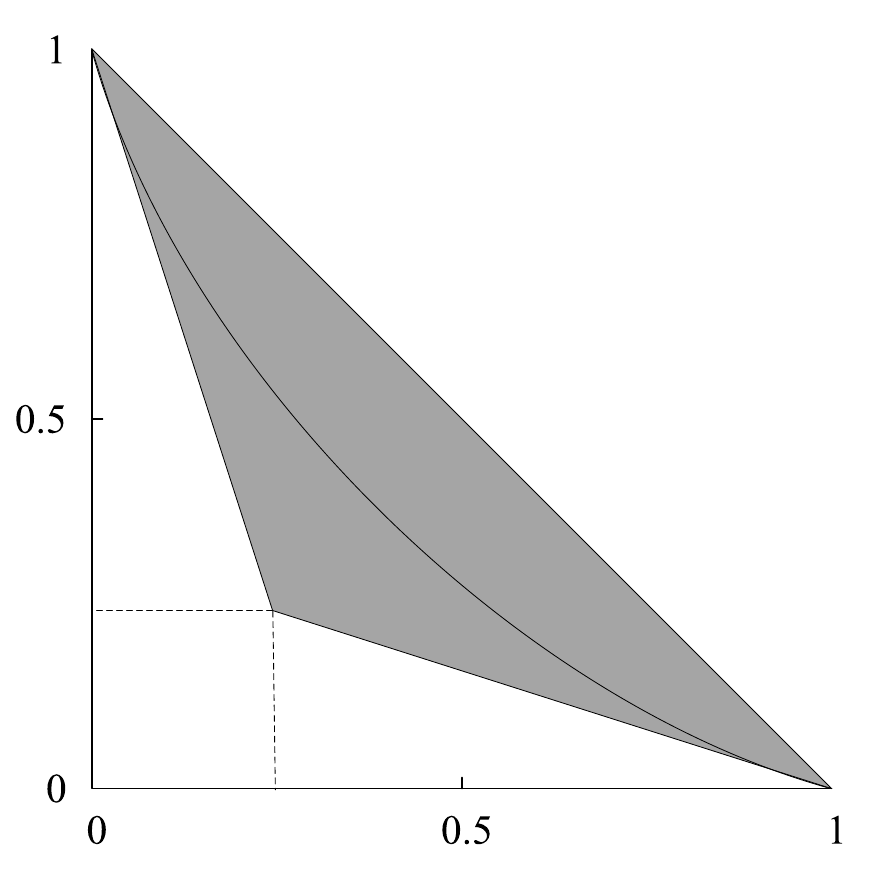}
     \put(-255,275){$\tilde{M}_0(S^c)$}
     \put(-5,30){$\tilde{M}_1(S)$}
     \put(-190,95){\small$R(\tilde{M}_0,\tilde{M}_1)$}
     \put(-130,95){\small$R({M}_0,{M}_1)$}
     \end{center}
     \caption{Hypothesis testing regions for two mechanisms.}
     \label{fig:region}
  \end{figure}%

For any other mechanism satisfying the $\varepsilon$-local differential privacy for
$\varepsilon\leq \varepsilon^*$,
the above lemma implies that for any choice of the rejection region $S$,
the slopes  satisfy $ -{M}_0(S)/ {M}_1(S) \geq -((e^\varepsilon-1)\PP_0(T)+1)/((e^\varepsilon-1)\PP_1(T)+1)$ and
$ -{M}_0(S^c)/ {M}_1(S^c) \leq -((e^\varepsilon-1)\PP_0(T^c)+1)/((e^\varepsilon-1)\PP_1(T^c)+1)$.
In the hypothesis testing region, this implies that
\begin{eqnarray*}
	R( {M}_0, {M}_1) \subseteq R(\tilde{M}_0,\tilde{M}_1)\;,
\end{eqnarray*}
as in the following Figure \ref{fig:region}.

From Theorem 2.5 of \cite{KOV14b}, we know that
this implies a certain Markov property.
Precisely,
let $Y_{\rm bin}$ denote the output of the binary mechanism, and
$Y_{\rm dp}$ denote the output of any $\varepsilon$-local differentially private mechanism.
Then, it follows that there exists a coupling of $Y_{\rm bin}$ and $Y_{\rm dp}$ such that
they form a Markov chain: $\nu$--$Y_{\rm bin}$--$Y_{\rm dp}$, where
$\nu$ is the hypothesis on $\PP_\nu$ whether the data was generated from $\nu=0$ or $\nu=1$.
Then, it follows from the data processing inequality of $f$-divergences that
\begin{eqnarray*}
	D_f(\tilde{M}_0,\tilde{M}_1) \geq D_f({M}_0, {M}_1)\;.
\end{eqnarray*}
It follows that there is no algorithm with larger $f$-divergence than the binary mechanism.

\subsection{Proof of Lemma \ref{lemma:hyp_test_region_bin}}
We start by showing that the binary mechanism achieves the
upper and lower bounds presented in the statement of the lemma. Let $\Mo^B$ and $\Mt^B$ denote the induced marginals under the binary mechanism given in (\ref{eq:defbin2}). For $\nu \in \{0,1\}$, we have that
\begin{eqnarray}
\mm^B_{\nu}\left(0\right)&=&\sum_{\x \in \mathcal{\X}}\po\left(\x\right)\Q(0|x)= \frac{1}{e^\varepsilon+1}\left(\left(e^{\varepsilon}-1\right)\PP_\nu\left(\T\right)+1\right) \nonumber \\
\mm^B_{\nu}\left(1\right)&=&\sum_{\x \in \mathcal{\X}}\po\left(\x\right)\Q(1|x)=\frac{1}{e^\varepsilon+1}\left(\left(e^{\varepsilon}-1\right)\PP_\nu\left(\T^c\right)+1\right).
\end{eqnarray}
Computing $\mo^B\left(0\right)/\mt^B\left(0\right)$ and $\mo^B\left(1\right)/\mt^B\left(1\right)$ we see that the binary mechanism achieves the upper and lower bounds for all values of $\varepsilon$.

As in Lemma
\ref{lemma:diff_priv_reps}, for any $\n \in \mathbb{N}$, $\Q \in
\mathcal{\D}_{\varepsilon,\n}$ can be represented as $\Q=\s\B$, where $\s \in [1,e^\varepsilon]^{\kk\times\n}$
and $\B=\mbox{diag}\left(\bb_1,...,\bb_{\n}\right)$ with its diagonal entries in
$\mathbb{R}_{+}$. We now show that for any $\n \in \mathbb{N}$ and any $\Q \in
\mathcal{\D}_{\varepsilon,\n}$, the following upper bound
holds
\begin{equation}
\frac{\mo(\y)}{\mt(\y)}= \frac{\sum_{i\in[\kk]}\po\left(x_i\right)\s_{ij}}{\sum_{i\in[\kk]}\pt\left(x_i\right)\s_{ij}}\leq \frac{\left(e^{\varepsilon}-1\right)\Po\left(\T\right)+1}{\left(e^{\varepsilon}-1\right)\Pt\left(\T\right)+1},
\end{equation}
for all $\y \in \mathcal{\Y}$ and sufficiently small $\varepsilon$. The above expression can be
alternatively written as
\begin{eqnarray}
\label{eq: hypo_test_priv}
\left(e^{\varepsilon}-1\right)\left(\Po\left(\T\right)-\Pt\left(\T\right)\right)+\left(e^{\varepsilon}-1\right)\sum_{i\in[\kk]}\left(\s_{ij}-1\right)\left(\Po\left(\T\right)\pt\left(x_i\right)-\Pt\left(\T\right)\po\left(x_i\right)\right) \nonumber \\
-\sum_{i\in[\kk]}\left(\s_{ij}-1\right)\left(\po\left(x_i\right)-\pt\left(x_i\right)\right)\geq 0,
\end{eqnarray}
where $\s_j \in [1,e^{\varepsilon}]^{\kk}$. Equation (\ref{eq:
hypo_test_priv}) is linear in $\s_{j}$ and is therefore
minimized (and maximized) at the corner points of $\ [1,e^{\varepsilon}]^{\kk
\times \n}$, a cube in $\mathbb{R}^{\kk\times\n}_{+}$. The corner points of
this cube are given by the staircase patterns: $\s_j \in
\{1,e^{\varepsilon}\}^{\kk}$. To begin with, let $\s_j$ be a
staircase pattern with $\T_j=\{x_i: \s_{ij}=e^{\varepsilon}\}
\neq \T$. Then Equation (\ref{eq: hypo_test_priv}) is equivalent to
\begin{eqnarray}
\label{eq: hypo_test_priv2}
\left(e^{\varepsilon}-1\right)\left\{\left(\Po\left(\T\right)-\Pt\left(\T\right)\right)-\left(\Po\left(\T_j\right)-\Pt\left(\T_j\right)\right)\right\}~~~~~~~~~~~~~~~~~~~~~~~~~~~~~~~~~~~~~~~~~~~~~~ \nonumber \\
+\left(e^{\varepsilon}-1\right)^2\left\{\Po\left(\T\right)\Pt\left(\T_j\right)-\Pt\left(\T\right)\Po\left(\T_j\right)\right\}\geq 0.
\end{eqnarray}
Using the fact that $\Po\left(\T\right)-\Pt\left(\T\right)
> \Po\left(\T_j\right)-\Pt\left(\T_j\right)$ for all
$\T_j\neq\T$, the inequality in (\ref{eq: hypo_test_priv}) holds true for all $\varepsilon$ whenever $\Po\left(\T\right)\Pt\left(\T_j\right) \geq \Pt\left(\T\right)\Po\left(\T_j\right)$. If $\Po\left(\T\right)\Pt\left(\T_j\right) < \Pt\left(\T\right)\Po\left(\T_j\right)$, then the inequality in (\ref{eq: hypo_test_priv}) holds true for all $\varepsilon \leq \varepsilon(j)$, where
\begin{equation}
 \varepsilon(j) = \log\left(\frac{\left(\Po\left(\T\right)-\Pt\left(\T\right)\right)-\left(\Po\left(\T_j\right)-\Pt\left(\T_j\right)\right)}{\Pt\left(\T\right)\Po\left(\T_j\right)-\Po\left(\T\right)\Pt\left(\T_j\right)}+1\right)>0.
\end{equation}
On the other hand, it is easy to verify that when $\T_j=\T$, we
have that
\begin{eqnarray}
\left(e^{\varepsilon}-1\right)\left\{\left(\Po\left(\T\right)-\Pt\left(\T\right)\right)-\left(\Po\left(\T_j\right)-\Pt\left(\T_j\right)\right) \right. ~~~~~~~~~~~~~~~~~~~~~~~~~~~~~~~~~~~~~~~~~~~~~~\nonumber \\
\left.+\left(e^{\varepsilon}-1\right)\left(\Po\left(\T\right)\Pt\left(\T_j\right)-\Pt\left(\T\right)\Po\left(\T_j\right)\right)\right\}= 0,
\end{eqnarray}
for all $\varepsilon$. In this case, set $\varepsilon(j)=0$ and $\varepsilon_1 = \min_{j \in [2^\kk]} \varepsilon(j)$. Therefore, for any $\n \in \mathbb{N}$ and any $\Q \in
\mathcal{\D}_{\varepsilon,\n}$, the upper bound in the statement of the lemma holds for all $\varepsilon \leq \varepsilon_1$.

We now show that for for any $\n \in \mathbb{N}$ and any $\Q \in
\mathcal{\D}_{\varepsilon,\n}$, the following lower bound
holds
\begin{equation}
\frac{\left(e^{\varepsilon}-1\right)\Po\left(\T^c\right)+1}{\left(e^{\varepsilon}-1\right)\Pt\left(\T^c\right)+1}\leq \frac{\mo(\y)}{\mt(\y)}= \frac{\sum_{i\in[\kk]}\po\left(x_i\right)\s_{ij}}{\sum_{i\in[\kk]}\pt\left(x_i\right)\s_{ij}},
\end{equation}
for all $\y \in \mathcal{\Y}$ and sufficiently small $\varepsilon$. The above expression can be
alternatively written as
\begin{eqnarray}
\label{eq: hypo_test_priv_low}
\left(e^{\varepsilon}-1\right)\left(\Po\left(\T\right)-\Pt\left(\T\right)\right)+\left(e^{\varepsilon}-1\right)\sum_{i\in[\kk]}\left(\s_{ij}-1\right)\left(\Po\left(\T\right)\pt\left(x_i\right)-\Pt\left(\T\right)\po\left(x_i\right)\right) \nonumber \\
+e^{\varepsilon}\sum_{i\in[\kk]}\left(\s_{ij}-1\right)\left(\po\left(x_i\right)-\pt\left(x_i\right)\right)\geq 0,
\end{eqnarray}
where $\s_j \in [1,e^{\varepsilon}]^{\kk}$. Equation (\ref{eq:
hypo_test_priv_low}) is linear in $\s_{j}$ and is
therefore minimized at the corner points of $\ [1,e^{\varepsilon}]^{\kk}$, a cube in $\mathbb{R}^{\kk}_{+}$. The corner points of
this cube are given by staircase patterns: $\s_j \in
\{1,e^{\varepsilon}\}^{\kk}$. To begin with, let $\s_j$ be a
staircase pattern with $\T_j=\{x_i: \s_{ij}=e^{\varepsilon}\}
\neq \T^c$, then Equation (\ref{eq: hypo_test_priv_low}) is equivalent to
\begin{eqnarray}
\label{eq: small_eps_low_ boud}
\left(e^{\varepsilon}-1\right)\left\{\left(\Po\left(\T\right)-\Pt\left(\T\right)\right)+e^{\varepsilon}\left(\Po\left(\T_j\right)-\Pt\left(\T_j\right)\right)\right\} ~~~~~~~~~~~~~~~~~~~~~~~~~~~~~~~~~~~~~~~~~~~~~~\nonumber \\
+\left(e^{\varepsilon}-1\right)^2\left\{\Po\left(\T\right)\Pt\left(\T_j\right)-\Pt\left(\T\right)\Po\left(\T_j\right)\right\}\geq 0.
\end{eqnarray}
Using the fact that $\Po\left(\T\right)-\Pt\left(\T\right)
> \Pt\left(\T_j\right)-\Po\left(\T_j\right)$ for all
$\T_j\neq\T^c$, then for sufficiently small $\varepsilon$, Equation (\ref{eq: hypo_test_priv_low}) can be written as
\begin{equation}
\varepsilon\left\{\left(\Po\left(\T\right)-\Pt\left(\T\right)\right)-\left(\Pt\left(\T_j\right)-\Po\left(\T_j\right)\right)\right\}+O\left(\varepsilon^2\right)> 0.
\end{equation}
This proves that there exists a positive $\varepsilon(j)$ such that the left hand side of Equation (\ref{eq: small_eps_low_
boud}) is positive for all $\varepsilon \leq \varepsilon(j)$. On the other hand, it is easy to verify that when
$\T_j=\T^c$, we have that
\begin{eqnarray}
\left(e^{\varepsilon}-1\right)\left\{\left(\Po\left(\T\right)-\Pt\left(\T\right)\right)+e^{\varepsilon}\left(\Po\left(\T_j\right)-\Pt\left(\T_j\right)\right) \right. ~~~~~~~~~~~~~~~~~~~~~~~~~~~~~~~~~~~~~~~~~~~~~~\nonumber \\
\left. +\left(e^{\varepsilon}-1\right)\left(\Po\left(\T\right)\Pt\left(\T_j\right)-\Pt\left(\T\right)\Po\left(\T_j\right)\right)\right\}= 0,
\end{eqnarray}
for all $\varepsilon$. In this case, let $\varepsilon(j)=0$ and let $\varepsilon_2 = \min_{j \in [2^\kk]} \varepsilon(j)$. Therefore, for any $\n \in \mathbb{N}$ and any $\Q \in
\mathcal{\D}_{\varepsilon,\n}$, the lower bound in the statement of the lemma holds for all $\varepsilon \leq \varepsilon_2$. To conclude, let $\varepsilon^* = \min(\varepsilon_1,\varepsilon_2)$. Then both, the upper and lower bounds,
hold for all $\varepsilon \leq \varepsilon^*$.

%
%
\subsection{Proof of Theorem \ref{thm:hyptv}}
The total variation (TV) distance $\|\Mo-\Mt\|_{\rm TV}$ is a special case of
$f$-divergence $D_f(\Mo||\Mt)$ with $f(x)=\frac{1}{2}|x-1|$. Therefore, by Theorem
\ref{thm:lp}, we have that
\begin{equation}
\label{eq:TV_lp}
\begin{aligned}
\max_{\Q \in
\mathcal{\D}_{\varepsilon}} \big\|\Mo-\Mt\big\|_{\rm TV} & = & \underset{\bb}{\text{maximize}} & & \mu^{T}\bb \\
& & \text{subject to} & & \s^{(\kk)}\bb = \unity\\
& & & &  \bb \geq 0,
\end{aligned}
\end{equation}
where $ \mu_j = \mu\left(\s^{(\kk)}_j\right)=\frac{1}{2}\big|\sum_{i \in [\kk]}\left(\po(\x_i)-\pt(\x_i)\right)\s^{(\kk)}_{ij}\big|$ for $j \in \{1, \ldots,2^\kk\}$ and $\s^{(\kk)}$ is the $\kk \times 2^\kk$ staircase pattern matrix given in
Definition \ref{def:staircase_matrix}.

The polytope given by $\s^{(\kk)}\bb = \unity$ and $\bb \geq 0$ is a closed and bounded one. Thus,
there is no duality gap and solving the above linear program is
equivalent to solving its dual
\begin{equation}
\label{eq:dual_TV_lp}
\begin{aligned}
& \underset{\dd}{\text{minimize}} & & \unity^{T}\dd \\
& \text{subject to} & & {\s^{(\kk)}}^{T}\dd \geq \mu.
\end{aligned}
\end{equation}
Note that any satisfiable solution $\dd^*$ to (\ref{eq:dual_TV_lp}) provides
an upper bound to (\ref{eq:TV_lp}) since $\max \mu^{T}\bb = \min
\unity^{T}\dd \leq \unity^{T}\dd^*$. Let $\T =\left\{\x: \po(\x)\geq
\pt(\x)\right\}$ and $\T_j =\{\x_i: \s^{(\kk)}_{ij}=e^\varepsilon\}$ for $j
\in [2^\kk]$. Consider
the following choice of dual variable
\begin{equation}
\dd^*_i = \frac{1}{2}\frac{e^\varepsilon -1}{e^\varepsilon+1}\big|\po(\x_i)-\pt(\x_i)\big|,
\end{equation}
for $i \in [\kk]$. Observe that
\begin{eqnarray}
\label{eq:upper_bound_TV}
\unity^{T}\dd^*&=& \frac{1}{2}\frac{e^\varepsilon -1}{e^\varepsilon+1}\sum_{i\in [\kk]}\big|\po(\x_i)-\pt(\x_i)\big| \nonumber \\
&=&\frac{1}{2}\frac{e^\varepsilon -1}{e^\varepsilon+1}\big\|\Po-\Pt\big\|_{1} \nonumber \\
&=&\frac{e^\varepsilon -1}{e^\varepsilon+1}\big\|\Po-\Pt\big\|_{\rm TV}.
\end{eqnarray}
We claim that $\dd^*$ is a feasible dual variable for all values of
$\varepsilon$. In order to prove that $\dd^*$ is a feasible dual variable, we
show that ${\s^{(\kk)}}^{T}_j\dd^* -\mu_j \geq 0 $ for all $j
\in [2^\kk]$ and all $\varepsilon$. For all $j \in [2^\kk]$, we have that
 \begin{eqnarray}
 g_j &=& 2\left({\s^{(\kk)}}^{T}_j\dd^* -  \mu_j\right) \nonumber \\
 &=& \frac{e^{\varepsilon}-1}{e^{\varepsilon}+1}\sum_{i \in [\kk]}\left|\po(\x_i)-\pt(\x_i)\right|\s^{(\kk)}_{ij} -\left|\sum_{i \in [\kk]}\left(\po(\x_i)-\pt(\x_i)\right)\s^{(\kk)}_{ij}\right| \nonumber \\
&=&\frac{e^{\varepsilon}-1}{e^{\varepsilon}+1}\left\{ \sum_{\x_i\in \T}\left(\po(\x_i)-\pt(\x_i)\right)\s^{(\kk)}_{ij} + \sum_{\x_i\in \T^c}\left(\pt(\x_i)-\po(\x_i)\right)\s^{(\kk)}_{ij} \right\} \nonumber \\
&&~~~-\left|\sum_{\x_i\in \T}\left(\po(\x_i)-\pt(\x_i)\right)\s^{(\kk)}_{ij}-\sum_{\x_i\in \T^c}\left(\pt(\x_i)-\po(\x_i)\right)\s^{(\kk)}_{ij}\right|.
\label{eq:feasibility}
\end{eqnarray}
Notice that we have arranged the equation such that all the summands are
non-negative. Without loss of generality, we will assume that
\begin{equation*}
\sum_{\x_i\in
\T}\left(\po(\x_i)-\pt(\x_i)\right)\s^{(\kk)}_{ij} \geq \sum_{\x_i\in
\T^c}\left(\pt(\x_i)-\po(\x_i)\right)\s^{(\kk)}_{ij}.
\end{equation*}
From the equality $\sum_{\x_i\in
\T}\left(\po(\x_i)-\pt(\x_i)\right) = \sum_{\x_i\in
\T^c}\left(\pt(\x_i)-\po(\x_i)\right)$ and the fact
that $\s^{(\kk)}_{ij} \in \{1,e^\varepsilon\}$ for all $i$ and $j$, it follows that
\begin{eqnarray}
	e^{-\varepsilon} \sum_{\x_i\in \T}\left(\po(\x_i)-\pt(\x_i)\right)\s^{(\kk)}_{ij} \leq \sum_{\x_i\in \T^c}\left(\pt(\x_i)-\po(\x_i)\right)\s^{(\kk)}_{ij}\;.
	\label{eq:lb}
\end{eqnarray}
This is true because the right-hand side is minimized when the $\s^{(\kk)}_{ij}$'s
for $\x_i\in \T^c$ are all equal to 1 and the left-hand side is maximized when
the $\s^{(\kk)}_{ij}$'s for $\x_i\in \T$ are all equal to $e^\varepsilon$. Now, \eqref{eq:feasibility}
can be written as
\begin{eqnarray*}
g_j &=& \frac{1}{e^{\varepsilon}+1}\left\{ -2 \sum_{\x_i\in \T}\left(\po(\x_i)-\pt(\x_i)\right)\s^{(\kk)}_{ij} + 2e^\varepsilon \sum_{\x_i\in \T^c}\left(\pt(\x_i)-\po(\x_i)\right)\s^{(\kk)}_{ij} \right\} \\
&\geq& 0 \;,
\end{eqnarray*}
where the last inequality follows from \eqref{eq:lb}.

This establishes the satisfiability of $\dd^*$ for
all $\varepsilon$ which, in turn, shows that (\ref{eq:upper_bound_TV}) is indeed an upper bound to the primal problem. It remains to show that
this upper bound can be achieved via the binary mechanism. To this
extent, recall that for a given $\Po$ and $\Pt$, the binary mechanism is defined as a staircase mechanism with only two outputs
$y\in\{0,1\}$ satisfying
\begin{eqnarray}
\Q(0|x) \,=\, \left\{
\begin{array}{rl}
	\frac{e^\varepsilon}{1+e^\varepsilon}& \text{ if } \po(x)\geq \pt(x)\;,\\
	\frac{1}{1+e^\varepsilon}& \text{ if } \po(x)< \pt(x)\;.\\
\end{array}
\right.\;\;\;
\Q(1|x) \,=\, \left\{
\begin{array}{rl}
	\frac{e^\varepsilon}{1+e^\varepsilon}& \text{ if } \po(x)< \pt(x)\;,\\
	\frac{1}{1+e^\varepsilon}& \text{ if } \po(x)\geq \pt(x)\;. \\
\end{array}
\right.
\label{eq:defbin3}
\end{eqnarray}
Computing the TV distance between $\Mo$ and $\Mt$ under (\ref{eq:defbin3}), we get that
\begin{equation}
 \big\|\Mo-\Mt\big\|_{\rm TV} =\frac{e^\varepsilon -1}{e^\varepsilon+1}\big\|\Po-\Pt\big\|_{\rm TV}.
\end{equation}
Hence, the binary mechanism in (\ref{eq:defbin3}) achieves the upper bound in (\ref{eq:upper_bound_TV}). This proves the optimality of the binary mechanism for all $\varepsilon$.

%
%
\subsection{Proof of Theorem \ref{thm:hyprr}}
The Kullback-Leibler (KL) divergence $\Dkl(\Mo||\Mt)$ is a special $f$-divergence $D_f(\Mo||\Mt)$ with $f(x)=x\log x$. Therefore, by Theorem
\ref{thm:lp}, we have that
\begin{equation}
\label{eq:KL_lp}
\begin{aligned}
\max_{\Q \in
\mathcal{\D}_{\varepsilon}} \Dkl(\Mo||\Mt) & = & \underset{\bb}{\text{maximize}} & & \mu^{T}\bb \\
& & \text{subject to} & & \s^{(\kk)}\bb = \unity\\
& & & &  \bb \geq 0,
\end{aligned}
\end{equation}
where $ \mu_j = \mu\left(\s^{(\kk)}_j\right)=\sum_{i \in [\kk]}\po(\x_i)\s^{(\kk)}_{ij}\log\left(\frac{\sum_{i \in [\kk]}\po(\x_i)\s^{(\kk)}_{ij}}{\sum_{i \in [\kk]}\pt(\x_i)\s^{(\kk)}_{ij}}\right)$ for $j \in \{1, \ldots,2^\kk\}$ and $\s^{(\kk)}$ is the $\kk \times 2^\kk$ staircase pattern matrix given in
Definition \ref{def:staircase_matrix}.

The polytope given by $\s^{(\kk)}\bb = \unity$ and $\bb \geq 0$ is a closed and bounded one. Thus,
there is no duality gap and solving the above linear program is
equivalent to solving its dual
\begin{equation}
\label{eq:dual_KL_lp}
\begin{aligned}
& \underset{\dd}{\text{minimize}} & & \unity^{T}\dd \\
& \text{subject to} & & {\s^{(\kk)}}^{T}\dd \geq \mu.
\end{aligned}
\end{equation}
Note that any satisfiable solution $\dd^*$ to (\ref{eq:dual_KL_lp}) provides
an upper bound to (\ref{eq:KL_lp}) since $\max \mu^{T}\bb = \min
\unity^{T}\dd \leq \unity^{T}\dd^*$. Let $\T =\left\{\x: \po(\x)\geq
\pt(\x)\right\}$ and $\T_j =\{\x_i: \s^{(\kk)}_{ij}=e^\varepsilon\}$ for $j
\in [2^\kk]$. Set $j_i = \{j: \T_j = \x_i\}$ for $i \in [\kk]$, and consider
the following choice of dual variable
\begin{equation}
\dd^*_i = \frac{1}{\left(e^\varepsilon -1\right)\left(e^\varepsilon+\kk-1\right)}\left\{\left(e^\varepsilon+\kk-2\right)\mu\left(\s^{(\kk)}_{j_i}\right)-\sum_{l \in [\kk], l \neq i}\mu\left(\s^{(\kk)}_{j_l}\right) \right\},
\end{equation}
for $i \in [\kk]$. Observe that since $\T_{j_i} =\x_i$ we have that
$\PP_{\nu}\left(\T_{j_i}\right)=\pp_{\nu}\left(\x_i\right)$ and since
 \begin{eqnarray}
 \mu_j &=& \sum_{i \in [\kk]}\po(\x_i)\s^{(\kk)}_{ij}\log\left(\frac{\sum_{i \in [\kk]}\po(\x_i)\s^{(\kk)}_{ij}}{\sum_{i \in [\kk]}\pt(\x_i)\s^{(\kk)}_{ij}}\right) \nonumber \\
 &=&  \left(\Po\left(\T_j\right)\left(e^\varepsilon-1\right)+1\right)\log\frac{\left(\Po\left(\T_j\right)\left(e^\varepsilon-1\right)+1\right)}{\left(\Pt\left(\T_j\right)\left(e^\varepsilon-1\right)+1\right)}
 \end{eqnarray}
we have that
\begin{eqnarray}
\label{eq:upper_bound_KL_low_prvcy}
\unity^{T}\dd^*&=& \frac{1}{\left(e^\varepsilon -1\right)\left(e^\varepsilon+\kk-1\right)} \sum_{i \in [\kk]}\left\{\left(e^\varepsilon+\kk-2\right)\mu\left(\s^{(\kk)}_{j_i}\right)-\sum_{l \in [\kk], l \neq i}\mu\left(\s^{(\kk)}_{j_l}\right) \right\} \nonumber \\
&=&\frac{1}{\left(e^\varepsilon -1\right)\left(e^\varepsilon+\kk-1\right)}\left\{\left(e^\varepsilon+\kk-2\right)\sum_{i \in [\kk]}\mu\left(\s^{(\kk)}_{j_i}\right) - \sum_{i \in [\kk]}\sum_{l \in [\kk], l \neq i}\mu\left(\s^{(\kk)}_{j_l}\right) \right\} \nonumber \\
&=&\frac{1}{\left(e^\varepsilon -1\right)\left(e^\varepsilon+\kk-1\right)}\left\{\left(e^\varepsilon+\kk-2\right)\sum_{i \in [\kk]}\mu\left(\s^{(\kk)}_{j_i}\right) - (\kk-1)\sum_{i \in [\kk]}\mu\left(\s^{(\kk)}_{j_i}\right) \right\} \nonumber \\
&=&\frac{1}{\left(e^\varepsilon+\kk-1\right)}\sum_{i \in [\kk]}\mu\left(\s^{(\kk)}_{j_i}\right) \nonumber \\
&=&\frac{1}{\left(e^\varepsilon+\kk-1\right)}\sum_{i \in [\kk]}\left(\Po\left(\x_i\right)\left(e^\varepsilon-1\right)+1\right)\log\frac{\left(\Po\left(\x_i\right)\left(e^\varepsilon-1\right)+1\right)}{\left(\Pt\left(\x_i\right)\left(e^\varepsilon-1\right)+1\right)}.
\end{eqnarray}
We claim that $\dd^*$ is a feasible dual variable for sufficiently large
$\varepsilon$. In order to prove that $\dd^*$ is a feasible dual variable, we
show that ${\s^{(\kk)}}^{T}_j\dd^* -\mu_j \geq 0 $ for all $j
\in [2^\kk]$ for all $\varepsilon \geq \varepsilon^*$, where $\varepsilon^*$
is a positive quantity that depends on the priors $\Po$ and $\Pt$. Using the
facts that
\begin{eqnarray}
\log\left(a+e^\varepsilon b\right) &=& \varepsilon + \log b + O\left(e^{-\varepsilon}\right) \nonumber \\
\frac{1}{e^\varepsilon+\kk-1} &=& e^{-\varepsilon} + O\left(e^{-2\varepsilon}\right),
\end{eqnarray}
for large $\varepsilon$, we get that
  \begin{eqnarray}
 \mu_j &=& \left(\Po\left(\T_j\right)\left(e^\varepsilon-1\right)+1\right)\log\frac{\left(\Po\left(\T_j\right)\left(e^\varepsilon-1\right)+1\right)}{\left(\Pt\left(\T_j\right)\left(e^\varepsilon-1\right)+1\right)} \nonumber \\
 &=& \left(\Po\left(\T_j\right) \log \frac{\Po\left(\T_j\right)}{\Pt\left(\T_j\right)}\right)e^\varepsilon + \left(1-\Po\left(\T_j\right)\right)\log \frac{\Po\left(\T_j\right)}{\Pt\left(\T_j\right)} + O\left(e^{-\varepsilon}\right).
 \end{eqnarray}
On the other hand,
  \begin{eqnarray}
  {\s^{(\kk)}}^{T}_j\dd^* &=& \frac{1}{\left(e^\varepsilon -1\right)\left(e^\varepsilon+\kk-1\right)}\left\{\sum_{i \in [\kk]}\s^{(\kk)}_{ij}\left(e^\varepsilon+\kk-2\right)\left(\po\left(\x_i\right)\log\frac{\po\left(\x_i\right)}{\pt\left(\x_i\right)}e^\varepsilon + O\left(1\right)\right) \right\} \nonumber \\
 & &- \frac{1}{\left(e^\varepsilon -1\right)\left(e^\varepsilon+\kk-1\right)}\left\{\sum_{i \in [\kk]}\sum_{l \in [\kk], l \neq i}\s^{(\kk)}_{ij}\left(\po\left(\x_l\right)\log\frac{\po\left(\x_l\right)}{\pt\left(\x_l\right)}e^\varepsilon + O\left(1\right) \right)\right\} \nonumber \\
 &=& \frac{1}{\left(e^\varepsilon -1\right)\left(e^\varepsilon+\kk-1\right)} \left(\left(\sum_{\x_i \in \T_j}\po\left(\x_i\right)\log\frac{\po\left(\x_i\right)}{\pt\left(\x_i\right)}\right)e^{3\varepsilon} + O\left(e^{2\varepsilon}\right)\right) \nonumber \\
 &=& \left(\sum_{\x_i \in \T_j}\po\left(\x_i\right)\log\frac{\po\left(\x_i\right)}{\pt\left(\x_i\right)}\right)e^{\varepsilon} + O\left(1\right).
 \end{eqnarray}
 Assume, to begin with, that $j \neq \{j_1,j_2,...,j_{\kk}\}$. Then
 \begin{equation}
{\s^{(\kk)}}^{T}_j\dd^* -  \mu_j = \left(\Po\left(\T_j\right) \log \frac{\Po\left(\T_j\right)}{\Pt\left(\T_j\right)}-\sum_{\x_i \in \T_j}\po\left(\x_i\right)\log\frac{\po\left(\x_i\right)}{\pt\left(\x_i\right)}\right)e^{\varepsilon} + O\left(1\right).
\end{equation}
Notice that for $j \neq \{j_1,j_2,...,j_{\kk}\}$, $\Po\left(\T_j\right) \log \frac{\Po\left(\T_j\right)}{\Pt\left(\T_j\right)}>\sum_{\x_i \in \T_j}\po\left(\x_i\right)\log\frac{\po\left(\x_i\right)}{\pt\left(\x_i\right)}$ by the log-sum inequality. Therefore, there exists a
$\varepsilon(j) >0 $ such that
${\s^{(\kk)}}^{T}_j\dd^* -  \mu_j \geq 0$ for all $\varepsilon
\geq \varepsilon(j)$. If $j \in \{j_1,j_2,...,j_{\kk}\}$,
it is not hard to check that ${\s^{(\kk)}}^{T}_j\dd^* -  \mu_j =
0$ for all $\varepsilon$. In this case, set $\varepsilon(j)=0$. This establishes the satisfiability of $\dd^*$ for
all $\varepsilon \geq \varepsilon^* = \max_{j \in [2^\kk]}\varepsilon(j)$. The satisfiability of $\dd^*$, in turn, shows that (\ref{eq:upper_bound_KL_low_prvcy}) is indeed an upper bound to the primal problem. It remains to show that
this upper bound can be achieved via the randomized response. To this
extent, recall that the randomized response is given by
\begin{eqnarray}
	Q(y|x) \,=\, \left\{
\begin{array}{rl}
	\frac{e^\varepsilon}{|\cX|-1+e^\varepsilon}& \text{ if } y=x\;,\\
	\frac{1}{|\cX|-1+e^\varepsilon}& \text{ if } y\neq x\;.\\
\end{array}
\right.
	\label{eq:defrr2}
\end{eqnarray}
Computing the KL divergence between $\Mo$ and $\Mt$ under (\ref{eq:defrr2}), we get that
\begin{equation}
\Dkl(\Mo||\Mt) = \frac{1}{\left(e^\varepsilon+\kk-1\right)}\sum_{i \in [\kk]}\left(\Po\left(\x_i\right)\left(e^\varepsilon-1\right)+1\right)\log\frac{\left(\Po\left(\x_i\right)\left(e^\varepsilon-1\right)+1\right)}{\left(\Pt\left(\x_i\right)\left(e^\varepsilon-1\right)+1\right)}.
\end{equation}
Hence, the randomized response in (\ref{eq:defrr2}) achieves the upper bound in (\ref{eq:upper_bound_KL_low_prvcy}). This proves the optimality of the randomized response for all $\varepsilon \geq \varepsilon^* $.

%
%
\subsection{Proof of Theorem \ref{thm:hypappx}}
We start the proof with a fundamental bound on the symmetrized KL divergence between the $\Mo$ and $\Mt$.
\begin{lemma}
\label{lemma:duchi_KL_bound}
For any $\varepsilon \geq 0$, let $\Q$ be any conditional distribution that guarantees $\varepsilon$ differential privacy. Then for any pair of distributions $\Po$ and $\Pt$, the induced marginals $\Mo$ and $\Mt$ must satisfy the bound
\begin{equation}
\Dkl\big(\MM_0||\MM_1\big) + \Dkl\big(\MM_1||\MM_0\big) \leq 4 \left(e^\varepsilon -1 \right)^2 \big \| \Po-\Pt\big\|^2_{\rm TV}.
\end{equation}
\end{lemma}

The above lemma appears as Theorem 1 in \cite{DJW13}. By Lemma \ref{lemma:duchi_KL_bound}, we have that
\begin{equation}
{\rm OPT} =\max_{\Q \in \mathcal{\D}_\varepsilon}\Dkl\big(\Mo||\Mt\big)\leq 4 \left(e^\varepsilon -1 \right)^2 \big \| \Po-\Pt\big\|^2_{\rm TV}.
\label{eq:KL_bound}
 \end{equation}
Let $\Mo^B$ and $\Mt^B$ be the marginals obtained by using the binary
mechanism given in (\ref{eq:defhypbin}). By Corollary \ref{coro:hyptv}, we
have that $\| \Mo^B-\Mt^B\|_{\rm
TV}=\frac{e^{\varepsilon}-1}{e^\varepsilon+1}\|\PP_0-\PP_1 \big\|_{\rm TV}$.
Consequently, by applying Pinsker's inequality to the KL divergence between
$\Mo^B$ and $\Mt^B$ we get that
\begin{eqnarray}
{\rm BIN} &=& \Dkl\big(\Mo^B||\Mt^B\big)  \nonumber \\
&\geq& 2 \big\| \Mo^B-\Mt^B \big\|^2_{\rm TV} \nonumber \\
&=& 2 \left(\frac{e^{\varepsilon}-1}{e^\varepsilon+1}\right)^2 \big\| \Po-\Pt \big\|^2_{\rm TV}.
\label{eq:revrs_bound}
\end{eqnarray}

Combining (\ref{eq:KL_bound}) and (\ref{eq:revrs_bound}) we get that ${\rm BIN} \geq  \frac{1}{2(e^\varepsilon+1)^2} {\rm OPT}$ which was to be shown.

\section{Proofs for Information Preservation}
\subsection{Proof of Theorem \ref{thm:estbin}}
\label{sec:estbin}
By Theorem \ref{thm:lp}, we have that
\begin{equation}
\label{eq:mutual_info_lp}
\begin{aligned}
\max_{\Q \in
\mathcal{\D}_{\varepsilon}} I\left(\X;\Y\right) & = &  \underset{\bb}{\text{maximize}} & & \mu^{T}\bb \\
& & \text{subject to} & & \s^{(\kk)}\bb = \unity\\
& & & &  \bb \geq 0,
\end{aligned}
\end{equation}
where $ \mu_j = \mu\left(\s^{(\kk)}_j\right)= \sum_{i \in [\kk]}\pp\left(\x_i\right)\s^{(\kk)}_{ij}\log\left(\frac{\s^{(\kk)}_{ij}}{\sum_{i \in [\kk]}\pp\left(\x_i\right)\s^{(\kk)}_{ij}}\right)$ for $j \in \{1, \ldots,2^\kk\}$ and $\s^{(\kk)}$ is the $\kk \times 2^\kk$ staircase pattern matrix given in
Definition \ref{def:staircase_matrix}. The polytope given by $\s^{(\kk)}\bb = \unity$ and $\bb \geq 0$ is a closed and bounded one. Thus,
there is no duality gap and solving the above linear program is
equivalent to solving its dual
\begin{equation}
\label{eq:dual_mutual_information_lp}
\begin{aligned}
& \underset{\dd}{\text{minimize}} & & \unity^{T}\dd \\
& \text{subject to} & & {\s^{(\kk)}}^{T}\dd \geq \mu.
\end{aligned}
\end{equation}
Note that any satisfiable solution $\dd^*$ to
(\ref{eq:dual_mutual_information_lp}) provides an upper bound to
(\ref{eq:mutual_info_lp}) since $\max \mu^{T}\bb = \min \unity^{T}\dd \leq
\unity^{T}\dd^*$. Let $\T_j =\{\x_i:
\s^{(\kk)}_{ij}=e^\varepsilon\}$ and set $j_1 = \{j: \T_j =\T\}$ and $j_2= \{j:\T_j
= \T^c\}$. Consider the following choice of dual variable
\begin{equation}
\dd^*_i=\frac{1}{\left(e^\varepsilon+1\right)\left(e^\varepsilon-1\right)}\begin{cases} \frac{e^\varepsilon\mu\left(\s^{(\kk)}_{j_1}\right)-\mu\left(\s^{(\kk)}_{j_2}\right)}{|\T|} & \forall i \in \T \\\frac{e^\varepsilon\mu\left(\s^{(\kk)}_{j_2}\right)-\mu\left(\s^{(\kk)}_{j_1}\right)}{|\T^c|} & \forall i \in \T^c \end{cases}.
\end{equation}
 Observe that since $\T_{j_1} =\T$, $\T_{j_2}
=\T^c$, and
 \begin{eqnarray}
 \mu_j &=&  \PP\left(\T_j\right)e^{\varepsilon}\log\frac{e^{\varepsilon}}{\PP\left({\T_j} ^c\right)+e^{\varepsilon}\PP\left({\T_j} \right)}+\PP\left({\T_j} ^c\right)\log\frac{1}{\PP\left({\T_j} ^c\right)+e^{\varepsilon}\PP\left({\T_j} \right)},
 \end{eqnarray}
we have that
\begin{eqnarray}
\label{ub_mi_lp}
\unity^{T}\dd^*&=&\frac{1}{\left(e^\varepsilon+1\right)\left(e^\varepsilon-1\right)}\left\{ \sum_{i \in \T}\frac{1}{|\T|}\left(e^\varepsilon\mu\left(\s^{(\kk)}_{j_1}\right)-\mu\left(\s^{(\kk)}_{j_2}\right)\right) \right.  \nonumber \\
&=& \left. ~~~~ +\sum_{i \in \T^c}\frac{1}{|\T^c|}\left(e^\varepsilon\mu\left(\s^{(\kk)}_{j_2}\right)-\mu\left(\s^{(\kk)}_{j_1}\right)\right) \right\} \nonumber \\
&=& \frac{1}{\left(e^\varepsilon+1\right)}\left(\mu\left(\s^{(\kk)}_{j_1}\right)+\mu\left(\s^{(\kk)}_{j_1}\right)\right) \nonumber \\
&=&\frac{1}{e^{\varepsilon}+1}\left\{\PP\left(\T\right)e^{\varepsilon}\log\frac{e^{\varepsilon}}{\PP\left({\T} ^c\right)+e^{\varepsilon}\PP\left({\T} \right)}+\PP\left({\T} ^c\right)\log\frac{1}{\PP\left({\T} ^c\right)+e^{\varepsilon}\PP\left({\T} \right)}\right\}  + \nonumber \\
&&\frac{1}{e^{\varepsilon}+1}\left\{\PP\left({\T}^c\right)e^{\varepsilon}\log\frac{e^{\varepsilon}}{\PP\left({\T} \right)+e^{\varepsilon}\PP\left({\T} ^c\right)}+\PP\left({\T} \right)\log\frac{1}{\PP\left({\T} \right)+e^{\varepsilon}\PP\left({\T} ^c\right)}\right\}.
\end{eqnarray}
We claim that
$\dd^*$ is a feasible dual variable for sufficiently small $\varepsilon$. In order to prove that $\dd^*$ is a feasible dual variable,
we show that $\left({\s^{(\kk)}}^{T}\dd^*\right)_j -\mu_j \geq 0 $ for all $j \in \{1, \ldots,2^\kk\}$ and all $\varepsilon \leq \varepsilon^*$,
where $\varepsilon^*$ is a positive quantity that
depends on $\PP$. Using the following facts
 \begin{eqnarray}
 e^\varepsilon &=& 1 + \varepsilon + \frac{1}{2}\varepsilon + O\left(\varepsilon^3\right) \nonumber \\
 \log\left(a+e^\varepsilon b\right)&=& b\varepsilon + \frac{b(1-b)}{2}\varepsilon^2 + O\left(\varepsilon^3\right) \nonumber \\
 \frac{1}{1+e^\varepsilon}&=& \frac{1}{2} -\frac{1}{4} \varepsilon + O\left(\varepsilon^2\right),
 \end{eqnarray}
 for small $\varepsilon$, we get that
  \begin{eqnarray}
 \mu_j &=& \PP\left(\T_j\right)e^{\varepsilon}\log\frac{e^{\varepsilon}}{\PP\left({\T_j} ^c\right)+e^{\varepsilon}\PP\left({\T_j} \right)}+\PP\left({\T_j} ^c\right)\log\frac{1}{\PP\left({\T_j} ^c\right)+e^{\varepsilon}\PP\left({\T_j} \right)} \nonumber \\
 &=& \PP\left(T_j\right)e^\varepsilon\varepsilon - \left(\PP\left(T_j\right)\left(e^\varepsilon-1\right)+1\right)\log\left(\PP\left(T_j\right)\left(e^\varepsilon-1\right)+1\right) \nonumber \\
 &=& \frac{1}{2}\PP\left(\T_j\right)\PP\left(\T^c_j\right)\varepsilon^2 + O\left(\varepsilon^3\right).
 \end{eqnarray}
 On the other hand,
  \begin{eqnarray}
 \left({\s^{(\kk)}}^{T}\dd^*\right)_j&=& {\s^{(\kk)}_j}^{T}\dd^* \nonumber \\
 &=&\frac{1}{\left(e^\varepsilon+1\right)\left(e^\varepsilon-1\right)}\left\{ \sum_{i \in \T}\frac{\s^{(\kk)}_{ij}}{|\T|}\left(e^\varepsilon\mu\left(\s^{(\kk)}_{j_1}\right)-\mu\left(\s^{(\kk)}_{j_2}\right)\right) \right. \nonumber \\
 && ~~~~~~~~~~~~~~~~~~~~~~~~~~~~~~
\left. +\sum_{i \in \T^c}\frac{\s^{(\kk)}_{ij}}{|\T^c|}\left(e^\varepsilon\mu\left(\s^{(\kk)}_{j_2}\right)-\mu\left(\s^{(\kk)}_{j_1}\right)\right) \right\} \nonumber \\
 &=&\frac{1}{\left(e^\varepsilon+1\right)\left(e^\varepsilon-1\right)}\left(e^\varepsilon\mu\left(\s^{(\kk)}_{j_1}\right)-\mu\left(\s^{(\kk)}_{j_2}\right)\right)\left(\frac{|\T_j\cap\T|}{|\T|}e^\varepsilon+\frac{|\T_j^c\cap\T|}{|\T|}\right) \nonumber \\
& & ~ + \frac{1}{\left(e^\varepsilon+1\right)\left(e^\varepsilon-1\right)}\left(e^\varepsilon\mu\left(\s^{(\kk)}_{j_2}\right)-\mu\left(\s^{(\kk)}_{j_1}\right)\right)\left(\frac{|\T_j\cap\T^c|}{|\T^c|}e^\varepsilon+\frac{|\T_j^c\cap\T^c|}{|\T^c|}\right) \nonumber \\
&=&\frac{1}{\left(e^\varepsilon+1\right)}\left(\frac{1}{2}\PP\left(\T\right)\PP\left(\T^c\right)\varepsilon^2 + O\left(\varepsilon^3\right)\right)\left\{\frac{|\T_j\cap\T^c|}{|\T^c|}+\frac{|\T_j^c\cap\T^c|}{|\T^c|} \right. \nonumber \\
&&  ~~~~~~~~~~~~~~~~~~~~~~~~~~~~~
\left. +\frac{|\T_j\cap\T|}{|\T|}+\frac{|\T_j^c\cap\T|}{|\T|} + O\left(\varepsilon\right) \right\} \nonumber \\
&=& \frac{1}{2}\PP\left(\T\right)\PP\left(\T^c\right)\varepsilon^2 + O\left(\varepsilon^3\right),
 \end{eqnarray}
where we have used the facts that $\T_{j_1}=\T$, $\T_{j_2}=\T^c$, and
\begin{eqnarray}
\mu\left(\s^{(\kk)}_{j_1}\right)&=&\frac{1}{2}\PP\left(\T\right)\PP\left(\T^c\right)\varepsilon^2
+ O\left(\varepsilon^3\right) \nonumber \\
\mu\left(\s^{(\kk)}_{j_2}\right)&=&\frac{1}{2}\PP\left(\T\right)\PP\left(\T^c\right)\varepsilon^2
+ O\left(\varepsilon^3\right).
\end{eqnarray}
Let $f(z)=|z-\frac12|$, $g(z)=-z\log z - (1-z)\log(1-z)$, and $h(z)= z(1-z)$ for
$0 \leq z \leq 1$. On the one hand, $g$ and $h$ are monotonically increasing
over $0 \leq z \leq \frac12$ and monotonically decreasing over $\frac12 \leq z \leq 1$ but
on the other hand, $f$ is monotonically decreasing over $0 \leq z \leq \frac12$ and
monotonically increasing over $\frac12 \leq z \leq 1$. Therefore,
\begin{eqnarray}
\T \in \underset{A\subseteq \cX}{\arg\min}\;\; \Big| \PP(A) - \frac12 \Big| &\Leftrightarrow& \T \in \underset{A\subseteq \cX}{\arg\max}\;\; -\PP(A)\log \PP(A) - \PP(A^c)\log \PP(A^c) \nonumber \\
& \Leftrightarrow & \T \in \underset{A\subseteq \cX}{\arg\max}\;\; \PP(A)\PP(A^c).
\end{eqnarray}
Since the set $\T$ was chosen so
that it maximizes $\PP\left(\T\right)\PP\left(\T^c\right)$,
we have that
$\PP\left(\T\right)\PP\left(\T^c\right)\geq\PP\left(\T_j\right)\PP\left(\T^c_j\right)$
for all $j \in \{1, \ldots,2^\kk\}$. Assume, to begin with, that $j \neq \{j_1,j_2\}$.
Then by the uniqueness of the maximizer assumption stated in the theorem, we
have that
$\PP\left(\T\right)\PP\left(\T^c\right)>\PP\left(\T_j\right)\PP\left(\T^c_j\right)$.
\begin{equation}
\left(\s^{T}\dd^*\right)_j -  \mu_j =
\frac{1}{2}\left(\PP\left(\T\right)\PP\left(\T^c\right)-\PP\left(\T_j\right)\PP\left(\T^c_j\right)\right)\varepsilon^2
+ O\left(\varepsilon^3\right),
\end{equation}
 and thus, there exists an $\varepsilon^*$ that depends on $\PP$ such that $\left({\s^{(\kk)}}^{T}\dd^*\right)_j -  \mu_j \geq 0$ for all $\varepsilon \leq \varepsilon^*$. If $j= \{j_1,j_2\}$, it is not hard
to check that $\left({\s^{(\kk)}}^{T}\dd^*\right)_j -  \mu_j = 0$ for all
$\varepsilon$. This establishes the satisfiability of $\dd^*$ for all
$\varepsilon \leq \varepsilon^*$ which proves an upper bound on the primal problem (given in \eqref{ub_mi_lp}). It remains to show that
the upper bound can be indeed achieved via the binary mechanism. To this extent, recall that the binary mechanism is given by
\begin{eqnarray}
\Q(0|x) \,=\, \left\{
\begin{array}{rl}
	\frac{e^\varepsilon}{1+e^\varepsilon}& \text{ if } \x \in \T\;,\\
	\frac{1}{1+e^\varepsilon}& \text{ if } \x\notin \T\;.\\
\end{array}
\right.\;\;\;
\Q(1|x) \,=\, \left\{
\begin{array}{rl}
	\frac{e^\varepsilon}{1+e^\varepsilon}& \text{ if } \x\notin \T \;,\\
	\frac{1}{1+e^\varepsilon}& \text{ if } \x\in \T \;. \\
\end{array}
\right.
\label{eq:defestbin2}
\end{eqnarray}
Computing the $I\left(\X;\Y\right)$ under \eqref{eq:defestbin2}, we get that
\begin{eqnarray}
I\left(\X;\Y\right)&=& \frac{1}{e^{\varepsilon}+1}\left\{\PP\left(\T\right)e^{\varepsilon}\log\frac{e^{\varepsilon}}{\PP\left({\T} ^c\right)+e^{\varepsilon}\PP\left({\T} \right)}+\PP\left({\T} ^c\right)\log\frac{1}{\PP\left({\T} ^c\right)+e^{\varepsilon}\PP\left({\T} \right)}\right\}  + \nonumber \\
&&\frac{1}{e^{\varepsilon}+1}\left\{\PP\left({\T}^c\right)e^{\varepsilon}\log\frac{e^{\varepsilon}}{\PP\left({\T} \right)+e^{\varepsilon}\PP\left({\T} ^c\right)}+\PP\left({\T} \right)\log\frac{1}{\PP\left({\T} \right)+e^{\varepsilon}\PP\left({\T} ^c\right)}\right\}.
\end{eqnarray}
Hence, the binary mechanism in \eqref{eq:defestbin2} achieves the
upper bound in \eqref{ub_mi_lp}. This proves the optimality of the binary mechanism for all $\varepsilon \leq \varepsilon^*$.

\subsection{Proof of Theorem \ref{thm:estappx}}
We start by proving an upper bound on $\max_{\Q \in
\mathcal{\D}_{\varepsilon}} I\left(\X;\Y\right)$ which is tight for
$\varepsilon \leq 1$. Recall that by Theorem \ref{thm:lp}, we have that
\begin{equation*}
\begin{aligned}
{\rm OPT}&= &\max_{\Q \in
\mathcal{\D}_{\varepsilon}} I\left(\X;\Y\right) & = &  \underset{\bb}{\text{maximize}} & & \sum_{j=1}^{2^\kk}\mu_j\bb_j \\
& & & & \text{subject to} & & \s^{(\kk)}\bb = \unity\\
& & & & & &  \bb \geq 0,
\end{aligned}
\end{equation*}
where
\begin{eqnarray}
\mu_j &=& \mu\left(\s^{(\kk)}_j\right) \nonumber \\
&=& \sum_{i \in
[\kk]}\pp\left(\x_i\right)\s^{(\kk)}_{ij}\log\left(\frac{\s^{(\kk)}_{ij}}{\sum_{i \in
[\kk]}\pp\left(\x_i\right)\s^{(\kk)}_{ij}}\right) \nonumber \\
&=& \PP\left(\T_j\right)e^\varepsilon\varepsilon - \left(\PP\left(\T_j\right)\left(e^\varepsilon-1\right)+1\right)\log\left(\PP\left(\T_j\right)\left(e^\varepsilon-1\right)+1\right),
\end{eqnarray}
$\T_j =\{i: \s^{(\kk)}_{ij}=e^\varepsilon\}$, and $\s^{(\kk)}$ is the $\kk
\times 2^\kk$ staircase pattern matrix given in Definition
\ref{def:staircase_matrix}.
\begin{lemma}
\label{mi_up_all}
For all distributions $\PP$ and all $\varepsilon$, the following bound holds
\begin{equation}
{\rm OPT}= \max_{\Q \in
\mathcal{\D}_{\varepsilon}} I\left(\X;\Y\right)  \leq \left(\max_{j} \mu_j \right) \frac{\kk}{e^\varepsilon+\kk-1}.
\end{equation}
\end{lemma}The proof of this lemma is given in Section \ref{sec:mi_up_all}. In what follows, we will
make the dependency of $\mu_j$ on $\PP\left(\T_j\right)$ and $\varepsilon$ explicit by writing
$\mu_j\left(\PP\left(\T_j\right),\varepsilon\right)$ for $\mu_j$. From the proof of Theorem \ref{thm:estbin}, we have that the partition set $\T$ defined in \eqref{eq:tbin} is given by $\T \in \arg\max_{A\subseteq \cX} \PP(A)\PP(A^c)$. It is easy to check that the binary mechanism given in \eqref{eq:defestbin} achieves the following utility
 \begin{equation}
 {\rm BIN}=\frac{\mu\left(\PP\left(\T\right),\varepsilon\right)+\mu\left(\PP\left(\T^c\right),\varepsilon\right)}{e^\varepsilon+1}.
 \end{equation}
\begin{lemma}
\label{mu_bound}
For all distributions $\PP$ and all $\varepsilon \leq 1$, the following bound holds
\begin{equation}
\frac{\max_{j} \mu_j}{\mu\left(\PP\left(\T\right),\varepsilon\right)+\mu\left(\PP\left(\T^c\right),\varepsilon\right)} \leq 1.
\end{equation}
\end{lemma}
The proof of the above lemma is given in Section \ref{sec:mu_bound}. Combining the results of lemmas \ref{mi_up_all} and \ref{mu_bound} we get that
\begin{eqnarray}
\frac{{\rm OPT}}{{\rm BIN}} &\leq& \frac{\max_{j} \mu_j}{\mu\left(\PP\left(\T\right),\varepsilon\right)+\mu\left(\PP\left(\T^c\right),\varepsilon\right)} \frac{\kk}{e^\varepsilon+\kk-1} \left(e^\varepsilon+1\right) \nonumber \\
&\leq& \frac{\kk}{e^\varepsilon+\kk-1} \left(e^\varepsilon+1\right) \nonumber \\
&\leq& e^\varepsilon+1,
\end{eqnarray}
for all $\varepsilon \leq 1$. This concludes the proof.

\subsection{Proof of Lemma \ref{mi_up_all}}
\label{sec:mi_up_all}
To begin with, since $\s^{(\kk)}_1=\unity=\frac{1}{e^\varepsilon}\s^{(\kk)}_{2^\kk}$ and $\mu$ is homogenous, we have that $\bb_1\mu_1+\bb_{2^\kk}\mu_{2^\kk}=\left(\frac{1}{e^\varepsilon}\bb_1+\bb_{2^\kk}\right)\mu_{2^\kk}$. Therefore, the following two maximization problems are equivalent
\begin{equation}
\begin{aligned}
&  \underset{\bb}{\text{maximize}} & & \sum_{j=1}^{2^\kk}\mu_j\bb_j \\
& \text{subject to} & & \s^{(\kk)}\bb = \unity\\
& & &  \bb \geq 0
\end{aligned}
=\begin{aligned}
&  \underset{\bb}{\text{maximize}} & & \sum_{j=1}^{2^\kk-1}\tilde{\mu}_j\bb_j \\
& \text{subject to} & & \tilde{\s}^{(\kk)}\bb = \unity\\
& & &  \bb \geq 0,
\end{aligned}
\end{equation}
where $\tilde{\mu}_j=\mu_{j+1}$ and $\tilde{\s}^{(\kk)}$ is obtained by deleting the first column of $\s^{(\kk)}$.
Moreover, using the fact that $\max_{j \in [2^\kk-1]} \tilde{\mu}_j \leq \max_{j \in [2^\kk]} \mu_j$ and weak duality, we get that
\begin{eqnarray}
\begin{aligned}
&  \underset{\bb}{\text{maximize}} & & \tilde{\mu}^T\bb \\
& \text{subject to} & & \tilde{\s}^{(\kk)}\bb = \unity\\
& & &  \bb \geq 0
\end{aligned}
& &
 \begin{aligned}
\leq \; \; \left(\max_{j \in [2^\kk-1]} \tilde{\mu}_j \right) &  \underset{\bb}{\text{maximize}} & &  \unity^{T}\bb  \\
& \text{subject to} & & \tilde{\s}^{(\kk)}\bb = \unity\\
& & &  \bb \geq 0
\end{aligned} \nonumber \\
&&
 \begin{aligned}
\leq \; \; \left(\max_{j \in [2^\kk]} \mu_j \right)  & \underset{\dd}{\text{minimize}} & & \unity^{T}\dd \\
& \text{subject to} & & \tilde{\s}^{(\kk)^{T}} \dd \geq \unity. \\
& & &
\end{aligned}
\end{eqnarray}
Consider the following choice of dual variable $\dd^*_i= \frac{1}{e^\varepsilon+\kk-1}$. We claim that $\dd^*$ is satisfiable. This can be easily verified by noting that
\begin{equation}
\left(\tilde{\s}^{(\kk)^{T}}\dd^*\right)_j = \tilde{\s}_j^{(\kk)^{T}}\dd^* = \frac{|\T_j|e^\varepsilon + (\kk-|\T_j|)}{e^\varepsilon+\kk-1} = \frac{|\T_j|(e^\varepsilon-1) + \kk}{e^\varepsilon+\kk-1} \geq 1
\end{equation}
where the last inequality holds since $|\T_j| \geq 1$ (this is true because we have deleted the first column of $\s^{(\kk)}$). Therefore, ${\rm OPT} \leq   \left(\max_{j} \mu_j \right)\unity^{T}\dd^*= \left(\max_{j} \mu_j \right)\frac{\kk}{e^\varepsilon+\kk-1}$ which was to be shown.

\subsection{Proof of Lemma \ref{mu_bound}}
\label{sec:mu_bound}
Let $\mu\left(z,\varepsilon\right)$ be the function obtained by replacing $\PP\left(\T_j\right)$ by the continuous variable $z \in [0,1]$ in $\mu_j\left(\PP\left(\T_j\right),\varepsilon\right)$. Taking the derivative of $\mu\left(z,\varepsilon\right)$ with respect to $z$ we get
\begin{equation}
\mu'\left(z,\varepsilon\right)= e^\varepsilon\varepsilon - (e^\varepsilon-1) - (e^\varepsilon-1)\log\left(z(e^\varepsilon-1)+1\right).
\end{equation}
Observe that $\mu'\left(z,\varepsilon\right)>0$ for all
\begin{equation}
z < z^*(\varepsilon) = \frac{1}{e^\varepsilon-1}\left(e^{\left\{\frac{e^\varepsilon\varepsilon}{e^\varepsilon-1} -1\right\}}-1\right),
\end{equation}
$\mu'\left(z,\varepsilon\right)<0$ for all $z > z^*(\varepsilon)$, and $\mu'\left(z,\varepsilon\right)=0$ for $z=z^*(\varepsilon)$. Combining this with the fact that $\mu\left(0,\varepsilon\right)=\mu\left(1,\varepsilon\right)=0$ we get that $\mu\left(z,\varepsilon\right) \geq 0$ for all $z \in [0,1]$ and for any fixed $\varepsilon$, $\mu\left(z,\varepsilon\right)$ is maximized at $z^*(\varepsilon)$.

Set $\x^* \in \arg\max_{x \in \cX} \PP\left(\x\right)$ and fix an $\varepsilon \leq 1$. We will treat the following three cases separately.\\

\noindent
{\bf Case 1:} $\PP(\x^*) \in [1-z^*(\varepsilon),1]$.
\begin{claim}
\label{claim:max_cm} Let $\T=\{x^*\}$. Then $\{\T,\T^c\} = \arg \max_{A
\subseteq \cX} \PP(A)\PP(A^c)$ and $\max_{A \subseteq
\cX}\mu(\PP(A),\varepsilon)=\max\left(\mu(\PP(T),\varepsilon),\mu(\PP(T^c),\varepsilon)\right)$.
\end{claim}
\begin{proof}
Observe that $z^*(\varepsilon) \leq \frac12$ for all $\varepsilon$ and $\T^c=\cX \setminus \{\x^*\}$. The function $f(z)=z(1-z)$ decreases over the range $[\frac12,1] \supseteq [1-z^*(\varepsilon),1]$. Thus, for all  $A \supset \T$, $\PP(\T)\PP(\T^c) > \PP(A)\PP(A^c)$ because $\PP(\T) \geq 1-z^*(\varepsilon)$. This proves that $\T \in \arg \max_{A \subseteq \cX} \PP(A)\PP(A^c)$ and for all $A \supset \T$, $A \notin \arg \max_{A \subseteq \cX} \PP(A)\PP(A^c)$. Using a similar approach, we can show that $\T^c \in \arg \max_{A \subseteq \cX} \PP(A)\PP(A^c)$ and for all $A \subset \T^c$, $A \notin \arg \max_{A \subseteq \cX} \PP(A)\PP(A^c)$. Therefore, $\{\T,\T^c\} = \arg \max_{A \subseteq \cX} \PP(A)\PP(A^c)$. This proves the first part of the claim. The function $\mu\left(z,\varepsilon\right)$ increases over the range $[0,z^*(\varepsilon)]$. Thus, for all $A \subseteq \T^c$, $\mu(\PP(A),\varepsilon) \leq \mu(\PP(\T^c),\varepsilon)$ because $\PP(\T^c) \leq z^*(\varepsilon)$. On the other hand, note that $\mu\left(z,\varepsilon\right)$ decreases over the range $[z^*(\varepsilon),1]$ which includes the range $[1-z^*(\varepsilon),1]$. Thus, for all $A$ such that $A \supseteq \T$, $\mu(\PP(A),\varepsilon) \leq \mu(\PP(\T),\varepsilon)$ because $\PP(\T) \geq 1- z^*(\varepsilon)$. This proves that $\max\left(\mu(\PP(T),\varepsilon),\mu(\PP(T^c),\varepsilon)\right) =  \max_{A \subseteq \cX}\mu(\PP(A),\varepsilon)$.
\end{proof}
Using the above claim, we can conclude that the partition set $\T$ defined in \eqref{eq:tbin} is equal to $\{\x^*\}$ and
\begin{eqnarray}
\frac{\max_{j} \mu_j}{\mu\left(\PP\left(\T\right),\varepsilon\right)+\mu\left(\PP\left(\T^c\right),\varepsilon\right)} &=& \frac{\max_{A \subseteq \cX} \mu(\PP(A),\varepsilon)}{\mu\left(\PP\left(\T\right),\varepsilon\right)+\mu\left(\PP\left(\T^c\right),\varepsilon\right)}\nonumber \\
&\leq& \frac{\max_{A \subseteq \cX} \mu(\PP(A),\varepsilon)}{\max\left(\mu(\PP(T),\varepsilon),\mu(\PP(T^c),\varepsilon)\right)} \nonumber \\
&=& 1.
\end{eqnarray}
\noindent {\bf Case 2:} $\PP(\x^*)  \in [\frac12,1-z^*(\varepsilon)]$. Using
the first part of the proof of Claim \ref{claim:max_cm}, we can show that if
$\T=\{x^*\}$, then $\{\T,\T^c\} = \arg \max_{A \subseteq \cX}
\PP(A)\PP(A^c)$. Therefore, the partition set $\T$ defined in \eqref{eq:tbin}
is equal to $\{\x^*\}$ and
\begin{eqnarray}
\frac{\max_{j} \mu_j}{\mu\left(\PP\left(\T\right),\varepsilon\right)+\mu\left(\PP\left(\T^c\right),\varepsilon\right)} &=& \frac{\max_{A \subseteq \cX} \mu(\PP(A),\varepsilon)}{\mu\left(\PP\left(\T\right),\varepsilon\right)+\mu\left(\PP\left(\T^c\right),\varepsilon\right)}\nonumber \\
&\leq&  \frac{\mu(z^*(\varepsilon),\varepsilon)}{\mu\left(\PP\left(\x^*\right),\varepsilon\right)+\mu\left(1-\PP\left(\x^*\right),\varepsilon\right)} \nonumber \\
&\leq& 1.
\end{eqnarray}

\noindent {\bf Case 3:} $\PP(\x^*)  \in [0,\frac12]$.
\begin{claim}
\label{claim:lb_ub} There exists a set $A \subset \cX$ such that $\frac12 -
\PP(\x^*) \leq \PP(A) \leq \frac12$.
\end{claim}
\begin{proof}
Without loss of generality, assume that the sequence $\PP(\x_i)$, $i \in
[\kk]$, is sorted in increasing order. Let $l^* = \min\{l:
\sum_{i=1}^{l}\PP(\x_i) \geq \frac12\}$. From the definition of $l^*$,
$\PP(\{\x_1,\ldots,\x_{l^*-1}\}) < \frac12$ and
$\PP(\{\x_1,\ldots,\x_{l^*}\}) \geq \frac12$. Further,
\begin{equation*}
\PP(\{\x_1,\ldots,\x_{l^*-1}\})=\PP(\{\x_1,\ldots,\x_{l^*}\})-
\PP(\x_{l^*})
\end{equation*}
and since $\x^* \in \arg\max_{x \in \cX} \PP\left(\x\right)$, $\PP(\x_{l^*}) \leq
\PP(\x^*)$. Therefore, if $A=\{\x_1,\ldots,\x_{l^*-1}\}$, then $\frac12 -
\PP(\x^*) \leq \PP(A) \leq \frac12$.
\end{proof}

Let $\PP(\T)=\min\{\PP(B): B \in \arg \max_{A \subseteq
\cX}\PP(A)\PP(A^c)\}$. We claim that $\frac14 \leq \PP(\T) \leq \frac12$. The
upper bound on $\PP(\T)$ follows immediately from its definition. To prove
the lower bound on $\PP(\T)$, consider the set $A$ given in Claim
\ref{claim:lb_ub} and observe that
\begin{eqnarray}
\PP(\T)&\geq& \max(\PP(\x^*),\PP(A)) \nonumber \\
&\geq& \max(\PP(\x^*),\frac12-\PP(\x^*)) \nonumber \\
&\geq& \frac14.
\end{eqnarray}
All the inequalities follow from Claim \ref{claim:lb_ub} and the fact
that $\PP(\x^*) \in [0,\frac12]$.

Since $\frac14 \leq \PP(\T) \leq \frac12$, we have that
$\frac12 \leq \PP(\T^c) \leq \frac34$. Moreover, the function
$\mu\left(z,\varepsilon\right)$ decreases over the range
$[z^*(\varepsilon),1] \supset [\frac12,\frac34]$ and increases over the range
$[\frac14, z^*(\varepsilon)]$. Therefore, $\mu\left(\PP(\T^c),\varepsilon\right)
\geq \mu\left(\frac34,\varepsilon\right)$ and
$\mu\left(\PP(\T),\varepsilon\right) \geq \min\left(
\mu\left(\frac12,\varepsilon\right),\mu\left(\frac14,\varepsilon\right)\right)$.
Putting it all together, we have that
\begin{eqnarray}
\frac{\max_{j} \mu_j}{\mu\left(\PP\left(\T\right),\varepsilon\right)+\mu\left(\PP\left(\T^c\right),\varepsilon\right)} &=& \frac{\max_{A \subseteq \cX} \mu(\PP(A),\varepsilon)}{\mu\left(\PP\left(\T\right),\varepsilon\right)+\mu\left(\PP\left(\T^c\right),\varepsilon\right)}\nonumber \\
&\leq&  \frac{\mu(z^*(\varepsilon),\varepsilon)}{\min\left(
\mu\left(\frac12,\varepsilon\right),\mu\left(\frac14,\varepsilon\right)\right)+\mu\left(\frac34,\varepsilon\right)} \nonumber \\
&\leq& 1.
\end{eqnarray}

\subsection{Proof of Theorem \ref{thm:estrr}}
By Theorem \ref{thm:lp}, we have that
\begin{equation}
\label{eq:mutual_info_lp2}
\begin{aligned}
\max_{\Q \in
\mathcal{\D}_{\varepsilon}} I\left(\X;\Y\right) & = &  \underset{\bb}{\text{maximize}} & & \mu^{T}\bb \\
& & \text{subject to} & & \s^{(\kk)}\bb = \unity\\
& & & &  \bb \geq 0,
\end{aligned}
\end{equation}
where $ \mu_j = \mu\left(\s^{(\kk)}_j\right)= \sum_{i \in [\kk]}\pp\left(\x_i\right)\s^{(\kk)}_{ij}\log\left(\frac{\s^{(\kk)}_{ij}}{\sum_{i \in [\kk]}\pp\left(\x_i\right)\s^{(\kk)}_{ij}}\right)$ for $j \in \{1, \ldots,2^\kk\}$ and $\s^{(\kk)}$ is the $\kk \times 2^\kk$ staircase pattern matrix given in
Definition \ref{def:staircase_matrix}. The polytope given by $\s^{(\kk)}\bb = \unity$ and $\bb \geq 0$ is a closed and bounded one. Thus,
there is no duality gap and solving the above linear program is
equivalent to solving its dual
\begin{equation}
\label{eq:dual_mutual_information_lp2}
\begin{aligned}
& \underset{\dd}{\text{minimize}} & & \unity^{T}\dd \\
& \text{subject to} & & {\s^{(\kk)}}^{T}\dd \geq \mu.
\end{aligned}
\end{equation}
Note that any satisfiable solution $\dd^*$ to
(\ref{eq:dual_mutual_information_lp2}) provides an upper bound to
(\ref{eq:mutual_info_lp2}) since $\max \mu^{T}\bb = \min \unity^{T}\dd \leq
\unity^{T}\dd^*$. Let $\T_j =\{\x_i: \s^{(\kk)}_{ij}=e^\varepsilon\}$ and set $j_i =
\{j: \T_j =i\}$ for $i \in \{1, \ldots,\kk\}$. Consider the following choice
of dual variable
\begin{equation}
\dd^*_i = \frac{1}{\left(e^\varepsilon -1\right)\left(e^\varepsilon+\kk-1\right)}\left\{\left(e^\varepsilon+\kk-2\right)\mu\left(\s^{(\kk)}_{j_i}\right)-\sum_{l \in [\kk], l \neq i}\mu\left(\s^{(\kk)}_{j_l}\right) \right\},
\end{equation}
for $i \in \{1, \ldots,\kk\}$. Observe that since
$\T_{j_i} =i$ we have that $\PP\left(\T_{j_i}\right)=\pp\left(\x_i\right)$
and since
 \begin{eqnarray}
 \mu_j &=&  \PP\left(\T_j\right)e^{\varepsilon}\log\frac{e^{\varepsilon}}{\PP\left({\T_j} ^c\right)+e^{\varepsilon}\PP\left({\T_j} \right)}+\PP\left({\T_j} ^c\right)\log\frac{1}{\PP\left({\T_j} ^c\right)+e^{\varepsilon}\PP\left({\T_j} \right)},
 \end{eqnarray}
we have that
\begin{eqnarray}
\label{ub_mi_hp}
\unity^{T}\dd^*&=& \frac{1}{\left(e^\varepsilon -1\right)\left(e^\varepsilon+\kk-1\right)} \sum_{i \in [\kk]}\left\{\left(e^\varepsilon+\kk-2\right)\mu\left(\s^{(\kk)}_{j_i}\right)-\sum_{l \in [\kk], l \neq i}\mu\left(\s^{(\kk)}_{j_l}\right) \right\} \nonumber \\
&=&\frac{1}{\left(e^\varepsilon -1\right)\left(e^\varepsilon+\kk-1\right)}\left\{\left(e^\varepsilon+\kk-2\right)\sum_{i \in [\kk]}\mu\left(\s^{(\kk)}_{j_i}\right) - \sum_{i \in [\kk]}\sum_{l \in [\kk], l \neq i}\mu\left(\s^{(\kk)}_{j_l}\right) \right\} \nonumber \\
&=&\frac{1}{\left(e^\varepsilon -1\right)\left(e^\varepsilon+\kk-1\right)}\left\{\left(e^\varepsilon+\kk-2\right)\sum_{i \in [\kk]}\mu\left(\s^{(\kk)}_{j_i}\right) - (\kk-1)\sum_{i \in [\kk]}\mu\left(\s^{(\kk)}_{j_i}\right) \right\} \nonumber \\
&=&\frac{1}{\left(e^\varepsilon+\kk-1\right)}\sum_{i \in [\kk]}\mu\left(\s^{(\kk)}_{j_i}\right) \nonumber \\
&=&\frac{1}{\left(e^\varepsilon+\kk-1\right)}\sum_{i \in [\kk]}\left\{\pp\left(\x_i\right)e^{\varepsilon}\log\frac{e^{\varepsilon}}{\pp\left(\x_i\right)\left(e^{\varepsilon}-1\right)+1} \right. \nonumber \\
&& ~~~~~~~~~~~~~~~~~~~~~~~~~~~~~~~~~~~~~~~~  \left. +\left(1-\pp\left(\x_i\right)\right)\log\frac{1}{\pp\left(x_i\right)\left(e^{\varepsilon}-1\right)+1}\right\}.
\end{eqnarray}
We claim that $\dd^*$ is a feasible dual variable for sufficiently large
$\varepsilon$. In order to prove that $\dd^*$ is a feasible dual variable, we
show that $\left({\s^{(\kk)}}^{T}\dd^*\right)_j -\mu_j \geq 0 $ for all  $j \in \{1,
\ldots,2^\kk\}$ and all $\varepsilon \geq \varepsilon^*$, where
$\varepsilon^*$ is a positive quantity that depends on $\PP$. Using the fact
that
\begin{equation}
\log\left(a+e^\varepsilon b\right) = \varepsilon + \log b + O\left(e^{-\varepsilon}\right),
\end{equation}
for large $\varepsilon$, we get that
  \begin{eqnarray}
 \mu_j &=& \PP\left(\T_j\right)e^{\varepsilon}\log\frac{e^{\varepsilon}}{\PP\left({\T_j} ^c\right)+e^{\varepsilon}\PP\left({\T_j} \right)}+\PP\left({\T_j} ^c\right)\log\frac{1}{\PP\left({\T_j} ^c\right)+e^{\varepsilon}\PP\left({\T_j} \right)} \nonumber \\
 &=& \PP\left(T_j\right)e^\varepsilon\varepsilon - \left(\PP\left(T_j\right)\left(e^\varepsilon-1\right)+1\right)\log\left(\PP\left(T_j\right)\left(e^\varepsilon-1\right)+1\right) \nonumber \\
 &=& \PP\left(T_j\right)e^\varepsilon\varepsilon - \left(\PP\left(T_j\right)\left(e^\varepsilon-1\right)+1\right)\left(\varepsilon + \log \PP\left(\T_j\right) + O\left(e^{-\varepsilon}\right) \right) \nonumber \\
 &=& -\left(\PP\left(T_j\right)\log \PP\left(T_j\right)\right) e^\varepsilon + O\left(\varepsilon\right).
 \end{eqnarray}
On the other hand,
  \begin{eqnarray}
 \left({\s^{(\kk)}}^{T}\dd^*\right)_j&=& {\s^{(\kk)}_j}^{T}\dd^*  \nonumber \\
 &=& e^\varepsilon \sum_{i \in \T_j}\dd^*_i + \sum_{i \in \T^c_j}\dd^*_i \nonumber \\
  &=&  - \frac{1}{\left(e^\varepsilon -1\right)\left(e^\varepsilon+\kk-1\right)}\left\{\sum_{i \in [\kk]}\s^{(\kk)}_{ij}\left(e^\varepsilon+\kk-2\right)\left(\pp\left(\x_i\right)\log\pp\left(\x_i\right)e^\varepsilon \right. \right. \nonumber \\
& & ~~~~~~~~~~~~~~~~~~~~~~~~~~~~~~~~~~~~~~~~~~~~~~~~~~~~~~~~~~~~~~~~~~~~~~~~~~~~~~~ \left. + O\left(\varepsilon\right)\right)  \Bigg\} \nonumber \\
 & &+ \frac{1}{\left(e^\varepsilon -1\right)\left(e^\varepsilon+\kk-1\right)}\left\{\sum_{i \in [\kk]}\sum_{l \in [\kk], l \neq i}\s^{(\kk)}_{ij}\left( \left(\pp\left(\x_l\right)\log\pp\left(\x_l\right)\right)e^\varepsilon + O\left(\varepsilon\right) \right)\right\} \nonumber \\
 &=&- \frac{1}{\left(e^\varepsilon -1\right)\left(e^\varepsilon+\kk-1\right)} \left(\left(\sum_{i \in \T_j}\pp\left(\x_i\right)\log \pp\left(\x_i\right)\right)e^{3\varepsilon} + O\left(e^{2\varepsilon}\varepsilon\right)\right) \nonumber \\
 &=& -\left(\sum_{i \in \T_j}\pp\left(\x_i\right)\log \pp\left(\x_i\right)\right)e^{\varepsilon} + O\left(\varepsilon\right).
 \end{eqnarray}
 Assume, to begin with, that $j \neq \{j_1,j_2,...,j_{\kk}\}$. Then
 \begin{equation}
\left({\s^{(\kk)}}^{T}\dd^*\right)_j -  \mu_j = \left(\PP\left(T_j\right)\log \PP\left(T_j\right)-\sum_{i \in \T_j}\pp\left(\x_i\right)\log \pp\left(\x_i\right)\right)e^{\varepsilon} + O\left(\varepsilon\right).
\end{equation}
Notice that for $j \neq \{j_1,j_2,...,j_{\kk}\}$, $\PP\left(T_j\right)\log
\PP\left(T_j\right)>\sum_{i \in \T_j}\pp\left(\x_i\right)\log
\pp\left(\x_i\right)$. Therefore, there exists an $\varepsilon^* >0$ such
that $\left({\s^{(\kk)}}^{T}\dd^*\right)_j -  \mu_j \geq 0$ for all $\varepsilon \geq
\varepsilon^*$. If $j \in \{j_1,j_2,...,j_{\kk}\}$, it is not hard to check
that $\left({\s^{(\kk)}}^{T}\dd^*\right)_j -  \mu_j = 0$ for all $\varepsilon$. This
establishes the satisfiability of $\dd^*$ for all $\varepsilon \geq
\varepsilon^*$. It remains to show that the upper bound can be indeed
achieved via the randomized response mechanism. To this extent, recall that
the randomized response is given by
\begin{eqnarray}
	Q(y|x) \,=\, \left\{
\begin{array}{rl}
	\dfrac{e^\varepsilon}{|\cX|-1+e^\varepsilon}& \text{ if } y=x\;,\\
	\dfrac{1}{|\cX|-1+e^\varepsilon}& \text{ if } y\neq x\;.\\
\end{array}
\right.
	\label{eq:defrr3}
\end{eqnarray}
Computing the $I\left(\X;\Y\right)$ under \eqref{eq:defrr3}, we get that
\begin{eqnarray}
I\left(\X;\Y\right)& = & \frac{1}{e^\varepsilon + \kk -1}\sum_{i \in [\kk]}\left\{\pp\left(\x_i\right)e^{\varepsilon}\log\frac{e^{\varepsilon}}{\pp\left(\x_i\right)\left(e^{\varepsilon}-1\right)+1} \right. \nonumber \\
& & ~~~~~~~~~~~~~ \left. + \left(1-\pp\left(\x_i\right)\right)\log\frac{1}{\pp\left(x_i\right)\left(e^{\varepsilon}-1\right)+1}\right\}.
\end{eqnarray}
Hence, the randomized response mechanism achieves the upper bound
\eqref{ub_mi_hp}. This proves the optimality of the randomized response for
all $\varepsilon \geq \varepsilon^*$.
\section{Proof of Proposition \ref{dp_prop}}
\label{sec:dp_prop}
 Let $U\left(\Q\right)$ be a utility mechanism of the form
$U\left(\Q\right)=\sum_{\cY} \mu(\Q_\y)$, where $\mu$ is a sublinear
function. Consider a stochastic mapping $W$ of dimensions $\n \times m$ and
let $\Q W$ be the stochastic mapping obtained by first applying $\Q$ to $\X
\in \cX$ to obtain $\Y \in \cY$ and then applying $W$ to $\Y$ to obtain $Z
\in \mathcal{Z}$.
\begin{eqnarray}
U\left(\Q W\right) &=&\sum_{\mathcal{Z}} \mu\left(\left(\Q W\right)_z\right) \nonumber \\
&=&\sum_{\mathcal{Z}} \mu\left(\sum_{\cY}\Q_{\y} W_{\y,z}\right) \nonumber \\
&\leq& \sum_{\cY,\mathcal{Z}} W_{\y,z}\mu\left(\Q_{\y} \right) \nonumber \\
&=&\sum_{\cY} \mu(\Q_\y) \nonumber \\
&=& U\left(\Q\right),
\end{eqnarray}
where the inequality follows from sublinearity and the second to last
equality follows from the row stochastic property of $W$. Therefore,
$U\left(\Q\right)$ obeys the data processing inequality.

\acks{This research is supported in part by NSF CISE award CCF-1422278, NSF SaTC award
CNS-1527754, NSF CMMI award MES-1450848 and NSF ENG award ECCS-1232257.}


\bibliographystyle{plain}
\bibliography{privacy}

\end{document}